\documentclass[journal,onecolumn,12pt]{IEEEtran}

\usepackage{bbm}

\usepackage{graphicx}
\usepackage{epstopdf}
\usepackage{amssymb}
\usepackage{epsfig}
\usepackage{amsmath}
\usepackage{amsthm}
\usepackage{amsfonts}
\usepackage{setspace}
\usepackage{url}
\usepackage{ifthen}
\usepackage[active]{srcltx}
\usepackage{epstopdf}

\usepackage[english]{babel}

\usepackage[latin1]{inputenc}

\usepackage{times}
\usepackage[T1]{fontenc}

\usepackage{amssymb}
\usepackage{float}
\usepackage{subfigure}
\usepackage{ifpdf}
\usepackage{color}
\usepackage{comment}

\newtheorem{theorem}{Theorem}
\newtheorem{cor}{Corollary}
\newtheorem{lem}{Lemma}

\newtheorem{defn}{Definition}
\newtheorem{rem}{Remark}
\newtheorem{example}{Example}

\newcommand{\udl}{\underline}

\newcommand{\Al}{\alpha}
\newcommand{\eps}{\epsilon}

\newcommand{\ph}{P_{H}\left(b\right)}
\newcommand{\Gbin}{G_{bin}\pa{b}}
\newcommand{\Gon}{G_{remote}\pa{B}}
\newcommand{\Gof}{G_{insider}\pa{B}}
\newcommand{\qp}{\theta}

\newcommand{\pa}[1]{\left( #1 \right)}
\newcommand{\pac}[1]{\left\{ {#1} \right\}}
\newcommand{\paq}[1]{\left[ {#1} \right]}
\newcommand{\beq}[2]{\begin{equation}\label{#1} #2 \end{equation}}
\newcommand{\bal}[2]{{\setlength\arraycolsep{2pt}\begin{eqnarray}\label{#1} #2 \end{eqnarray}}}
\newcommand{\beqn}[2]{\begin{equation*}\label{#1} {#2} \end{equation*}}

\newcommand{\iid}{i.i.d.}

\newcommand{\nn}{\nonumber}

\newcommand{\barr}[2]{\left\{\begin{array}{#1}#2 \end{array}\right.}

\newcommand{\bitm}[1]{\begin{itemize} #1 \end{itemize}}

\newcommand{\hdps}{H\pa{s}+D\pa{s||p}}

\begin{document}
%
\title{Password Cracking: The Effect of Hash Function Bias on the Average Guesswork}
\author{Yair Yona, and Suhas Diggavi,~\IEEEmembership{Fellow,~IEEE}
\thanks{The material in
this paper was presented in part at the International Symposium on Information Theory (ISIT), 2016.}
\thanks{The work of the authors was partially supported by the NSF awards 
1740047 and by the UC-NL grant LFR-18-548554.}}
\maketitle

\begin{abstract}
Modern authentication systems store hashed values of passwords of users using cryptographic hash functions. Therefore, to crack a password an attacker
needs to guess a hash function input that is mapped to the hashed value, as opposed to the password itself. We call a hash function that maps the same number of inputs to each bin, as \textbf{unbiased}. However, cryptographic hash functions in use 
have not been proven to be unbiased (i.e., they may have an unequal number of inputs mapped to different bins).  A cryptographic hash function has the property that it is computationally difficult to find an input mapped
to a bin. In this work we introduce a structured
notion of biased hash functions for which we analyze the average guesswork under certain types of brute force attacks.

This work shows that in the presence of bias, the level of security depends on the set of bins to which passwords are hashed as well as the statistical profile of a hash function. We examine the average guesswork conditioned on the set of hashed values of passwords, and model the statistical profile through the empirical 
distribution of the number of inputs that are mapped to a bin. In particular, we focus on a class of statistical profiles (capturing the bias), which we call type-class statistical profiles,  that has an empirical distribution related to the probability of the type classes defined in the method of types.  For such profiles, we show that the average guesswork is related to basic measures in information theory such as entropy and divergence.
We use this to show that the effect of bias on the conditional average guesswork is limited compared to other system parameters such as the number of valid users who store their hashed passwords in the system. 

Finally, we show that bias can be used to increase the average guesswork, when using a backdoor mechanism that allows a system to efficiently modify the mappings of the hash function.
\end{abstract}

\section{Introduction}
A password is a means to authenticate users by allowing access only to an authorized user who knows it. Modern systems do not store the password of a valid user in plain text but rather hash it (a many
to one mapping) and store its hashed value, which is paired with a user name. This in turn protects the passwords of valid users when the system is compromised \cite{MarechalAdvancesinPasswordCracking},\cite{BurrNISTElectronicAuthenticationGuideline}. In order to gain access to a system a user provides the system with his user name and his password; the system then calculates its hash value and compares it to the one it has on file.

The most prevalent method of protecting passwords  is using cryptographic hash functions whose output can be computed easily but cannot be inverted easily \cite{ManberSimpleSchemeHardertoCracl}, \cite{PasswordCrackingUsingProbContFreeGrammars}, \cite{KelleyGuessAgainAndAgainAndAgain}, \cite{MarechalAdvancesinPasswordCracking}, \cite{BurrNISTElectronicAuthenticationGuideline}. Essentially, cryptographic hash functions are the ``workhorses of modern cryptography''
\cite{SchneierMd5Sha1Crypto} and, among other things, enable systems to protect passwords against attackers that break into their servers. In some cases such as message authentication codes (MAC) \cite{CanettiHMAC1996} it is meaningful to consider keyed hash functions. Keyed hash functions can be viewed as a set of hash functions, where the actual function which is used is determined by the value assigned to the key, which is a secret \cite{Wegman_CarterNewHashFunc1981}. A practical keyed hash function that enables the use of off-the-shelf computationally secure hash functions, is a keyed-hash message authentication code (HMAC) \cite{CanettiHMAC1996}.

Modern systems store for every user a user name along with a hash value of its password; in this paper the set of hash values of the passwords of valid users is referred to as \textbf{the set of occupied bins}, which is denoted by $B$. A common assumption that is usually made is that the hash functions are \textbf{unbiased}, that is, every possible output of the hash function has the same number of inputs mapped to it. However, this assumption has not been proved for cryptographic hash functions used in practice, and they can be biased. In this work we define \textbf{bias} based on the \textbf{statistical profile of a hash function}, which is basically the empirical distribution induced by the number of inputs that are mapped to every bin; in this sense bias means that the number of inputs that are mapped to every bin changes across bins. In terms of the set of occupied bins it means that some hash values of passwords are more vulnerable than others (unlike the unbiased case where every element in the set of occupied bins is equally protected). Indeed, this work addresses the question of how bias affects the level of security in terms of password cracking.

A basic type of attacks for cracking passwords are brute force attacks, where an attacker who does not know the structure of the hash function guesses inputs one by one until it finds \textbf{an input} that is mapped to an element in the set of occupied bins. We consider two brute force attacks: A remote attack where an attacker tries to break into an account of a user by guessing an input that has the same hash value as the hash value of the password of that user; and an insider attack, where an attacker has an insider access to all user-names and hash values of their passwords.  This enables it to compare the hashed values of his guesses against every element in the set of occupied bins, that is, it should guess an input whose hash value is equal to any of the elements in the set of occupied bins, to satisfy its goal of finding one valid user-name password pair. Note that in both cases, the attacker does not know which inputs (passwords) map to which bin, and can only check by applying the (unknown) hash function.

For brute force attacks, a meaningful measure of the level of security is guesswork. Guesswork is the number of attempts required to successfully guess a secret. The guesswork is a random variable whose probability mass function depends on the statistical profile according to which a secret is chosen (e.g., when there is bias in terms of the way the secret is drawn at random) along with the strategy used for guessing. It was first introduced and analyzed by Massey \cite{Messy1994} who lower bounded the average guesswork by the Shannon entropy. Arikan showed that the rate at which any moment of the guesswork increases is actually equal to the R\'{e}nyi entropy \cite{Arikan_Ineq_Guessing}, and is larger than the Shannon entropy unless the secret is uniformly distributed in which case both are equal. Guesswork has been analyzed in many other scenarios such as guessing up to a certain level of distortion \cite{MerhavArikanDist}, guessing under source uncertainty with and without side information \cite{Sunderesan07}, \cite{SunderesanGuessSourceUncertaintySideInfo}, using guesswork to lower bound the complexity of sequential decoding \cite{Arikan_Ineq_Guessing}, guesswork for Markov chains \cite{MaloneGuessworkandEntropy}, guesswork under a certain probability of failure of the attack \cite{YonaForensicsPUFS}, guesswork for the Shannon cipher system \cite{MerhavArikanShannonCipher} as well as other applications in security \cite{ShayevitzBurinReducingGuesswork}, \cite{Duffy15}, \cite{SassonVerduGuesswork}, \cite{ShayevitzGuessworkUnreliableOracle}, \cite{ShayevitzGuessingBooleanFUnctions}, \cite{YonaDiggaviISITGuesswork}, \cite{BeiramiMedardGeometricPerspectiveGuesswork}, \cite{MerhavCohenUniversalRandomizedGuessing}. As far as we know there has been no systematic study of guesswork for hash functions.

A hash function is non-injective. Therefore an attacker only has to guess \textbf{a password} that has the same hash value as the original password, and not necessarily \textbf{the password} of the user. In the presence of bias, and under the assumption that the attacker \textbf{does not know} the structure of the hash function, the problem of guesswork for cracking passwords deviates away from the classic definition of guesswork where an attacker is trying to correctly guess the exact secret. The reason for this is that in the problem of password cracking the attacker cannot guess bins (outputs of a hash function) directly. In fact, the attacker can only guess inputs to the hash function, since  it does not a-priori know to what bins they are mapped. Note that for crytographic hash functions, finding the inputs corresponding to the particular bin is computationally difficult.

In order to capture the effect of bias on hash functions under brute force attacks we define a new concept, the conditional average guesswork conditioned on \textbf{the set of occupied bins}; in the presence of bias the average guesswork may change as a function of the hash values to which passwords are mapped. Since the conditional average guesswork may change as a function of the set of occupied bins, we define the maximum and the minimum conditional average guesswork (i.e., find the set of occupied bins that maximizes/minimizes the conditional average guesswork). Furthermore, we define the most likely conditional average guesswork, which is basically the conditional average guesswork that is most likely to occur, assuming that  passwords and as a result the set of occupied bins are drawn at random. Furthermore, we also consider the unconditional average guesswork, which is basically the average over the conditional average guesswork.

We focus on the type-class statistical profile, where the empirical distribution is the same as the probability of a given type in the method of types, when the dominant type-class is $p$\footnote{Note that this is a special class of profiles, as the statistical profile could be arbitrary, but this special class is motivated when considering different groups of bins in a hash function that have different bias, as also explained in Section \ref{sec:GeneralExpressions}.}.
The first main result of this paper studies the problem for such a class of statistical profiles using the tools from the method of types \cite{CoverBook}. For such a profile, when the most likely (dominant) type-class has distribution $p$, with a subset of $2^{H\pa{p}\cdot m}$ bins out of the entire set of bins (in total $2^{m}$ bins when every bin is represented by a binary string of length $m$), having this distribution. For such hash functions we show that the most likely conditional average guesswork under an insider attack increases like $2^{\pa{H\pa{p}-R}\cdot m}$, where $H\pa{p}$ is the binary Shannon entropy, $2^{R\cdot m}$ is the number of users, $0\le R\le H\pa{p}$, and $p$ parameterizes the bias. Furthermore, when $R<H\pa{p}/2$ the probability for the most likely conditional average guesswork goes to 1 with $m$.  

The expression above reveals an interesting connection between the effect of bias on the level of security and the effect of increasing the number of users in a system. Basically, increasing the number of users has a far greater effect in terms of average guesswork. In addition, the rate at which the conditional average guesswork increases depends on the Shannon entropy rather than R\'{e}nyi entropy \cite{Arikan_Ineq_Guessing}, which is the dominated term in classic results on guesswork. We also present a concentration result showing that the most likely conditional average guesswork is also the most likely actual number of guesses, and also that the event where the number of guesses reaches this average is not an atypical event (on the other hand, in classic results on guesswork the event where the number of guesses hits the average guesswork is an atypical event). The mechanism that leads to concentration and why it is different from classical guesswork, is described in Subsection \ref{subsec:LargeDeviationVsConcentration}. Essentially, it results from the fact that there is a super exponential number of possible hash functions, whereas the maximum number of guesses increases exponentially.

Furthermore, the maximum conditional average guesswork as a function of the type-class of the statistical profile increases like $2^{m\cdot D\pa{s||p}}$, where the number of users is $2^{R\cdot m}=2^{H\pa{s}\cdot m}$, $1/2\le s\le 1$,  and $D\pa{\cdot || \cdot}$ is the Kullback-Leibler divergence, used above for $Bernoulli\pa{s}$ and $Bernoulli\pa{p}$ distributions. Note that this function is unbounded as a function of $p$, that is, as $p$ decreases and bias in the hash function increases, the exponent of the maximum conditional average guesswork increases \textbf{unboundedly} (as a function of the number of bins $2^{m}$, not the number of inputs which is always greater). The reason for this is that the non-uniform distribution makes it harder to guess inputs that are mapped to the less likely bins.
Other than the results presented above this paper contains more results on the average guesswork under remote and insider attacks. We present however only some of the results in the introduction for clarity. 

Another question that this work addresses is whether or not bias can be used to increase the average guesswork. It is shown in this paper that indeed bias can increase the average guesswork as long as there is a backdoor mechanism that the system can use that enables it to map passwords of users to the least likely bins. An efficient backdoor mechanism is presented and it is shown that under a remote attack the average guesswork increases like $2^{\pa{2H\pa{p}+D\pa{1-p||p}-R}\cdot m}$. Note that again, the exponent of the average guesswork increases unboundedly as $p$ decreases (i.e., the bias increases).

Another set of results in this paper analyze the average guesswork for strongly universal hash functions when the key according to which hash functions are chosen is biased. For the case where the binary elements of the key are drawn i.i.d. Bernoulli$\pa{p}$ it is shown that when a backdoor mechanism is in place (similar to the one discussed above) and when considering a target average guesswork, the length of a biased key required to achieve this target is shorter than the length of a key that is drawn uniformly (i.e., i.i.d. Bernoulli$\pa{1/2}$).

This paper presents new definitions and results on the average guesswork of hash functions. For better understanding of the structure of the paper it is recommended to go to Section \ref{sec:RestiltsOutline}, which outlines the main technical results and points to where they can be found. In addition, to better understand the mechanism that leads to concentration and how it is different from classical guesswork, we recommend to go to Subsection \ref{subsec:LargeDeviationVsConcentration}.

The paper is organized as follows. We begin in Section \ref{sec:BasicDefinitions} with background and basic definitions. The attack model is presented in Section \ref{sec:ProblemSetting}, followed by general expressions for the average guesswork of a bin in Section \ref{sec:GeneralExpressions}. We then present in Section \ref{sec:AverageGuessworkPhHashFunctionsNoMod} bounds on the conditional average guesswork under remote and insider attacks, and the average guesswork along with concentration result in Section \ref{sec:NonConditionalAverageGuessworkPhHashFunctionsNoMod}. The most likely average guesswork is analyzed in Section \ref{sec:MostLikelyAverageGuessworkPhHashFunctionsNoMod}, followed by overview of the results in Section \ref{sec:RestiltsOutline}.  We present in Section \ref{sec:ModifiedHashFunctionAverageGueddwork} the case where hash functions can be modified, and analysis of the average guesswork of strongly universal hash functions in Section \ref{sec:AverageGuessworkCrackingPasswords}.  Finally, Section \ref{sec:Discussion} shows that a shorter biased key can achieve the same average guesswork as a longer uniform key.

\section{Background and Basic Definitions}\label{sec:BasicDefinitions}
In this section we present basic definitions and results in the literature that we use throughout the paper.
\subsection{Guesswork}\label{subsec:BackgroundClassicGuesswork}
Consider the following game: Bob draws a sample $x$ from a random
variable $X$, and an attacker Alice who does not know $x$ but knows
the probability mass function $P_{X}\pa{\cdot}$, tries to guess it. An
oracle tells Alice whether her guess is right or wrong.

The number of guesses it takes Alice to guess $x$ successfully is a random variable $G\pa{X}$ (which is termed guesswork) that takes only positive integer values. The optimal strategy of guessing $X$ that minimizes all non negative moments of $G\pa{X}$
\beq{}
{
E\pa{G\pa{X}^{\rho}}=\sum_{x\in X} G\pa{x}^{\rho}\cdot P_{X}\pa{x}\quad\rho\ge 0
}
is guessing elements in $X$ based
on their probabilities in descending order \cite{Messy1994}, \cite{Arikan_Ineq_Guessing}, such that $G\pa{x}<G\pa{x^{\prime}}$ implies $p_{X}\pa{x}>p_{X}\pa{x^{\prime}}$, i.e. a dictionary attack \cite{NISTPassword}; where $G\pa{x}$ is the number of guesses after which the attacker guesses $x$. In this paper we analyze the optimal guesswork and therefore with slight abuse of notations we define $G\pa{X}$ to be the \emph{optimal guesswork}.

It has been shown that $E\pa{G\pa{X}^{\rho}}$ is dictated by the R\'{e}nyi entropy \cite{Arikan_Ineq_Guessing}
\beq{eq:RenyiEntropyFunctionofRho}{
H_{\frac{1}{1+\rho}}\pa{P_{X}}=\frac{1}{\rho}\log_{2}\pa{\pa{\sum_{x}P_{X}\pa{x}^{\frac{1}{1+\rho}}}^{1+\rho}}.}
For example, when
drawing a random vector $\underline{X}$ of length $k$, which is
independent and identically distributed (i.i.d.) with distribution
$P=[p_1,\ldots,p_M]$, the exponential growth rate of the average guesswork
scales according to the R\'{e}nyi entropy $H_{\alpha}\pa{X}$ with parameter $\alpha=1/2$
\cite{Arikan_Ineq_Guessing}:
\beq{eq:GuessworkExpGrowthRate}
{
\lim_{k\to\infty}\frac{1}{k}\log_{2}\pa{E\pa{G\pa{X}}}=
H_{1/2}\pa{P_{X}}=2\cdot\log_{2}\pa{\sum_{x}\pa{P_{X}\pa{x}}^{1/2}} } where
$H_{\frac{1}{2}}\pa{P}\ge H\pa{P}=-\sum_{x\in X}p\pa{x}\log\pa{p\pa{x}}$ which is the Shannon entropy, with equality only for the uniform
probability mass function. Furthermore, for any $\rho\ge 0$
\beq{eq:GuessworkExpGrowthRateAnyMoment}
{
\lim_{k\to\infty}\frac{1}{k}\log_{2}\pa{E\pa{G\pa{X}^{\rho}}}=
\rho\cdot H_{\frac{1}{1+\rho}}\pa{P_{X}}.
}

The definition of guesswork was also extended to the case where the attacker has a side information $Y$ available \cite{Arikan_Ineq_Guessing}. In this case the average guesswork for $Y=y$ is defined as $G\pa{X|Y=y}$, and the $\rho$th moment of $G\pa{X|Y}$ is
\beq{eq:ConditionalGuesswork}{
E\pa{G\pa{X|Y}^{\rho}}=\sum_{y}E\pa{G\pa{X|Y=y}^{\rho}}\cdot P_{Y}\pa{y}.
}
Arikan \cite{Arikan_Ineq_Guessing} has bounded the $\rho$th moment of the optimal guesswork, $G\pa{X|Y}$, by
\bal{eq:GeneralExp}{
\pa{1+\ln\pa{M}}^{-\rho}\sum_{y}\pa{\sum_{x}P_{X,Y}\pa{x,y}^{\frac{1}{1+\rho}}}^{1+\rho}
\le \nn\\
E\pa{G\pa{X|Y}^{\rho}}\le \sum_{y}\pa{\sum_{x}P_{X,Y}\pa{x,y}^{\frac{1}{1+\rho}}}^{1+\rho}
}
where $M=|X|$ is the cardinality of $X$. Furthermore, in \cite{Arikan_Ineq_Guessing} it has been shown that when $X$ and $Y$ are identically and independently distributed (\iid), the exponential growth rate of the optimal guesswork is
\beq{eq:asymptiticconditionalguess}{
\lim_{k\to\infty}\frac{1}{k}\log_{2}\pa{E\pa{G^{\ast}\pa{X|Y}^{\rho}}}=\rho\cdot H_{\frac{1}{1+\rho}}\pa{P_{X,Y}}
}
where $m$ is the size of $X$ and $Y$, and
\beq{eq:TheRenyiEntropy}{
H_{\frac{1}{1+\rho}}\pa{P_{X,Y}}=\frac{1}{\rho}\log_{2}\pa{\sum_{y}\pa{\sum_{x}P_{X,Y}\pa{x,y}^{\frac{1}{1+\rho}}}^{1+\rho}}} is R\'{e}nyi's conditional entropy of order $\frac{1}{1+\rho}$ \cite{Arikan_Ineq_Guessing}.

\subsection{Method of Types}
In this paper we analyze the average guesswork of hash functions based on the method of types \cite{CoverBook}. Therefore, we now present 2 basic results in method of types that are used throughout this paper. 

Assume a binary vector $\udl{X}$ of size $m$, the realization $\udl{x}$ is of type $q$ in case
$N\pa{1|\udl{x}}/m=q$
where $N\pa{1|\udl{x}}$ is the number of occurrences of the number $1$ in the binary vector $\udl{x}$.

The first result is related to the probability of drawing a vector of type $q$ when the underlying probability is i.i.d. Bernoulli$\pa{p}$.
\begin{lem}[\cite{CoverBook} Theorem 11.1.2]\label{th:methodoftypesprobth}
The probability $Q\pa{q}=P\pa{N\pa{1|\udl{x}}/m=q}$ which is the probability that $\udl{x}$ is of type $q$, is equal to
\beq{}
{
Q\pa{q}=2^{-m\pa{H\pa{q}+D\pa{q||p}}}
}
where $D\pa{q||p}=\sum_{i}q_{i}\log\pa{q_{i}/p_{i}}$ is the Kullback-Leibler divergence \cite{CoverBook}.
\end{lem}

The next lemma bounds the number of elements for each type.
\begin{lem}[\cite{CoverBook} Theorem 11.1.3]\label{th:methodoftypessize}
The number of elements for which $N\pa{1|\udl{x}}/m=q$ is bounded by the following terms
\beq{}
{
\frac{1}{\pa{m+1}^{2}}\cdot 2^{m\cdot H\pa{q}}\le |N\pa{1|\udl{x}}=q|\le 2^{m\cdot H\pa{q}}.
}
\end{lem}

\subsection{Hash Functions and their Statistical Profile}\label{subsec:SUSHF}
A hash function, $H_{F}\pa{\cdot}$, is a mapping from a larger domain $A$ to a smaller range $B$. In this work we assume that the input is of length $n$ bits whereas the output is $m$ bits long, where $n> m$. We term the binary string of length $m$  that represents the output of a hash function \emph{bin} which is denoted by $b\in\pac{1,\dots,2^{m}}$. 

For the analysis of the average guesswork of hash functions, we define the \emph{statistical profile} of a hash function which is a probability mass function defined over the set of bins, where each element is defined as the relative number of inputs that are mapped to a bin. For this, we use the following definition throughout the paper.
\begin{defn}[Statistical profile of a hash function]\label{def:FractionsofMappingsanyhash}
Consider a hash function $H_{F}\pa{\cdot}$ with an input of size $2^{n}$ and an output of size $2^{m}$, where $n> m$. The statistical profile of $H_{F}\pa{\cdot}$ is denoted by $P_{H}$, and defined as follows
\beq{}
{
\ph = \frac{1}{2^{n}}\sum_{i=1}^{2^{n}}\mathbbm{1}_{b}\pa{H_{F}\pa{i}}\quad \forall b\in\pac{1,\dots, 2^{m}}
}
where $\mathbbm{1}_{b}\pa{x}=\barr{cc}{1 &x=b\\ 0 &x\neq b}$. Thus, for every bin $P_{H}$ is the relative number of inputs that are mapped to it.
\end{defn}

\begin{defn}
A $P_{H}$\emph{-hash function} is a hash function $H_{F}\pa{\cdot}$ whose statistical profile is $P_{H}$.
\end{defn}

\begin{example}
When the relative number of inputs that are mapped to a bin is the same across all bins in $H_{F}\pa{\cdot}$, $P_{H}$ is distributed uniformly, that is, $\ph = 2^{-m}$ for every $b\in\pac{1,\dots,2^{m}}$.
\end{example}

\begin{rem}
Basically a $P_{H}$-hash function can be represented by a bipartite graph, where the degree of every input is one, whereas the degree of every bin  $b\in\pac{1,\dots,2^{m}}$, divided by $2^{n}$ is $\ph$. Figure \ref{fig:HashFunctionBipartiteRepresentation} illustrates this.
\end{rem}

\begin{figure}[h!]
\input{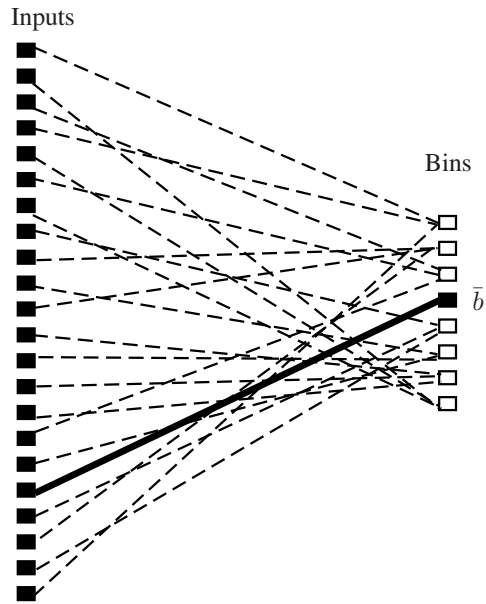}\caption{A bipartite graph representation of a hash function. In this case the statistical profile is $\pac{3/22,3/22,3/22,\emph{1/22},3/22,3/22,3/22}$.}\label{fig:HashFunctionBipartiteRepresentation}
\end{figure}

\begin{defn}[A keyed hash function]
A Keyed hash function is a set of hash functions, where the actual hash function being used is determined by the key value.
\end{defn}

Based on the definition above, a keyed hash function over a set of $K$ hash functions
$\pac{H_{F}^{\pa{1}}\pa{\cdot}, \dots,H_{F}^{\pa{K}}\pa{\cdot} }$ can be written as follows
\beq{}{
HK\pa{b,k} = H_{F}^{\pa{k}}\pa{b}\quad \forall b\in\pac{1,\dots,2^{m}},
}
where $k\in\pac{1,\dots,K}$, and $HK\pa{\cdot}$ is the keyed hash function.

\section{The Attack Model}\label{sec:ProblemSetting}
In this section we define the attack model based on how passwords are stored in systems. We present two types of attacks for this scenario: A remote attack and an insider attack.

The following definitions address the way passwords are stored in a system, and how access to a system is granted.
\begin{defn}[Hashed password storage method]
We define the method according to which the system stores passwords.
\bitm
{
\item There are $M$ users registered in the system.
\item In order to access the system a user provides his user name and password.
\item The system does not store the passwords, but rather stores the user names and the hashed values of the passwords.
\item The system hashes the password.
}
\end{defn}

\begin{defn}[Authentication scheme]\label{def:GivingUserAnAccess}
The following protocol grants a user access.
\bitm
{
\item The user sends its user name to the system.
\item The system finds the bin value which is coupled with this user name.
\item The user types a password; the system hashes the password; if the hashed value matches the bin from the previous bullet, then access is granted.
}
\end{defn}

Figure \ref{fig:Problem_Setting} presents the authentication scheme defined above.

\begin{defn}[The set of occupied bins]\label{def:TheSetofOccupieeBins}
The set of occupied bins $B$ is the set of bins to which the passwords of valid users are mapped, that is, these are the hash values of the passwords of valid users. The size of the set of occupied bins $|B|=2^{R\cdot m}$ is the number of bins in the set.  
\end{defn}

Figure \ref{fig:ExampleForOccupiedBins} illustrates what the set of occupied bins $B$ means.

\begin{figure}[!h]
\centering\scalebox{.3}{\input{Average_Set_of_Bins.TpX}}\caption{Overall there are $2^{m}$ bins out of which the passwords are mapped to a subset $B$, which is termed the occupied bins. Note that $UN$ represents user name, and $Ps_{i}$ represents the password of user $i$. }\label{fig:ExampleForOccupiedBins}
\end{figure}

Essentially, Definition \ref{def:TheSetofOccupieeBins} states that in case we have $K$ users whose passwords are $\pac{Ps_{1},\dots,Ps_{K}}$, then their passwords are mapped by a hash function $H_{F}\pa{\cdot}$ to the bins $\pac{H_{F}\pa{Ps_{1}}=b_{1},\dots,H_{F}\pa{Ps_{K}}=b_{K}}\subseteq \pac{1,\dots, 2^{m}}$. 

Next, we define two brute-force attack models under which we analyze the average guesswork. The first is a remote attack in which the attacker does not know the set of occupied bins. The second is an insider attack in which an attacker has insider access to the system so that it can see the set of occupied bins and test the hash function directly and compare the hashed values to the values in the  set of occupied bins. 

Figure \ref{fig:TheDifferenceBetweenOnlineAndOfflineAttacks} illustrates the difference  between these two types of attacks.

\begin{defn}[A remote attack]\label{def:AnOnlineAttackModel}
The following steps define a remote attack:
\bitm
{ 
\item The attacker does not know what the hash function is, that is, it does not know to which outputs (i.e., bins) the inputs of the hash function are mapped.\footnote{In fact, under this setting the attacker may know the statistical profile $P_{H}$ given in Definition \ref{def:FractionsofMappingsanyhash}, but it does not know to which bins the inputs are mapped, that is, in terms of Figure \ref{fig:HashFunctionBipartiteRepresentation} it does not know the connections on the left side, but it may know the degrees of the nodes on the right hand side.}
\item The attacker does not know the passwords of the users; it also does not know the set of occupied bins.
\item The attacker can guess the passwords one by one; it can choose any strategy for guessing the passwords.
\item The attacker chooses a user name and then guesses passwords one by one. Once the attacker guesses \textbf{a password} that is mapped to the same bin as \textbf{the password} of the user, it has \textbf{cracked} the password and the game is over.
}
\end{defn}

\begin{figure}[h!]
\input{Problem_Setting.TpX}\caption{The user/attacker enters a user name and a password. The server hashes the password with a hash function, and compares its output to the value that is stored on the server for the user name that was entered. When the values match access is granted; otherwise access is denied.}\label{fig:Problem_Setting}
\end{figure}

Basically, Definition \ref{def:AnOnlineAttackModel} states that the attacker does not know the structure of the hash function $H_{F}\pa{\cdot}$ that the system uses (or only knows the degree of the outputs in the bipartite graph representing this hash function, that is, its statistical profile). As an illustrative example consider the case where the attacker targets user 1 whose password $Ps_{1}$ is mapped to $H_{F}\pa{Ps_{1}} = b_{1}$. The attacker can only guess inputs to the hash function one by one until he finds an input $Ps^{\ast}$ such that $H_{F}\pa{Ps^{\ast}}=b_{1}=H_{F}\pa{Ps_{1}}$. Note that because the hash function is a many to one mapping, it is possible that $Ps^{\ast}\neq$ $Ps_{1}$.

\begin{defn}[An insider attack]\label{def:AnOfflineAttackModel}
We make the same assumptions as in Definition \ref{def:AnOnlineAttackModel} (in particular it still does not know the hash function), with the exception that this time the attacker has insider access so that it can test the hash function directly and compare the hashed values of his guesses to the elements in the set of occupied bins, as it knows the user-name, and hashed password pairs.
\bitm{
\item The password is cracked when the attacker finds \textbf{a password} that is mapped to \textbf{any of the occupied bins}.  
}
\end{defn}

The attack defined above essentially requires that the attacker finds an input $Ps^{\ast}$ so that $H_{F}\pa{Ps^{\ast}}\in\pac{b_{1},\dots,b_{K}}$. Note that this can happen even in the case where $Ps^{\ast}\not\in\pac{Ps_{1},\dots, Ps_{K}}$.

\begin{figure}
\subfigure[A remote attack, where the password is cracked when the attacker finds an input that is mapped to the same bin as the password of the user.]{
\scalebox{.6}{\input{Online_Attack.TpX}}}
\hspace{1in}
\centering\subfigure[An insider attack where the password is cracked  when the attacker finds an input that is mapped to any of the occupied bins, that is, any element in $B$.]{
\scalebox{.4}{\input{Ofline_Attack.TpX}}} \caption{An illustration of the difference between remote and insider attacks. Note that $UN$ represents a user-name.}\label{fig:TheDifferenceBetweenOnlineAndOfflineAttacks}
\end{figure}

\section{General Expressions For the Average Guesswork of a bin}\label{sec:GeneralExpressions}

In this section we present general expression for the average guesswork of guessing an input that is mapped to a bin. We begin by presenting a set of hash functions with the same statistical profile, we then define averaging arguments for the average guesswork, and then we present the average guesswork per bin under these assumptions.

We now define a set of hash functions, consisting of hash functions that have the same statistical profile.

\begin{defn}[The $P_{H}$-set of hash functions]\label{def:TheSetofPhHashFunctions}
The $P_{H}$-set of hash functions is the set of all hash functions whose statistical profile is  $P_{H}$ as given in Definition \ref{def:FractionsofMappingsanyhash}; that is, for any hash function in the set and for every bin $b\in\pac{1,\dots,2^{m}}$, the relative number of inputs mapped to $b$ is $\ph$; the difference between the hash functions in this set is the actual inputs that are mapped to every bin. 
\end{defn}

Essentially, Definition \ref{def:TheSetofPhHashFunctions} states the following: A hash function $H_{F}^{\ast}\in$ $P_{H}$-set of hash functions, if and only if its statistical profile $P_{H{\ast}}$ is equal to $P_{H}$, that is, \beq{}{
P_{H^{\ast}}\pa{b} = P_{H}\pa{b}\quad \forall 
b\in\pac{1,\dots,2^{m}}} 

\begin{rem}\label{rem:PermutationsHashFunctionsRepresentTheSet}
The $P_{H}$-set of hash functions is basically all the bipartite graph representations whose output degrees divided by $2^{n}$ give $\ph$, for every $b\in\pac{1,\dots,2^{m}}$. Therefore, it consists of all possible permutations of edges from inputs to bins that yield $\ph$.
\end{rem}
Figure \ref{fig:IntuitiveExplanationPermutation} illustrates the remark above.

\begin{defn}[A keyed $P_{H}$-hash function]
A set of $P_{H}$-hash functions, where the actual function is determined by a key, which chooses a permutation of the inputs. 
\end{defn}

Essentially, the definition above means that for every hash function  $H_{F}^{\ast}\in$ $P_{H}$-set of hash functions, there is a key value associates with it. 

\begin{figure}[!h]
\centering\scalebox{.8}{\input{Intuitive_Explenation_Permutation.TpX}}\caption{A $P_{H}$-set of hash functions. The outputs (right hand side) are fixed whereas each hash function is a different permutation of the edges that go to the inputs (the left hand side).In terms of keyed hash functions it means that the key chooses a permutation (i.e., a hash function).}\label{fig:IntuitiveExplanationPermutation}
\end{figure}

We now define the average guesswork for password cracking.
\begin{defn}[The guesswork of hash functions]\label{def:TheGuessworkOfACertainBin}
The following are the two types of guesswork that we analyze in this paper.
\bitm{
\item The number of guesses of inputs required to find \textbf{an input} that is mapped to a particular bin $b$ is $\Gbin$, where $b\in \pac{1,\dots,2^{m}}$.
\item Given a set of occupied bins $B$, the number of guesses of inputs required to find an input that is mapped to any element in $B$ (rather than a particular element) is $\Gof$.  
}
\end{defn}

This definition is in line with the attack models presented in Definition \ref{def:AnOnlineAttackModel} and Definition \ref{def:AnOfflineAttackModel}.

Guesswork is a random variable, where the number of guesses depends on a certain source of randomness. The following definition presents the source of randomness for the average guesswork of hash functions given a set of occupied bins.We present two methods for averaging guesswork, the first one is in line with classical guesswork, whereas the second can be used in the context of proof of work (see Remark \ref{rem:Guessing_Strategies_and_POW}). We put these definitions together as most of the theoretical results in the paper hold for both cases naturally. Furthermore, Equation \eqref{eq:Illustration_Source_Randomness_Hah_func}, and Equation \eqref{eq:Illustration_Source_Randomness_strategies} illustrate  definition \ref{def:AveragingArgument}.  

\begin{defn}[Averaging Arguments]\label{def:AveragingArgument}
We define two methods for averaging the number of guesses required to find an input that is mapped to a particular bin $b$ ($\Gbin$) as well as to any element in $B$ ($\Gof$).
\bitm
{
\item \textbf{A keyed $P_{H}$-hash function}: For any strategy of guessing passwords one by one we average over all hash functions whose statistical profile is $P_{H}$ (as presented in Figure \ref{fig:IntuitiveExplanationPermutation}). We average by assuming that a key which is uniformly distributed chooses an element form this set of hash functions.
\item\textbf{A $P_{H}$-hash function unknown to the attacker}: In this case we average over all strategies of guessing passwords one by one. We assume that every strategy of guessing is equally likely to be chosen.\footnote{Averaging over all strategies for guessing passwords one by one may not be a usual way of evaluating the level of security. In particular, this method is meaningful only when the passwords of the users are drawn \textbf{uniformly} (whereas the bias lies in the statistical profile of the hash function). Indeed, in our problem when the attacker does not know the mappings from inputs to outputs of \textbf{a hash function}, and passwords are drawn uniformly,  this definition makes sense and in fact achieves the same average guesswork as a keyed hash function and any strategy of guessing passwords one by one.}
}
\end{defn}

For simplicity let consider $E\pa{\Gbin}$. Essentially, Definition \ref{def:AveragingArgument} states that under the keyed $P_{H}$-hash function assumption we get that for a given strategy of guessing inputs one by one, the average guesswork can be written as
\beq{eq:Illustration_Source_Randomness_Hah_func}{
E\pa{\Gbin}=\frac{1}{|\mathbf{H}|}\sum_{H_{F}\in\mathbf{H}}G_{bin}\pa{b|H_{F}},
}
where $\mathbf{H}$ is the set of all $P_{H}$-hash functions, and $G_{bin}\pa{b|H_{F}}$ is the number of guesses required to find an input that is mapped to bin $b$ under this given strategy of guessing inputs. On the other hand, Definition \ref{def:AveragingArgument} states that under the $P_{H}$-hash function unknown to the attacker assumption, we get that given a hash function $H_{F}\pa{\cdot}$ whose statistical profile is $P_{H}$, the average guesswork can be written as  
\beq{eq:Illustration_Source_Randomness_strategies}{
E\pa{\Gbin}=\frac{1}{|\mathbf{S}|}\sum_{S\in\mathbf{S}}G_{bin}\pa{b|S},
}
where $\mathbf{S}$ is the set of all strategies of guessing passwords one by one, and 
$G_{bin}\pa{b|S}$ is the number of guesses required to find an input that is mapped to bin $b$ under strategy $S$ for the function $H_{F}\pa{\cdot}$.

\begin{rem}\label{rem:Guessing_Strategies_and_POW}
Although averaging over all strategies of guessing passwords one by one is not natural for classical guesswork (where the user gets to choose the strategy), it is useful for guessing using \textbf{nonce} in the context of proof of work. In this case the attacker draws an input at random (nonce) at every step as an input to the hash function. This is equivalent to drawing a strategy for guessing passwords one by one.
\end{rem}

\begin{theorem}[Average number of guesses of a bin]\label{th:TheAverageGuessworkPerBinGeneralExpressions}
Assume that the attacker guesses inputs to a hash function $H_{F}\pa{\cdot}$ whose statistical profile is $P_{H}$ until it finds an input (any input) that satisfies $H_{F}\pa{Ps^{\ast}}=b_{0}$, that is, an input that is mapped to bin $b_{0}$.  When $P_{H}\pa{b_{0}}=2^{-\Al\cdot m}$, $\Al>0$, and under the averaging arguments of Definition \ref{def:AveragingArgument} we get that when $n\ge \pa{1+\eps_{1}}\cdot\Al\cdot m$, $\eps_{1}>0$, the average guesswork of bin $b_{0}$, $G_{bin}\pa{b_{0}}$, is equal to
\beq{}
{
\lim_{m\to\infty}\frac{1}{m}\log\pa{E\pa{G_{bin}\pa{b_{0}}}}=\lim_{m\to\infty}\frac{1}{m}\log\pa{1/P_{H}\pa{b_{0}}}=\Al .}
\end{theorem}
\begin{proof}[Sketch of proof]
 The general idea is that when $n$ is large enough, the probability of finding an input that is mapped to bin $b_{0}$ converges to the geometric distribution $\ph\cdot\pa{1-\ph}^{i-1}$, where $i\ge 1$ is the number of guesses. The challenge in showing this lies in the fact that the chance of hitting this bin in the first guess is $P_{H}\pa{b_{0}}$ under the averaging arguments presented in the previous section, whereas as the number of guesses increases the chance of guessing an input that is mapped to $b_{0}$ increases. The reason is that the attacker does not repeat his guesses and the ensemble of hash functions is based on permutation of the edges that are connected to the inputs (Figure \ref{fig:IntuitiveExplanationPermutation}). When the input size is large enough and $\ph$ decreases exponentially with $m$, the probability of guessing converges to geometric distribution which means that as $m$ increases the average guesswork converges to $1/P_{H}\pa{b_{0}}=2^{\Al\cdot m}$. The proof is in Section \ref{sec:BoundsAverageGuessworkNoBinsAllocation}.
\end{proof}

\begin{rem}
Note that under the assumptions of Theorem \ref{th:TheAverageGuessworkPerBinGeneralExpressions}, the derivation of the average guesswork under an insider attack $E\pa{\Gof}$ is straightforward as long as $\sum_{b\in B}\ph=2^{-\Al^{\ast}\cdot m}$, where $\Al^{\ast}>0$.  
\end{rem}

\section{The Conditional Average Guesswork of type-class statistical profile hash functions}\label{sec:AverageGuessworkPhHashFunctionsNoMod}

We begin this section by defining the conditional average guesswork under insider and remote attacks. We then define a type-class statistical profile that we use in this paper to model bias in hash functions. Finally, we present bounds on the average guesswork for both types of attacks.

In this section we analyze the average guesswork under insider and remote attacks for a type-class statistical profile, which we \textbf{interchangeably use with the terminology $P_{H}$-hash functions} (the type-class statistical profile that we consider in this paper is defined in this section). Essentially, this type-class statistical profile enables us to model bias for hash functions and also come up with closed form expressions for the average guesswork that shed light on the effect of bias on the level of security.

First, let us define the average guesswork under remote and insider attacks.

\begin{defn}[The conditional average guesswork under a remote attack]\label{def:AverageGuessworkOnLineAttackNoEquation}
Assume that the number of users is $2^{R\cdot m}$, where $0\le R\le 1$, as presented in Definition \ref{def:TheSetofOccupieeBins}, and that their passwords are mapped to bins in the set $B$. The average guesswork under a remote attack \textbf{conditioned on the set of occupied bins} $B$, $E\pa{\Gon}$, is the average number of guesses required to find an input that is mapped to the same bin as the password of a user, averaged over all users.
\end{defn}

The following lemma puts in equation form Definition \ref{def:TheSetofOccupieeBins}.
\begin{lem}
The conditional average guesswork under a remote attack for a set of occupied bins $B$ (presented in Definition \ref{def:TheSetofOccupieeBins}) can be written as follows:
\beq{}
{
E\pa{\Gon}=2^{-R\cdot m}\sum_{b\in B}E\pa{\Gbin}.
}

\end{lem}
\begin{proof}
The proof is straight forward based on Definition \ref{def:AverageGuessworkOnLineAttackNoEquation}.
\end{proof}

\begin{defn}[The conditional average guesswork under an insider attack]
The average guesswork under an insider attack conditioned on the set of occupied bins $B$, is 
$E\pa{\Gof}$, where $\Gof$ is presented in Definition \ref{def:TheGuessworkOfACertainBin}.
\end{defn}

\begin{rem}
Note that $\Gof$ can be viewed as an extension of $\Gbin$ to the case where a password has to hit a ``super'' bin, which is the set of occupied bins. 
\end{rem}

\begin{rem}
Note that $\Gof$ and $\Gon$ are conditioned on the set of occupied bins $B$, that is, they are a function of $B$. Basically it means that as $B$ changes their underlying probability mass function as well as their average change too.
\end{rem}

In order to find closed form expressions for 
$$E\pa{\Gon}=2^{-R\cdot m}\sum_{b\in B}E\pa{\Gbin}$$
as well as for $E\pa{\Gof}$, and in order to gain better understanding of the effect of bias on the level of security of hashed passwords, we make in Definition \ref{def:DefinitionOfIidCase} an additional assumption about the statistical profile, that it is a type-class statistical profile, with  $P_{H}$ specified in Definition \ref{def:DefinitionOfIidCase}.

\begin{defn}\label{def:DefinitionOfIidCase}
\beq{}{
\ph = 2^{-m\cdot \pa{H\pa{q}+D\pa{q||p}}}\quad\forall b\in\pac{1,\dots,2^{m}}}
where $q=N\pa{1|b}/m$, and $N\pa{1|b}$ is the number of ones in the binary representation of $b\in\pa{1,\dots,2^{m}}$.
\end{defn}

\begin{rem}
The results in this paper under the assumption made in Definition \ref{def:DefinitionOfIidCase} also hold for any $P_{H}$ that is permutation of the mappings from $b$ to $P_{H}$. 
\end{rem}

\begin{rem}[Motivation for analyzing under the method of types]
The method of types means that elements in certain subsets of certain volumes are \emph{almost} equally likely (e.g., in the most likely subset the probability of each element decreases at rate $H\pa{p}$).  \end{rem}

Based on Definition \ref{def:DefinitionOfIidCase}, we can express the average guesswork for any bin presented in Theorem \ref{th:TheAverageGuessworkPerBinGeneralExpressions} as follows.

\begin{theorem}[Average number of guesses for each bin]\label{th:AverageGuessworkPerBinIidAssumption}
Assume that the attacker guesses inputs of a $P_{H}$-hash function one by one until it finds one that is mapped to bin $b$ (remote or insider attacks). In this case, under the averaging arguments of Definition \ref{def:AveragingArgument} we get that when $n\ge \pa{1+\eps_{1}}\cdot\log\pa{1/p}\cdot m$, $\eps_{1}>0$, the average guesswork of $\Gbin$ is equal to
\beq{eq:AverageGuessworkPerBinWhenUsersAreNotModified}
{
\lim_{m\to\infty}\frac{1}{m}\log\pa{E\pa{\Gbin}}=\lim_{m\to\infty}\frac{1}{m}\log\pa{1/\ph}=H\pa{q\pa{b}}+D\pa{q\pa{b}||p}}
where $b\in\pac{1,\dots,2^{m}}$, and $q\pa{b}=\frac{N\pa{1|b}}{m}$ is the type of $b$.
\end{theorem}
\begin{proof}[Sketch of proof]
Similarly to Theorem \ref{th:TheAverageGuessworkPerBinGeneralExpressions} we show that the probability of hitting a bin converges to geometric distribution. This in turn means that the average guesswork of bin $b\in\pac{1,\dots,2^{m}}$ can be achieved by assigning $\Al=H\pa{q\pa{b}}+D\pa{q\pa{b}||p}$, where this is achieved under the method of types and large deviation theory, which lead to $\ph \approx 2^{-\pa{H\pa{q\pa{b}}+D\pa{q\pa{b}||p}}\cdot m}$. However, unlike Theorem \ref{th:TheAverageGuessworkPerBinGeneralExpressions} here we want guesswork to converge to geometric distribution for every $b\in\pac{1,\dots,2^{m}}$. This in turn leads to the bound $n\ge \pa{1+\eps_{1}}\cdot\log\pa{1/p}\cdot m$. The proof is in Section \ref{sec:BoundsAverageGuessworkNoBinsAllocation}.
\end{proof}

\begin{rem}
The average guesswork ($E\pa{\Gon}$,  $E\pa{\Gof}$) is a function of the bins stored in the system, that is, a function of the set of occupied bins $B$. Therefore, unlike guesswork for the classic case, in the case of hashed passwords it is meaningful to consider upper and lower bounds on the average guesswork. This can be expressed through the maximum and minimum average guesswork as a function of $B$ and the number of users. The definition for maximum and minimum average guesswork can be found in Definition \ref{def:MaxMinAverageGiesswork}.
\end{rem}

Next, we present bounds for the average guesswork. We assume that each password consists of $n$ bits that are drawn i.i.d. Bernoulli$\pa{1/2}$. The assumption that passwords are drawn uniformly prevents the attacker form preferring one strategy of guessing password over the other, as long as he does not know the structure of the hash function.  In order to derive the upper bound we assume that the passwords are mapped to the set of occupied bins that achieves maximum average guesswork, whereas in order to derive the lower bound we assume that passwords are mapped to the set of occupied bins whose average guesswork is minimum. We formally define the maximum and minimum conditional average guesswork in Definition \ref{def:MaxMinAverageGiesswork}.

In the following theorems we make a slight change of notations for representation purposes and write $R=H\pa{s}$, where $1/2\le s\le 1$.

\begin{theorem}[Bounds on the conditional average guesswork under an insider attacks]\label{cor:LowerUpperTheAverageGuessworkForGuessingAnyPassword}
The average guesswork under an insider attack is bounded as follows:
\beq{eq:UpperLowerBoundIndiderGuessword190920}
{
\barr{cc}{D\pa{1-s||p} &0\le 1-s\le p\\ 0 &p\le 1-s\le 1/2 }\le\lim_{m\to\infty}\frac{1}{m}\log\pa{E\pa{\Gof}}\le D\pa{s||p},
}
under the following assumptions: 
\bitm
{
\item The passwords that the users have chosen are mapped to $2^{H\pa{s}\cdot m-1}$ different (unique) bins, where $1/2\le s\le 1$ (i.e., $|B| = 2^{H\pa{s}\cdot m-1}$); the passwords are drawn i.i.d. Bernoulli$\pa{1/2}$.
\item $P_{H}$ is defined in Definition \ref{def:DefinitionOfIidCase} under the method of types, where $p\le 1/2$.
\item $n\ge\pa{1+\eps}\cdot m\cdot \log\pa{1/p}$ where $\eps > 0$.
}
Where we average based on  Definition \ref{def:AveragingArgument}.
\end{theorem}
\begin{proof}[Sketch of proof]
The general idea is bounding the average guesswork based on equation \eqref{eq:MaxMinAveageGuessworkOnlineOffline}, \eqref{eq:MaxMinAveageGuessworkOnlineOffline1} and the method of types, assuming that the set of occupied bins is the set of bins of size $|B|=2^{H\pa{s}\cdot m-1}$ that occupies the $2^{H\pa{s}\cdot m-1}$ bins with highest (lowest) degrees in the bipartite graph that represent $P_{H}$. Basically, we show that the set of occupied bins whose bins have the highest degrees in $P_{H}$ achieves the maximum average guesswork, whereas the set of occupied bins whose bins have the lowest degrees, achieves the minimum, average guesswork. By finding the conditional average guesswork for
these two sets of occupied bins, we come up with the bounds for the average guesswork. We also wish to point out that the the bound $n\ge\pa{1+\eps}\cdot m\cdot \log\pa{1/p}$ where $\eps > 0$, guarantees that the underlying probability of the conditional guesswork (based on Definition \ref{def:AveragingArgument}) converges to a geometric distribution with parameter $P_{H}\pa{B}=\sum_{b\in B}P_{H}\pa{b}$  for every $B$; essentially, the $\log\pa{1/p}$ guarantees that even after guessing a large number of inputs, the underlying probability of hitting bin $b$ in the next guess is approximately $P_{H}\pa{b}$, where this approximation holds uniformly for every $b\in\pac{1,\dots,2^{m}}$. The full proof appears in Section \ref{sec:BoundsAverageGuessworkNoBinsAllocation}.
\end{proof}

\begin{rem}
The justification for the assumption that users are mapped to unique bins is based on Corollary \ref{cor:NoBinAllocationMostLikelyAverageGuesswork} and Corollary \ref{cor:NoBinAllocationMostLikelyAverageGuessworkOnLine} that show that for any practical purpose it is very likely that the passwords are mapped to unique bins. Essentially, the expressions can be modified to also reflect the case where users are not mapped to unique bins; however, this results in less elegant expressions.
\end{rem}

We now find bounds for the average guesswork under a remote attack.

\begin{theorem}[Bounds on the conditional average guesswork under a remote attack]\label{th:LowerUpperBoundExpAverageGuessworkAnyHashFunction}
The average guesswork when averaged according to Definition \ref{def:AveragingArgument}, is bounded by
\bal{eq:AverageGuessworkLowerUpperAnyhashFunctionFirstTimeAppears}
{
H\pa{s}+D\pa{1-s||p}&\le&\lim_{m\to\infty}\frac{1}{m}\log\pa{E\pa{\Gon}}\nn\\
&\le&\barr{cc}{H\pa{s}+D\pa{s||p} &\pa{1-p}\le s\le 1\\ 2\cdot H\pa{p}+D\pa{1-p||p}-H\pa{s} &1/2\le s\le \pa{1-p} },
}
under the following assumptions:
\bitm
{
\item There are $2^{H\pa{s}\cdot m-1}$ users whose passwords are mapped to different (unique) bins, where $1/2\le s\le 1$; passwords are drawn i.i.d. Bernoulli$\pa{1/2}$.
\item $P_{H}$ is based on the method of types as presented in Definition \ref{def:DefinitionOfIidCase}, where $p\le 1/2$.
\item $n\ge\pa{1+\eps}\cdot m\cdot \log\pa{1/p}$ where $\eps > 0$.
}
\end{theorem}
\begin{proof}
The full proof appears in Section \ref{sec:BoundsAverageGuessworkNoBinsAllocation}.
\end{proof}

\begin{rem}[Tightness of the bounds]
Note that the bounds in equation \eqref{eq:UpperLowerBoundIndiderGuessword190920} and equation \eqref{eq:AverageGuessworkLowerUpperAnyhashFunctionFirstTimeAppears} are tight in the sense that there exist sets of occupied bins whose average guesswork achieve those appear and lower bounds. In particular, the upper bound is achieved by the set of bins whose degree is maximum in the bipartite graph that represents the $P_{H}$-set of hash functions, whereas the lower bound is achieved by the bins whose degrees are minimum. 
\end{rem}

\begin{rem}
Note that the upper bounds in equation \eqref{eq:UpperLowerBoundIndiderGuessword190920} and equation \eqref{eq:AverageGuessworkLowerUpperAnyhashFunctionFirstTimeAppears} are unbounded functions that increase as $p$ decreases; this means that the average guesswork over a set of $2^{m}$ bins can be larger than $2^{m}$. This is possible because the guesses are done over the set of inputs to the hash function, whose size is $2^{n}$ combined with the fact that $n\ge\log\pa{1/p}\cdot m$. In particular, as $p$ decreases these bounds become more asymptotic.
\end{rem}

\begin{rem}
In this subsection we assume that passwords are drawn uniformly. When passwords are biased there is a certain optimal strategy \cite{Arikan_Ineq_Guessing} and therefore, there is no use of averaging over all possible strategies (i.e.,  the $P_{H}$-hash function unknown to the attacker case in Definition \ref{def:AveragingArgument}); in this case averaging is done over a keyed $P_{H}$-hash function as presented in Definition \ref{def:AveragingArgument}. Furthermore, similarly to Corollary \ref{cor:EffectofBiasedPassword} the average guesswork in this case is equal to the dominant term between the average number of guesses required to guess \textbf{the password}, and the average number of guesses required to guess \textbf{a password} that is mapped to the same bin.
\end{rem}

\section{The Average Guesswork of type-class statistical profile hash functions}\label{sec:NonConditionalAverageGuessworkPhHashFunctionsNoMod}
In this section we define the probability over the set of occupied bins induced by drawing passwords at random. Furthermore, we present the average guesswork averaged over all passwords under a remote attack as well as a concentration result for $\Gbin$. Finally, we explain why $E\pa{\Gbin}$ is not affected by large deviations and show how it is different from classic average guesswork.  

The fact that passwords of users are chosen at random induces a probability over the set of occupied bins $B$. It results from the fact that for a given hash function $H_{F}\pa{\cdot}$,  choosing passwords at random leads to certain probability that a set $B$ of occupied bins is occupied. Furthermore, when passwords are drawn uniformly and independently, the probability that these passwords occupy some set of occupied bins is the same across all hash functions in a $P_{H}$-set of hash functions. These facts are defined and proven bellow.

\begin{defn}[Probability over the set of occupied bins]\label{def:DefinitionOfDistributionOverSetofOccupiedBins}
Given a hash function $H_{F}\pa{\cdot}$, the fact that passwords of users are drawn at random according to some probability mass function, induces a probability mass function over the set of occupied bins, which we denote by $P_{B}^{\pa{H_{F}}}$.
\end{defn}

\begin{lem}\label{lem:TheProbabilityoverBWhenPasswordsAreDrawnAtRandom}
Consider a $P_{H}$-set of hash functions and assume that passwords of the users are drawn uniformly and independently over the set of inputs to these hash functions. In this case we get 
\beq{}{
P_{B}^{\pa{H_{F}}}=P_{B}\quad\forall H_{F}\in\text{$P_{H}$-set of hash functions}.
}
\end{lem}
\begin{proof}
When considering the bipartite representation of every hash function in the $P_{H}$-set of hash functions, based on Definition \ref{def:FractionsofMappingsanyhash} of the statistical profile, one can see that the outputs have the same degree profile across bins. The only difference between hash functions in the set is permutations over the connections of edges to the inputs (see Remark \ref{rem:PermutationsHashFunctionsRepresentTheSet} as well as Figure \ref{fig:IntuitiveExplanationPermutation}). The set of occupied bins is essentially a set of bins, and the probability of occupying a certain set of occupied bins $B^{\ast}=\pac{b_{1},\dots,b_{K}}$ is determined by the number of passwords that hit edges that are connected to bins $\pac{b_{1},\dots,b_{K}}$. Since passwords are drawn uniformly, permutations over the edges that go to the inputs (Figure \ref{fig:IntuitiveExplanationPermutation}) lead to the same probability of hitting edges of the bins $\pac{b_{1},\dots,b_{k}}$. Therefore, $P_{B}^{\pa{H_{F}}}$ is the same for every $H_{F}\in P_{H}$-set of hash functions, and can be denoted by $P_{B}$.
\end{proof}

Essentially, Lemma \ref{lem:TheProbabilityoverBWhenPasswordsAreDrawnAtRandom} states that all hash functions in the ensemble have the same probability induced over the set of occupied bins as long as passwords are drawn uniformly and independently.

Next, we define the average guesswork, when averaging the conditional average guesswork over all sets of occupied bins. 

\begin{defn}[The average guesswork over all sets of occupied bins]\label{def:TheAverageGuessworkAveragedOverAllPassords}
The average guesswork over all sets of occupied bins is averaging $E\pa{G_{remote/insider}\pa{B}}$ over all possible sets of occupied bins $B$ through $P_{B}$. We denote this by $E\pa{G_{P}\pa{B}}$. 
\end{defn}

\begin{lem}[The average guesswork over the sets of all occupied bins]\label{lem:AverageGuessworkSetOccupiedBins}
The average guesswork over all sets of occupied bins under a remote/insider attacks, when the average for every set of occupied  bins is done according to Definition \ref{def:AveragingArgument} is
\beq{eq:AveragingOverAllPasswords}{
E\pa{G_{P}\pa{B}} = \sum_{B}E\pa{G_{remote/insider}\pa{B}}\cdot P_{B}\pa{B}.
}
\end{lem}
\begin{proof}
The proof is straight forward based on Definition \ref{def:TheAverageGuessworkAveragedOverAllPassords}.
\end{proof}

\begin{rem}
$E\pa{G_{remote/insider}\pa{B}}$ can be viewed as the conditional average guesswork conditioned on the set of occupied bins $B$, whereas  $E\pa{G_{P}\pa{B}}$ can be viewed as the average over the conditional average guesswork, that is, the average guesswork when averaging over all sets of occupied bins.
\end{rem}

The next corollary presents the average guesswork under a remote attack when averaging over all passwords/set of occupied bins as presented in equation \eqref{eq:AveragingOverAllPasswords}.
\begin{cor}[The Average Guesswork under a remote attack when averaged over all passwords/all sets of occupied bins]\label{Cor:AverageGuessworkOverAllPasswordsNoBInsAllocation}
When each user chooses its password uniformly and passwords are statistically independent, the average guesswork under a remote attack when averaging over the passwords increases at rate that is equal to
\beqn{}
{
\lim_{m\to\infty}\frac{1}{m}\log\pa{E\pa{G_{P}\pa{B}}}=1
}
where $n\ge\pa{1+\eps}\cdot m\cdot\log\pa{1/p} $.
\end{cor}
\begin{proof}[Sketch of proof]
Basically, averaging over all passwords is equivalent to averaging over all possible sets of occupied bins. Based on Lemma \ref{def:TheSetofOccupieeBins} the conditional average guesswork under a remote attack can be broken into $\Gbin$ for every user. This in turn allows us to look at the probability that a password of a user is mapped to a certain bin rather than the entire probability over the set of occupied bins $P_{B}$. When the passwords of the users are drawn i.i.d. from a uniform distribution the probability that a password is mapped to bin $b$ and its average guesswork is in turn $1/\ph$ is $\ph$; multiplying this two terms gives 1 for any $b\in\pac{1,\dots,2^{m}}$. The average is achieved by summing over all bins and so the average guesswork is $2^{m}$ as there are $2^{m}$ bins. This holds for every user and so when averaging over all users we also get that the average guesswork under a remote attack is $2^{m}$. The lower bound on the input size is required so that as $m$ increases the average guesswork converges to $1/\ph$ for every $b\in\pac{1,\dots,2^{m}}$. The full proof appears in Section \ref{sec:BoundsAverageGuessworkNoBinsAllocation}.
\end{proof}

\begin{rem}
The fact that the probability that a password is mapped to some bin is the inverse of the average guesswork of that bin, as presented in Corollary \ref{Cor:AverageGuessworkOverAllPasswordsNoBInsAllocation}, leads to average guesswork over all passwords that scales like $2^{m}$ which is larger than $2^{m\cdot H_{1/2}\pa{p}}$.
\end{rem}

Finally, we provide a concentration result for $\Gbin$ (this also holds for the set of occupied bins).

\begin{theorem}[Concentration result under both remote and insider attacks]\label{th:LowerUpperConcentrationAnyHash}
\beq{}
{
-\lim_{m\to\infty}\frac{1}{m}\log\pa{P\pa{\Gbin\le 2^{\pa{1-\eps_{1}}\cdot\pa{H\pa{q\pa{b}}+D\pa{q\pa{b}||p}}\cdot m}}}\ge \eps_{1}\cdot\log\pa{1/p},
}
where $0<\eps_{1}<1$ and bin $b$ is of type $q\pa{b}$, under the following assumptions:
\bitm
{
\item $P_{H}$ is based on the method of types as presented in Definition \ref{def:DefinitionOfIidCase}, where $p\le 1/2$.
\item $n\ge\pa{1+\eps}\cdot m\cdot \log\pa{1/p}$ where $\eps > 0$.
\item Averaging according to Definition \ref{def:AveragingArgument}.
}
\end{theorem}
\begin{proof}
The proof is in Section \ref{sec:BoundsAverageGuessworkNoBinsAllocation}.
\end{proof}

A few remarks are in order.

\begin{rem}
From Theorem \ref{th:LowerUpperConcentrationAnyHash} we can state that for any $P_{H}$-hash function and any bin, the fraction of strategies of guessing passwords one by one, for which the number of guesses is smaller than $2^{\pa{1-\eps_{1}}\cdot\pa{H\pa{q\pa{b}}+D\pa{q\pa{b}||p}}\cdot m}$, decreases like $2^{-\eps_{1}\cdot\log\pa{1/p}\cdot m}$; in addition, for the $P_{H}$-set of hash functions, any bin and any strategy of guessing passwords one by one,  the fraction of $P_{H}$-hash functions for which the number of guesses is smaller than $2^{\pa{1-\eps_{1}}\cdot\pa{H\pa{q\pa{b}}+D\pa{q\pa{b}||p}}\cdot m}$, also decreases like $2^{-\eps_{1}\cdot\log\pa{1/p}\cdot m}$.
\end{rem}

\begin{rem}
The concentration result of Theorem \ref{th:LowerUpperConcentrationAnyHash} shows that the probability mass function of $\Gbin$ is concentrated around its mean value $E\pa{\Gbin}$. This is in contrast with classic average guesswork for guessing a password \cite{Arikan_Ineq_Guessing} in which case the probability mass function is concentrated around the typical set in the i.i.d. case, whereas the average guesswork can be derived based on a large deviations argument \cite{MaloneGuessworkandEntropy}. On the other hand, when averaging over all passwords as in Corollary \ref{Cor:AverageGuessworkOverAllPasswordsNoBInsAllocation} we get that the average guesswork $E\pa{G_{P}\pa{B}}$ lies in the tail of the probability mass function similarly to the classic case of \cite{Arikan_Ineq_Guessing}.
\end{rem}

\subsection{Large Deviation Vs. Concentration: Another Difference From Classic Guesswork}\label{subsec:LargeDeviationVsConcentration}
Theorem \ref{th:LowerUpperConcentrationAnyHash} along with Theorem \ref{th:AverageGuessworkPerBinIidAssumption} show that the probability mass function of $\Gbin$ is concentrated around its mean value. This is in contrast with classic guesswork \cite{Arikan_Ineq_Guessing} where the average guesswork is based on large deviations (i.e., the average lies in the tail of the probability mass function).  

In this subsection we explain the mechanism that leads to concentration of the probability mass function of $\Gbin$ around its mean value, and highlight the differences between this mechanism and the underlying mechanism that enable large deviations to dominate in classic average guesswork (i.e., directly guessing a password).

We begin by explaining why the probability mass function of $\Gbin$ is centered around its mean value $E\pa{\Gbin}$. For this let us consider a bin $b$ of type $q$.

In general the maximum number of guesses is $2^{n} >> 2^{m}$ (i.e., generally, the set of inputs of a hash function is much larger than the set of outputs). We first focus on the case where the number of guesses $1\le i_{g}\le 2^{\log\pa{1/p}\cdot m}$ for which
\beq{eq:ClassicGwVsGbin190921}{
P\pa{\Gbin = i_{g}}\approx \pa{1-p_{g}}^{i_{g}}\cdot p_{g},
}
where $p_{g}=2^{-\pa{D\pa{q|p}+H\pa{q}}\cdot m}$, when $n\ge\pa{1+\eps}\cdot m\cdot \log\pa{1/p}$ and $\eps > 0$; this is based on Theorem \ref{cor:LowerUpperTheAverageGuessworkForGuessingAnyPassword} and Theorem \ref{th:LowerUpperBoundExpAverageGuessworkAnyHashFunction} that show that in this regime the underlying probability of $\Gbin$ converges to geometric distribution. We later address the case where $i_{g} > 2^{\pa{1+\eps}\cdot m\cdot \log\pa{1/p}}$.

Theorem \ref{th:LowerUpperConcentrationAnyHash} shows that the probability that the number of guesses $i_{g} < E\pa{\Gbin}$ goes to $0$. On the other hand, when the number of guesses is in the order of $1/p_{g}=2^{\pa{D\pa{q|p}+H\pa{q}}\cdot m}$, that is, when $i_{g}\approx E\pa{\Gbin}$ we get that the first term in
\eqref{eq:ClassicGwVsGbin190921}, $\pa{1-p_{g}}^{i_{g}}$ goes to a constant and so in this case $P\pa{\Gbin\approx E\pa{G_{bin}}}\approx p_{g}$, where $\Gbin\approx E\pa{\Gbin}$ means that the number of guesses is in the exponential order of the mean value. 

When, $E\pa{\Gbin} < i_{g} \le 2^{\pa{1+\eps}\cdot m\cdot \log\pa{1/p}}$ we get that the first terms in \eqref{eq:ClassicGwVsGbin190921} decays super exponentially. In particular, when $i_{g} = 2^{\pa{D\pa{q|p}+H\pa{q}+\delta}\cdot m}$, where $\delta > 0$ we get that
\beq{}{
\pa{1-p_{g}}^{i_{g}}\approx 2^{-m\cdot 2^{\delta\cdot m}}.} 
This means that the probability of large deviations in the number of guesses decreases super exponentially, whereas the number of guesses can only increase exponentially. This leads to the fact that large deviations are extremely unlikely in this case, and as a result the probability mass function of $\Gbin$ is centered around $E\pa{\Gbin}$.

The underlying mechanism that leads to super exponential decrease in the probability of large deviations is directly related to the $P_{H}$-ensemble of hash functions (\textbf{the fact that there is a super exponential number of such hash functions}), and in particular to the bipartite graph presented in Figure \ref{fig:IntuitiveExplanationPermutation}. When $\Gbin= i_{g}$, it means that there are no edges in the bipartite graph that are connected to bin $b$ and that are also connected to any of the first $i_{g}-1$ guessed inputs. When $n$ is large enough, and $i_{g}>1/p_{g}$ the probability of this event decays super exponentially, as for every guess you have a new draw of edges from bin $b$ that can be mapped to the guessed input, the probability that an edge from bin $b$ hits the guessed input decreases exponentially, and the number of guesses that have been made at this point is also exponential.

When $i_{g}$ is in the order of $2^{n}$ the super exponential decay has even higher rate than what is presented above, due to the same mechanism as the one discussed in the previous paragraph.

On the other hand, in classic guesswork, the probability of large deviation in the number of guesses under ideal strategy is simply the probability of drawing a password that lies in the tail of the probability mass function of passwords. Since this probability decreases exponentially, and the number of guesses increases exponentially, the average guesswork in the classic case can lie in the tail of its probability mass function.

\section{The Most Likely Conditional Average Guesswork of type-class statistical profile hash functions}\label{sec:MostLikelyAverageGuessworkPhHashFunctionsNoMod}
In this section we define and analyze the most likely conditional average guesswork under both remote and insider attacks.

The results of this section conbined with the results of Section \ref{sec:AverageGuessworkPhHashFunctionsNoMod} provide quantifiable bounds for the effect of bias as well as the effect of the number of users, on the conditional average guesswork of a hash function. Furthermore, these results show that increasing the number of stored hashed passwords of users has a far greater effect on the average guesswork than bias.

The next definitions present the most likely conditional average guesswork, based on the fact that the conditional average guesswork is a function of the set of occupied bins $B$, and the fact that drawing passwords at random induces certain probability on the set of occupied bins, $P_{B}$, as presented in Definition \ref{def:DefinitionOfDistributionOverSetofOccupiedBins} and Lemma \ref{lem:TheProbabilityoverBWhenPasswordsAreDrawnAtRandom}. 

Essentially, some sets of occupied bins have the same conditional average 
guesswork; given a certain conditional average guesswork value, we are interested in the probability of  drawing a set of occupied bins whose conditional average guesswork is equal to that value, that is, the probability of drawing a set of occupied bins $B_{0}$, whose conditional average guesswork is $E\pa{G_{remote/insider}\pa{B_{0}}}=a$. In particular, we are interested in the conditional average guesswork whose value is the most likely in terms of drawing a set of occupied bins that achieves this value.

\begin{defn}[The most likely conditional average guesswork]\label{def:TheMostLikelyAverageGuesswork}
Given that the number of users is $|B|$, the most likely conditional average guesswork under a remote/insider attack is the  most probable average number of guesses required to break into the system; the probability of a conditional average number of guesses is determined by  $P_{B}$, and the sets of occupied bins whose conditional average guesswork is equal to this number of guesses. \end{defn}

For analysis simplicity, we focus on sets of occupied bins that consist of \textbf{unique} elements; we show in Corollary \ref{cor:NoBinAllocationMostLikelyAverageGuesswork} and Corollary \ref{cor:NoBinAllocationMostLikelyAverageGuessworkOnLine} that when passwords are drawn uniformly and independently it is very likely that the set of occupied bins consists of unique elements.

\begin{lem}[The most likely conditional average guesswork over sets of occupied bins consisting of unique elements]\label{lem:TheMostLikelyAverageGuesswork}
The most likely conditional average guesswork as a function of the sets of occupied bins of size $|B|=A$ that consist of unique elements is 
\beq{eq:ThemostLikelyaverageGiesswprk2}{
\arg\max_{a\ge 0}\sum_{B_{0}:E\pa{G_{remote/insider}\pa{B_{0}}}=a, |B_{0}|=A, B_{0}\text{ consists of unique elements}}P_{B}\pa{B_{0}}
}
when averaging the guesswork according to definition \ref{def:AveragingArgument}, and passwords are drawn at random according to a probability mass function that induces $P_{B}$ as presented in Definition \ref{def:DefinitionOfDistributionOverSetofOccupiedBins}.
\end{lem}
\begin{proof}
Basically the conditional average guesswork is a function of the set of occupied bins $B$. Furthermore, certain sets of occupied bins can have the same conditional average guesswork. Therefore, it is a question of what is the most probable group (or set) of sets of occupied bins whose guesswork is the same. This leads to equation \eqref{eq:ThemostLikelyaverageGiesswprk2}. 
\end{proof}

\begin{rem}
$|B_{0}|=A$ in Lemma \ref{lem:TheMostLikelyAverageGuesswork} represents the number of users. Therefore, the most likely average guesswork depends on $A$, as shown in the expressions in Corollary \ref{cor:NoBinAllocationMostLikelyAverageGuesswork} and Corollary \ref{cor:NoBinAllocationMostLikelyAverageGuessworkOnLine}.
\end{rem}

Figure \ref{fig:ExampleMostLikelyAverageGuesswork} illustrates the concept of the most likely conditional average guesswork.

\begin{figure}[!h]
\centering\scalebox{.8}{\input{Most_Likely_Average_Guesswork.TpX}}\caption{Assume that $B_{1}$, $B_{2}$, and $B_{3}$ are the only sets of occupied bins whose average guesswork is $E\pa{G_{remote/insider}\pa{B_{1}}}=E\pa{G_{remote/insider}\pa{B_{2}}}=E\pa{G_{remote/insider}\pa{B_{3}}}=a$, and also that $|B_{1}|=|B_{2}|=|B_{3}|=A$. In this case when $P_{B}\pa{B_{1}}+P_{B}\pa{B_{2}}+P_{B}\pa{B_{3}}$ is maximum compared to the probability of  other groups of occupied bins of size $A$, where the elements in each group have the same conditional average guesswork (not $a$), we get that the most likely average guesswork is $a$ and the probability of this guesswork is $P_{B}\pa{B_{1}}+P_{B}\pa{B_{2}}+P_{B}\pa{B_{3}}$.    }\label{fig:ExampleMostLikelyAverageGuesswork}
\end{figure}

We now find the most likely conditional average guesswork as presented in Lemma \ref{lem:TheMostLikelyAverageGuesswork} when the elements in the set of occupied bins are unique. Again, for representation purposes we denote the number of users $|B|=2^{m\cdot H\pa{s}}$, where $1/2\le s\le 1$.
\begin{cor}[The most likely conditional average guesswork under an insider attacks]\label{cor:NoBinAllocationMostLikelyAverageGuesswork}
When the number of users is $|B|=2^{H\pa{s}\cdot m}$, the most likely conditional average guesswork under an insider attack as presented in Lemma \ref{lem:TheMostLikelyAverageGuesswork} increases like
\beq{}{
2^{m\cdot\pa{H\pa{p}-H\pa{s}}},
}
as long as $0\le 1-s\le p$ and $p\le 1/2$. In addition, we get the following:
\bitm
{
\item The probability that all passwords are mapped to different bins of type $q$ such that $0\le 1-s< q\le p$, decreases like $$e^{-2^{m\cdot\pa{2\cdot H\pa{s}-H\pa{q}}}}\times 2^{-m\cdot D\pa{q||p}\cdot 2^{H\pa{s}\cdot m}}.$$
\item In this case the average guesswork $E\pa{\Gof}$ increases like  $2^{m\cdot\pa{H\pa{q}+D\pa{q||p}-H\pa{s}}}$.
\item The most likely conditional average guesswork is achieved when $q=p$ (with probability that goes to 1).
\item When $p\le 1-s\le 1/2$ the most likely conditional average guesswork is again achieved when $q=p$; in this case the exponent of the conditional average guesswork is equal to $0$.
}

This is achieved under the following assumptions:
\bitm
{
\item Passwords are drawn i.i.d. Bernoulli$\pa{1/2}$, and the elements in $B$ are unique.
\item $P_{H}$ defined based on the method of types and presented in Definition \ref{def:DefinitionOfIidCase}.
\item $n\ge\pa{1+\eps}\cdot m\cdot \log\pa{1/p}$ where $\eps > 0$.
}
Finally, we average based on Definition \ref{def:AveragingArgument}.
\end{cor}
\begin{proof}[Sketch of proof]
The general idea behind the proof is that the guesswork is a function of the set of occupied bins $B$. Based on Theorem \ref{th:AverageGuessworkPerBinIidAssumption} it can be shown that for every set of occupied bins that are in the typical set of $P_{H}$, the average guesswork is the same. Using Lemma \ref{th:methodoftypesprobth} and the method of types it can be shown that for the typical set $q=p$ we get $\ph=2^{-m\cdot H\pa{p}}$, in which case the average guesswork increases like $2^{m\cdot\pa{H\pa{p}-H\pa{s}}}$. Furthermore, when $|B|\le 2^{H\pa{p}\cdot m}$ and passwords are drawn i.i.d. Bernoulli$\pa{1/2}$ the probability that the set of occupied bins is in the typical set of $P_{H}$ goes to 1. The full proof appears in Section \ref{sec:BoundsAverageGuessworkNoBinsAllocation}.
\end{proof}

\begin{rem}
Corollary \ref{cor:NoBinAllocationMostLikelyAverageGuesswork} shows that when $0\le 1-s\le p$ it is most likely that $B$ consists of unique bins of type $p$ (with probability that goes to 1), in which case the average guesswork increases like $2^{m\cdot\pa{H\pa{p}-H\pa{s}}}$.
\end{rem}

\begin{rem}[The effect of bias vs. the effect of the number of users]\label{rem:EffectofBiasonAverageGuesswork}
The effect of the number of users on the average guesswork is far greater than the effect of bias, that is, as the number of users increases the average guesswork decreases at a higher rate than the one when the number of users remains constant and bias increases. In order to illustrate this statement let us focus on the rate at which the most likely average guesswork increases, $H\pa{p}-H\pa{s}$, where $0\le 1-s \le p$ (although it holds for the bounds of Theorem \ref{cor:LowerUpperTheAverageGuessworkForGuessingAnyPassword} as well).
First, let us focus on the effect of increasing the number of users on the average guesswork.
The first derivative of $H\pa{p}-H\pa{s}$ with respect to $p$ is
\beq{}{
\log_{2}\pa{p}-\log_{2}\pa{1-p}
}
which is equal to zero at $p=1/2$ (i.e., when there is no bias). In addition, the first derivative around $p=1/2$ is very small and therefore bias has a little effect on the average guesswork. On the other hand the rate at which the number of users increases is $H\pa{s}$, and so the average guesswork decreases linearly with $H\pa{s}$. This in turn shows that change in bias does not affect the average guesswork to the same extent as change in the number of users. We illustrate this observation in Table \ref{table:guess}.
\end{rem}

\begin{table}
\centering
\caption{The rate at which the most likely average guesswork increases along with the lower and upper bounds. Note that at $p=1/2$ the bounds meet. Furthermore, this example illustrates that increasing the number of users has a far greater effect than bias as stated is Remark \ref{rem:EffectofBiasonAverageGuesswork}. When there is no bias and the number of users increases like $2^{m\cdot H\pa{0.2}}$ the average guesswork increases at the same rate as in the case when the number of users is $2^{m\cdot H\pa{0.1}}$ and the bias is $p=0.21$. Furthermore, when the bias is $p=0.45$ the average guesswork is very close to one.}
\begin{tabular}{|c|c|c|c|}
\hline
$p$, $s$  &$ H\pa{p}-H\pa{s}$ & $D\pa{s||p}$ & $D\pa{1-s||p}$ \\
\hline
p=1/2, $s=0$ & $1$ & $1$ & $1$ \\
\hline
$p=0.45$, $s=0$ & $0.9948$ & $0.8625$ & $1.15$ \\
\hline
$p=1/2$, $s=0.2$ & $0.2781$  & $0.2781$ & $0.2781$  \\
\hline
$p=0.21$, $s=0.1$ & $0.2725$ & $0.0622$ & $1.5914$ \\
\hline

\end{tabular}
\label{table:guess}
\end{table}

Next, we find the most likely average guesswork under remote attacks.
\begin{cor}[The most likely conditional average guesswork under a remote attack]\label{cor:NoBinAllocationMostLikelyAverageGuessworkOnLine}
When the number of users is $|B|=2^{H\pa{s}\cdot m}$, the most likely conditional average guesswork under a remote attack as presented in Lemma \ref{lem:TheMostLikelyAverageGuesswork} increases like
\beq{}{
2^{m\cdot H\pa{p}},
}
as long as $0\le 1-s\le p$ and $p\le 1/2$. In addition, we get the following:
\bitm
{
\item The probability that all passwords are mapped to different bins of type $q$ such that $0\le 1-s< q\le p$, decreases like $$e^{-2^{m\cdot\pa{2\cdot H\pa{s}-H\pa{q}}}}\times 2^{-m\cdot D\pa{q||p}\cdot 2^{H\pa{s}\cdot m}}.$$
\item In this case  the average guesswork of any of the users $E\pa{\Gbin}$ as well as $E\pa{\Gon}$ increase like  $2^{m\cdot \pa{H\pa{q}+D\pa{q||p}}}$.
\item The most likely conditional average guesswork is achieved when $q=p$ (with probability that goes to 1).
}
This is achieved under the following assumptions:
\bitm
{
\item Passwords are drawn i.i.d. Bernoulli$\pa{1/2}$, and the elements in $B$ are unique.
\item $P_{H}$ is based on the method of types as Defined in Definition \ref{def:DefinitionOfIidCase}.
\item $n\ge\pa{1+\eps}\cdot m\cdot \log\pa{1/p}$ where $\eps > 0$.
}
Finally, the average is done in accordance with Definition \ref{def:AveragingArgument}.
\end{cor}
\begin{proof}
The full proof appears in Section \ref{sec:BoundsAverageGuessworkNoBinsAllocation}.
\end{proof}

\section{Overview of the Results}\label{sec:RestiltsOutline}
In this section we outline the main results presented in the paper and  reference the relevant theorems and definitions in the paper, so that the motivation and the essence of the results can stand out. 

\subsection{Hashing Passwords, Brute Force Attacks, and Bias}

Modern systems do not store the passwords of their users in plain text but rather hash them (a many to one mapping) and store their hash value, which is  paired with a user name, as described in Figure \ref{fig:Problem_Setting} and presented in Definition \ref{def:GivingUserAnAccess}. This in turn protects the passwords of users when the system is compromised. In order to gain access to a system a user provides the system with his user name and his password; the system then calculates its hash value and compares it to the one stored in the system.  

The set of bins to which the passwords of the users are mapped is termed \textbf{the set of occupied bins} and denoted by $B$. This is presented in Definition \ref{def:TheSetofOccupieeBins} and Figure \ref{fig:ExampleForOccupiedBins}.

In a brute force attack an adversary who \textbf{does not know the structure of the hash function}, and therefore does not know which inputs map to which bins, guesses inputs one by one in order to find \textbf{an input} that is mapped to an element in $B$ (one of the hash values stored in the system), that is, a hash value of a password of a valid user. This sort of attacks deviate away from classic works on guesswork, which consider attacks where an attacker has to find \textbf{the password} of a user (go to Subsection \ref{subsec:BackgroundClassicGuesswork} for background on classic guesswork). The reason for this is that in password cracking the attacker does not know the hash function and so cannot directly guess specific bin values, but rather guess inputs, and it does not know to what outputs they are mapped. Indeed, in Theorem \ref{th:AverageGuessworkAnyHashFunction} we analyze the case where the attacker does know the hash function and also therefore knows the inputs mapping to each bin, thereby it can guess bins directly (i.e., cracking passwords over a broken hash function); in this case we show that the average guesswork is very similar to classic guesswork. 

In this work we consider two types of brute force attacks: A remote attack and an insider attack. These attacks are presented in Figure \ref{fig:TheDifferenceBetweenOnlineAndOfflineAttacks} and defined in Definition \ref{def:AnOnlineAttackModel} and Definition \ref{def:AnOfflineAttackModel} respectively. In a remote attack the adversary targets a specific user and has to find an input that is mapped to the element in $B$ that is paired with that user (i.e., the hash value of his password). Whereas under an insider attack  an adversary (who knows the set of user-names and hashed passwords) has to find an input whose hash value matches any of the hash values that are stored in the system, that is, it is not targeting a specific user. Clearly, the chance of success of an insider attack is higher than a remote attack; the reason is that given a set of occupied bins the attacker has to find an input that is mapped to any of the bins in this set, whereas in a remote attack he has to target a certain bin in the set that is paired with a user name. This is because for an insider attack, it just needs one password that maps to any of the known hashed values, then as it knows all user-name and hashed passwords, it has compromised the system. For a remote attacker, it needs to find a password that maps to a \emph{specific} bin (or hashed value) that is associated with a user-name.

We now discuss bias in hash functions. A hash function maps a large set of inputs to a set of smaller size, where the elements in the output are termed \textbf{bins}. When it comes to cryptographic hash functions and in particular password hashing functions, it is assumed for every bin that it is hard to find an input that is mapped to a bin. One of the most basic assumptions to support this is that a hash function is \textbf{not biased}, that is, the same number of inputs is mapped to any of the bins. However, in practical cryptographic hash functions it has never been proven that this assumption holds.

In this work we represent bias based on the statistical profile of a hash function. As presented in Figure \ref{fig:HashFunctionBipartiteRepresentation} a hash function can be represented by a bipartite graph. The degree distribution \cite{UrbankeBook2008} of the bins (the outputs) is the statistical profile of the hash function. The reason for this is that the set of degrees of the outputs can be normalized to a probability mass function as shown in Definition \ref{def:FractionsofMappingsanyhash}. 

When there is no bias and the set of bins is of size $2^{m}$, we get that the number of mappings per bin is the same across all bins and so the statistical profile is a uniform PMF, that is, $2^{-m}$ for any $b\in\pac{1,\dots,2^{m}}$.

We will use the statistical profile of a hash function to calculate the average guesswork per bin as well as under remote and insider attacks.

\subsection{Defining Average 
Guesswork for Hash Functions}

In this work we derive the average guesswork for every bin as well  as the average guesswork under remote and insider attacks, under certain assumptions about the statistical profile of a hash function. Therefore, a question that may come to mind is what is the source of randomness over which we average, that is, what are the \textbf{averaging arguments}. 

The first source of randomness is the probability according to which users choose their passwords; the second source of randomness can be either a key in a keyed hash function or alternatively the strategy according to which an attacker guesses passwords one by one; this two sources of randomness are presented in Definition \ref{def:AveragingArgument}.

First, we assume that users draw their passwords independently and uniformly over the set of inputs. This assumption can be partially supported by the use of \textbf{passwords managers}. Furthermore, in Subsection \ref{subsec:TheCaseWhereoasswordsAreBiased} we analyze the case where passwords are biased.

The first averaging argument says that we average over all hash functions of a certain degree distribution, where the degree of each bin is fixed (i.e., permutation over the connection of edges to inputs in the bipartite graph as illustrated in Figure \ref{fig:IntuitiveExplanationPermutation}). Essentially, this is the same as considering a keyed hash function that is biased, where all elements in the set of hash functions have the same statistical profile. 

The second averaging argument (which provides an alternative to the first one) considers a hash function with a given statistical profile and averages over all strategies of guessing passwords one by one, that is, over all brute force attacks. We wish to point out that this averaging method is meaningful when passwords are drawn uniformly, however it falls short when passwords are biased (as in this case an attacker is likely to favor a certain strategy over the others). However, this averaging argument comes handy when considering guessing using  nonce in the context of proof of work. In this case the attacker draws an input at random (nonce) at every step as an input to the hash function. This is equivalent to drawing a strategy for guessing passwords one by one.


Under these averaging arguments we present in Definition \ref{def:TheGuessworkOfACertainBin} the average guesswork of a bin  $\Gbin$. Basically, when there is bias, this average varies across bins as a function of the statistical profile (i.e., as a function of $b$). In Theorem \ref{th:TheAverageGuessworkPerBinGeneralExpressions} we show that when there are $2^{m}$ bins and the normalized degree (number of inputs mapping to it normalized by the total number of inputs)  of of bin $b_{0}$ is $2^{-\Al\cdot m}$, where $\Al>0$, the average guesswork of this bin $E\pa{G_{bin}\pa{b_{0}}}=2^{\Al\cdot m}$, which is the average number of guesses to find the input mapping to this particular bin. Note that the underlying probability mass function of $\Gbin$ as well as its average are not dependent on the passwords or the probability according to which they are drawn, but rather on the structure (statistical profile) of the hash function.

Since the attacker does not know the hash function, and there is bias (non uniform statistical profile), it is meaningful in our problem to define the conditional average guesswork, conditioned on the set of occupied bins $B$. Essentially, the conditional average guesswork changes as a function of the set of occupied bins due to bias. In this case users draw their passwords uniformly and these passwords are mapped to a certain set of occupied bins $B$. Another measure that we analyze is the average guesswork (i.e., the unconditional average guesswork), which is basically the average over the conditional average guesswork  and the average is taken over the set of occupied bins.

In terms of the conditional average guesswork under remote and insider attacks we answer the following questions under the assumption that the set of occupied bins consists of $A$ unique elements \emph{i.e.,} $|B|=A$, and a given statistical profile:
\begin{itemize}
    \item What is the set of occupied bins $B$ whose conditional average guesswork is maximum; what is the maximum conditional average guesswork (see Definition \ref{def:MaxMinAverageGiesswork}).
    \item What is the set $B$ whose conditional average guesswork is minimum; what is the minimum conditional average guesswork (see Definition \ref{def:MaxMinAverageGiesswork}).
    \item What is the most likely conditional average guesswork, that is, what is the most likely group of sets of occupied bins that have the same conditional average guesswork, and what is their conditional average guesswork (see Definition \ref{def:TheMostLikelyAverageGuesswork} and Figure \ref{fig:ExampleMostLikelyAverageGuesswork}).
\end{itemize}

Furthermore, when considering the average over the passwords that users choose or alternatively averaging over all possible sets of occupied bins, we answer the following question:
\begin{itemize}
    \item What is the average guesswork when averaging over all possible sets of bins $B$. This is presented in Definition \ref{def:TheAverageGuessworkAveragedOverAllPassords}.
\end{itemize}

\subsection{The average Guesswork under the type-class statistical profile}

In order to derive closed form expressions for the average guesswork under remote and insider attacks we model bias (statistical profile) by using the distribution of type-classes from the method of types. We assume a natural mapping from bins to the statistical profile, that is, when considering the binary representation of bin $b$ we get  
\beq{eq:NatualMappingBitsToPmf}{\ph = 2^{-m\cdot \pa{H\pa{q}+D\pa{q||p}}},} 
where $0<p<1/2$ and $q=(\# \text{ ones in b})/m$. This means that the number of edges of bin $b$ divided by the total number of edges in the bipartite graph is $\ph$.

For type-class statistical profile, given in  equation \eqref{eq:NatualMappingBitsToPmf}, we get that bins of the same type have the same degree/probability (at least in terms of their exponent). This in turn means that under equation \eqref{eq:NatualMappingBitsToPmf} we get sets of bins that are biased and are equally likely. In particular, every bin in the typical set has normalized degree that decreases like $2^{-m\cdot H\pa{p}}$ and there are $2^{m\cdot H\pa{p}}$ such bins. Therefore, under this assumption we get that the hash function is biased towards a certain subset of bins, whose elements are (approximately) uniformly distributed within the set. This is presented in Lemma \ref{th:methodoftypesprobth} and Lemma \ref{th:methodoftypessize}.

We wish to point out that under both remote and insider attacks, the average guesswork does not change when considering permutation of the degree distribution; the average guesswork $E\pa{G_{remote/insider}\pa{B}}$ remains the same up to permutation over the set of occupied bins.  In terms of the example given in Figure \ref{fig:HashFunctionBipartiteRepresentation} it means that considering the degree distribution $\pac{3/22,3/22,3/22,\emph{1/22},3/22,3/22,3/22}$ and the degree distribution $\pac{\emph{1/22},3/22,3/22,3/22,3/22,3/22,3/22}$ leads to the same average guesswork under remote and insider attacks up to a permutation over the set of occupied bins that achieves this average guesswork. Hence, the assumption made in equation \eqref{eq:NatualMappingBitsToPmf} is not very restrictive. 

We begin by presenting the average guesswork per bin $E\pa{\Gbin}$ for a type-class statistical profile. In this case, when bin $b$ is of type $q$ and the input size is large enough we get that $E\pa{\Gbin}$ increases like $1/\ph=2^{m\cdot \pa{H\pa{q}+D\pa{q||p}}}$ for large $m$, which is nothing but the average over geometric distribution, where the probability of success is $\ph$. This is presented in Theorem \ref{th:AverageGuessworkPerBinIidAssumption}.

Next, we focus on the conditional average guesswork under an insider attack in order to highlight our results on the most likely conditional average guesswork as well as the maximum and minimum conditional average guesswork. For all of this cases we assume that the number of users is $2^{R\cdot m}=2^{m\cdot H\pa{s}}$, where $1/2\le s\le 1$, and $|B|=2^{R\cdot m}$. Furthermore, we assume that the passwords of the users are mapped to unique bins\footnote{We assume that users are mapped to unique bins for simplicity. The expressions can be modified for the case where they are not mapped to unique bins; essentially, it only affects the expressions for guesswork under an insider attack as the set of occupied bins decreases. In addition, as shown in Corollary \ref{cor:NoBinAllocationMostLikelyAverageGuesswork} and Corollary \ref{cor:NoBinAllocationMostLikelyAverageGuessworkOnLine}, for any practical purpose it is very likely that the passwords are mapped to unique bins.}.

The most likely conditional average guesswork in this case increases like $2^{m\cdot\pa{H\pa{p}-R}}$, when $0\le R\le H\pa{p}$. Furthermore, when $0\le R\le H\pa{p}/2$ the probability for the most likely conditional average guesswork goes to 1 (the probability presented in equation \eqref{eq:ThemostLikelyaverageGiesswprk2}), that is, the probability that all elements in $B$ are unique and the average guesswork is $2^{m\cdot\pa{H\pa{p}-R}}$ goes to 1. This result is presented in Corollary \ref{cor:NoBinAllocationMostLikelyAverageGuesswork}. Note that Corollary \ref{cor:NoBinAllocationMostLikelyAverageGuesswork} also shows that the probability of hitting the most likely average guesswork scales super exponentially. 

In Remark \ref{rem:EffectofBiasonAverageGuesswork} we show that the effect of increasing $R$ (i.e., the number of users) is far greater than increasing bias (when $p$ deviates away from $1/2$) in terms of the exponent of the average guesswork.

We now present the maximum conditional average guesswork under an insider attack. In this case the conditional average guesswork increases like $2^{m\cdot D\pa{s||p}}$. This is achieved when passwords are mapped to the $2^{m\cdot R}=2^{m\cdot H\pa{s}}$ bins with lowest degrees, where $1/2\le s\le 1$. Interestingly enough, as $p$ decreases basically the exponent $D\pa{s||p}$  increases unboundedly as it becomes less likely for an attacker to hit these bins. We present this result in Theorem \ref{cor:LowerUpperTheAverageGuessworkForGuessingAnyPassword}. Note that the chance of drawing a set of occupied bins of size $|B|=2^{R\cdot m}$ that is mapped to the least likely bins decreases super exponentially in $m$. 

The minimum conditional average guesswork when $|B|=2^{R\cdot m}$ is also presented in Theorem \ref{cor:LowerUpperTheAverageGuessworkForGuessingAnyPassword}. In this case as long as $0\le R\le H\pa{p}$ the conditional average guesswork under an insider attack increases like $2^{m\cdot D\pa{1-s||p}}$, where again $R=H\pa{s}$. This function \textbf{does not} increase unboundedly as $p$ decreases.   

We wish to point out that in this subsection we do not average over all possible sets of occupied bins (or alternatively over the passwords), but rather use the averaging arguments of Definition \ref{def:AveragingArgument}. 

\subsection{How is it Different from Classical Results on Guesswork?}

We now consider the case where averaging over all sets of occupied bins or alternatively averaging over the conditional average guesswork as presented in Definition \ref{def:TheAverageGuessworkAveragedOverAllPassords} and Lemma \ref{lem:AverageGuessworkSetOccupiedBins} and denoted by $E\pa{G_{P}\pa{B}}$. Note that this is equivalent to averaging over all passwords.

We focus on the average guesswork under a remote attack. In this case the conditional average guesswork increases like $2^{H\pa{p}\cdot m}$ (this is presented in Corollary \ref{cor:NoBinAllocationMostLikelyAverageGuessworkOnLine}), when every element of the set of occupied bins is in the typical set. When $R<H\pa{p}/2$ we get that the probability that every bin of the set of occupied bins is in the typical set and is unique goes to 1.

Furthermore, the results of Theorem \ref{th:LowerUpperConcentrationAnyHash} show that the probability that the conditional average guesswork is smaller than $2^{m\cdot H\pa{p}}$ goes to zero. Indeed, the conditional average guesswork is concentrated around its mean value (i.e., concentrated around $2^{H\pa{p}\cdot m}$), when every element in the set of occupied bins is in the typical set and $0\le R<H\pa{p}/2$. This is in contrast with classical results on average guesswork (see \cite{Arikan_Ineq_Guessing} as well as Subsection \ref{subsec:BackgroundClassicGuesswork}), where it is shown that the average guesswork, when guessing a secret directly, lies in the tail of the probability function of guesswork, which means that in the classic case the probability mass function is not concentrated around its mean value

Subsection \ref{subsec:LargeDeviationVsConcentration} presents the mechanism that leads to concentration of the probabilty function of $\Gbin$. Essentially, it shows that it is related to the fact that the tail of the probability function decays super-exponentially because of the number of hash functions in the $P_{H}$-set of hash functions, which is super-exponential.

On the other hand, when averaging over all passwords (basically over all sets of occupied bins) and passwords are drawn uniformly i.i.d. we get that the average guesswork $E\pa{G_{P}\pa{B}}$ increases like $2^{m}$ \textbf{which is exactly the average guesswork when there is no bias.} This result is presented in Corollary \ref{Cor:AverageGuessworkOverAllPasswordsNoBInsAllocation}.

Therefore, we say that $E\pa{G_{P}\pa{B}}$ is a measure that is too crude to quantify the effect of bias on the average guesswork of hash functions. In this sense the conditional average guesswork presented in the previous subsection captures the effect of bias better than averaging over all elements including the set of occupied bins. 
It also means that the when $0\le R\le H\pa{p}/2$, the random variable $G_{P}\pa{B}$  is concentrated around $2^{H\pa{p}\cdot m}$, whereas its average lies in the tail of its probability mass function. This is similar to results on the average guesswork in the classic case \cite{YonaForensicsPUFS}.
This is also illustrated in Figure \ref{fig:ConcentrationVsNonConcentrationaverageGuesswork}.

\begin{figure}
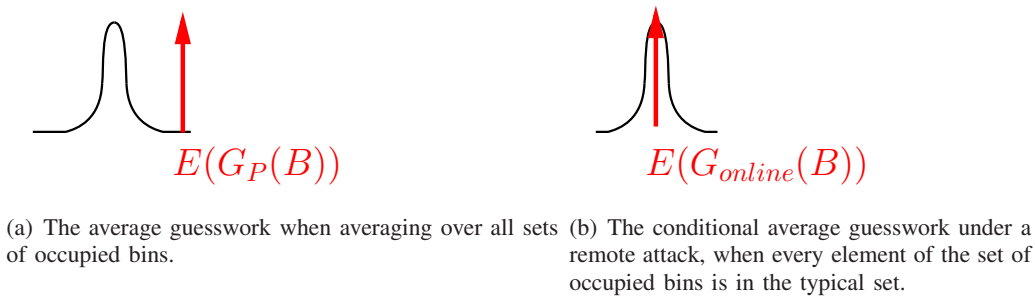

\centering\subfigure[The average guesswork when averaging over all sets of occupied bins.]{
\input{Av_Ps.TpX}} \centering\subfigure[The conditional average guesswork under a remote attack, when every element of the set of occupied bins is in the typical set.]{
\input{Concentration.TpX}} \caption{\textbf{On the right hand side:} The conditional average guesswork. When $0\le R < H\pa{p}/2$ the probability that the conditional average guesswork is $2^{H\pa{p}\cdot m}$ goes to 1. Therefore, the probability mass function of $G_{P}\pa{B}$ is concentrated around $2^{H\pa{p}\cdot m}$. \textbf{On the left hand side:} When averaging over all passwords/sets of occupied bins we get that $E\pa{G_{P}\pa{B}}$ lies in the tail of its probability mass function.}\label{fig:ConcentrationVsNonConcentrationaverageGuesswork}
\end{figure}

\subsection{Using Bias to Increase the Average Guesswork}
Another question that may come to mind is whether bias can be used to increase the average guesswork of hash functions. As we show in Theorem \ref{th:LowerBoundExpAverageGuessworkAnyHashFunction} indeed bias can increase the average guesswork as long as there is a backdoor mechanism (such as the one presented in Definition \ref{def:BackdoorForAnyHash}) that maps password to the least likely set of occupied bins. In this case the average guesswork under a remote attack increases like $2^{\pa{2\cdot H\pa{p}+D\pa{1-p||p}-R}\cdot m}$; this function increases unboundedly as $p$ decreases. This is also illustrated in Figure \ref{fig:Intuitive_Explanation}. We wish to point out that as $p$ decreases, the size of inputs to the hash functions required for this result to hold also increases, and so there is no contradiction in the fact that the average guesswork increases unboundedly as $2^{n}>> 2^{\pa{2\cdot H\pa{p}+D\pa{1-p||p}-R}\cdot m}$, that is, the set of inputs to the hash function is much larger than the average guesswork.

\section{The Average Guesswork of a Modified $P_{H}$-Hash Function}\label{sec:ModifiedHashFunctionAverageGueddwork}
In this section we again analyze the average guesswork of the $P_{H}$-hash function of Section \ref{sec:AverageGuessworkPhHashFunctionsNoMod} with the exception that a centralized procedure maps the passwords of the users to the least likely bins. We show that in this case the average guesswork can increase unboundedly as a function of bias.

We begin by defining a centralized procedure that maps the password of the users to the least likely bins. We then find the average guesswork under both a remote and an insider attacks, when passwords are drawn uniformly. Finally, we quantify how bias in drawing the passwords of the users affects the average guesswork. 

\begin{defn}[bin allocation]\label{def:PasswordAllocationwithNoDetails}
Assume that the number of users $M=2^{m\cdot H\pa{s}-1}\le 2^{m}$, the bin allocation method is defined as follows. For each user a procedure allocates a different bin $b\in\left\{1,\dots,2^{m}\right\}$ according to the following method: The first user receives the least likely bin (i.e., when considering any hash function, this is the bin that has the smallest fraction of mappings; when considering a universal set of hash functions, this is the all ones bin), the second user receives the second least likely bin, whereas the last user receives the $M$th least likely bin.
\end{defn}

Three remarks regarding bins allocation are in order.
\begin{rem}\label{rem:passwordallocationwithnodetails}
In Definition \ref{def:BackdoorForAnyHash} we present a ``backdoor'' mechanism that efficiently implements  the bins allocation procedure presented in Definition \ref{def:PasswordAllocationwithNoDetails}, without decreasing the average guesswork.
\end{rem}

\begin{rem}
In Definition \ref{def:PasswordAllocationwithNoDetails} we state that when two users have different passwords, they are mapped to different bins. In general, this is required in order for the backdoor mechanism discussed in remark \ref{rem:passwordallocationwithnodetails} not to decrease the average guesswork.
\end{rem}

\begin{rem}
mapping passwords to the least likely bins, either based on the fractions of mappings or on the distribution of the key when considering a keyed hash function, is essential for using bias to increase the average guesswork.
\end{rem}

Next, we find the average guesswork for every bin $b\in\pac{1,\dots, 2^{m}}$.

\begin{cor}\label{cor:AverageGuessworkAnyHashFunctionAnyBin}
Assume that the attacker guesses inputs of a $P_{H}$-hash function one by one until it finds one that is mapped to bin $b$ (remote or insider attacks). In this case, under the averaging arguments of Definition \ref{def:AveragingArgument} and when bins are allocated to users according to Definition \ref{def:PasswordAllocationwithNoDetails} we get for $n\ge\pa{1+\eps}\cdot m\cdot \pa{\log\pa{1/p}+H\pa{s}}$ where $\eps > 0$; the average guesswork of bin $b\in\pac{1,\dots,2^{m}}$ 
\beq{}
{
\lim_{m\to\infty}\frac{1}{m}\log\pa{E\pa{\Gbin}}= \lim_{m\to\infty}\frac{1}{m}\log\pa{1/\ph}=H\pa{q\pa{b}}+D\pa{q\pa{b}||p}
}
where $q\pa{b}$ is the type of bin $b\in\pac{1,\dots,2^{m}}$ and averaging is done according to Definition \ref{def:AveragingArgument}.
\end{cor}

\begin{proof}
The proof is similar to the poof of Theorem \ref{th:AverageGuessworkPerBinIidAssumption} with the exception that here bin allocation is in place. In order to make sure that bin allocation does not affect the average guesswork the minimum input length increases from $n\ge \pa{1+\eps_{1}}\cdot\log\pa{1/p}\cdot m$ bits  to $n\ge\pa{1+\eps}\cdot m\cdot \pa{\log\pa{1/p}+H\pa{s}}$. The proof appears in Section \ref{sec:GuessworkAnyHash}.
\end{proof}

\subsection{The Case where Passwords are Drawn Uniformly}

We begin by finding the average guesswork across users under a remote attack when the passwords are drawn uniformly, followed by a concentration result. We then consider an insider attack in which the attacker tries to find a password that is mapped to any of the assigned bins.

\begin{theorem}[The average guesswork under a remote attack]\label{th:LowerBoundExpAverageGuessworkAnyHashFunction}
Under the following assumptions.
\bitm
{
\item The attack defined in Definition \ref{def:AnOnlineAttackModel}.
\item Bins are allocated as in Definition \ref{def:BackdoorForAnyHash} to $M=2^{H\pa{s}\cdot m-1}$ users, where $1/2\le s\le 1$.
\item Passwords are drawn i.i.d. Bernoulli$\pa{1/2}$.
\item $P_{H}$ is based on the method of types as presented in Definition \ref{def:DefinitionOfIidCase}, where $p\le 1/2$.
\item $n\ge\pa{1+\eps}\cdot m\cdot \pa{\log\pa{1/p}+H\pa{s}}$ where $\eps > 0$.
}
The average guesswork when averaging according to Definition \ref{def:AveragingArgument} is equal to
\beq{eq:AverageGuessworkAnyhashFunctionFirstTimeAppears}
{
\lim_{m\to\infty}\frac{1}{m}\log\pa{E\pa{\Gon}}=\barr{cc}{H\pa{s}+D\pa{s||p} &\pa{1-p}\le s\le 1\\ 2\cdot H\pa{p}+D\pa{1-p||p}-H\pa{s} &1/2\le s\le \pa{1-p} }.
}

\end{theorem}
\begin{proof}
The full proof appears in Section \ref{sec:GuessworkAnyHash}.
\end{proof}

Figure \ref{fig:The_Average_Guesswork} illustrates the above result.

\begin{figure}
\epsfig{figure=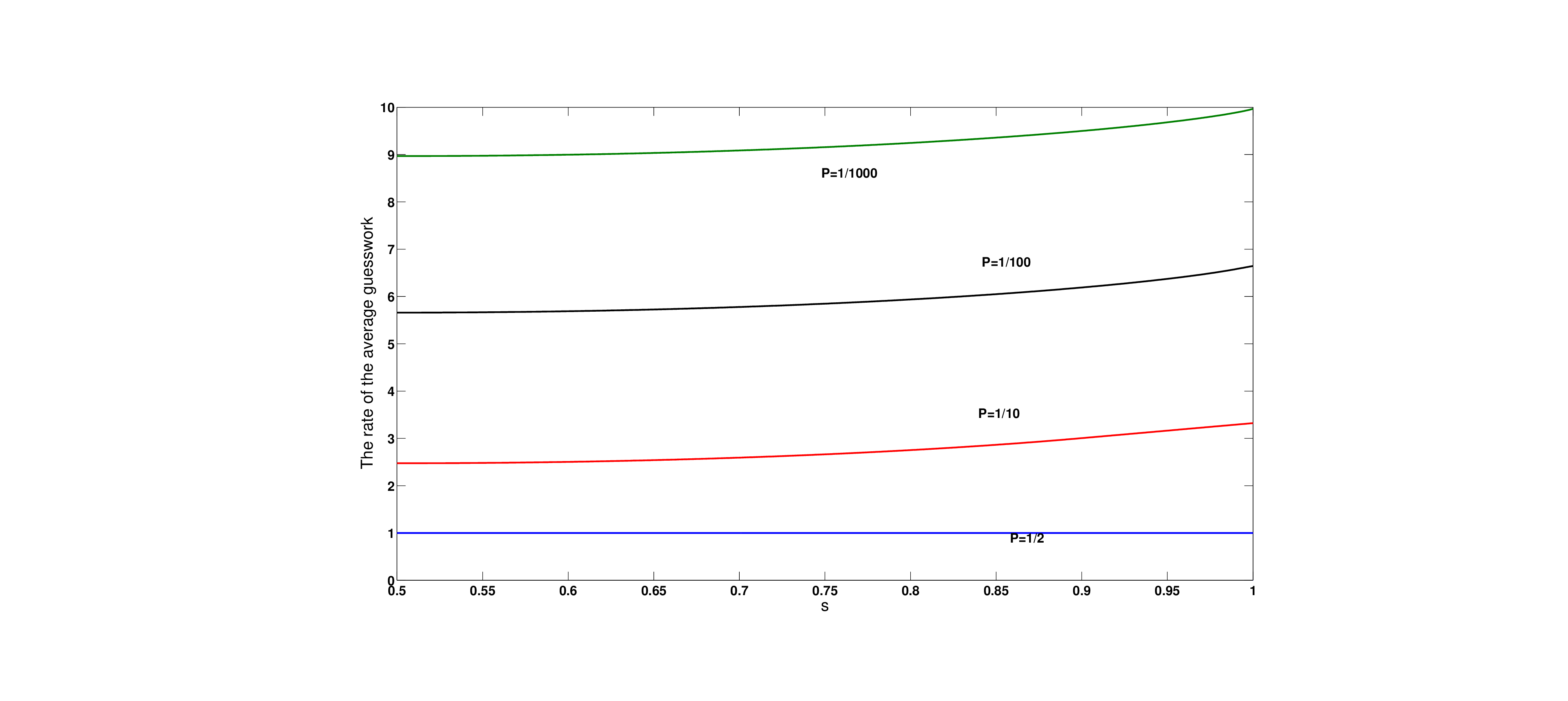,height=8cm}
\caption{The rate at which the average guesswork increases (i.e., $\lim_{m\to\infty}\frac{1}{m}\log\pa{E\pa{\Gon}}$) as a function of $s\in\left[1/2,1\right]$; the number of users is $2^{H\pa{s}\cdot m-1}$ and so the larger $s$ is, the smaller the number of users is. Hence, $s=1/2$ represents the largest number of users, whereas $s=1$ represents the smallest number. For $1/2\le s \le \pa{1-p}$ the average guesswork is equal to $2\cdot H\pa{p}+D\pa{1-p||p}-H\pa{s}$, and for $\pa{1-p}\le s\le 1$ the average guesswork equals $H\pa{s}+D\pa{s||p}$. Note that for $p=1/2$ the average guesswork increases at rate that equals $1$ regardless of the number of users, whereas as $p$ decreases the rate of the average guesswork increases.}\label{fig:The_Average_Guesswork}
\end{figure}

\begin{theorem}[A concentration result]\label{th:ConcentrationAnyHash}
Under the following assumptions.
\bitm
{
\item The attacks defined Definition \ref{def:AnOnlineAttackModel} and Definition \ref{def:AnOfflineAttackModel}.
\item Bins are allocated as in Definition \ref{def:BackdoorForAnyHash} to $M=2^{H\pa{s}\cdot m-1}$ users, where $1/2\le s\le 1$.
\item Passwords are drawn i.i.d. Bernoulli$\pa{1/2}$.
\item $P_{H}$ is based on the method of types as presented in Definition \ref{def:DefinitionOfIidCase}, where $p\le 1/2$.
\item $n\ge\pa{1+\eps}\cdot m\cdot \pa{\log\pa{1/p}+H\pa{s}}$ where $\eps > 0$.

}
The following concentration result holds for all users when avraging according to Definition \ref{def:AveragingArgument}:
\beq{}
{
-\lim_{m\to\infty}\frac{1}{m}\log\pa{P\pa{\Gbin\le 2^{\pa{1-\eps_{1}}\cdot\pa{H\pa{s}+D\pa{s||p}}\cdot m}}}\ge \eps_{1}\cdot\log\pa{1/p}\quad\forall b\in B
}
where $0<\eps_{1}<1$.
\end{theorem}
\begin{proof}
The proof is in Section \ref{sec:GuessworkAnyHash}.
\end{proof}

Figure \ref{fig:Intuitive_Explanation} illustrates the intuition behind the above results.

\begin{figure}[h!]
\input{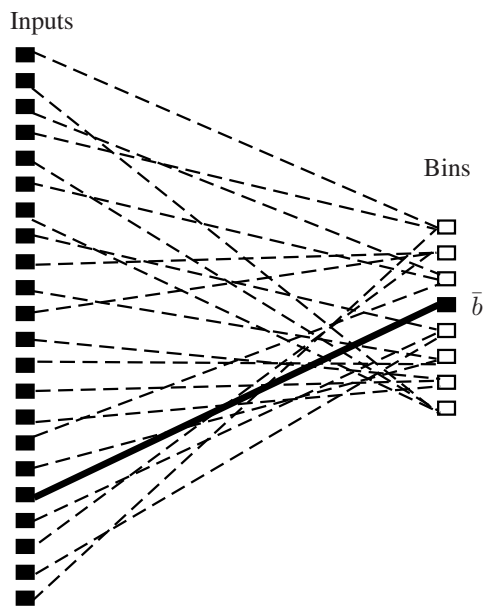}\caption{Due to unbalance, there is only one mapping from the set of inputs to the solid black bin $\bar{b}$ (the solid black line). Hence, the average number of guesses required to find a password that is mapped to it, is larger than the average number of guesses of the other bins, that is, the average guesswork of $\bar{b}$ is larger than the average guesswork of the others.}\label{fig:Intuitive_Explanation}
\end{figure}

\begin{cor}[The average guesswork under an insider attack]\label{cor:TheAverageGuessworkForGuessingAnyPassword}
When the following assumptions hold.
\bitm
{
\item The attack defined in Definition \ref{def:AnOfflineAttackModel}.
\item Bins are allocated as in Definition \ref{def:BackdoorForAnyHash} to $M=2^{H\pa{s}\cdot m-1}$ users, where $1/2\le s\le 1$.
\item Passwords are drawn i.i.d. Bernoulli$\pa{1/2}$.
\item $P_{H}$ is based on the method of types as presented in Definition \ref{def:DefinitionOfIidCase}, where $p\le 1/2$.
\item $n\ge\pa{1+\eps}\cdot m\cdot \pa{\log\pa{1/p}+H\pa{s}}$ where $\eps > 0$.
}
the average is equal to
\beq{}
{
\lim_{m\to\infty}\frac{1}{m}\log\pa{E\pa{\Gof}}=D\pa{s||p}.
}
\end{cor}
\begin{proof}
The full proof appears in Section \ref{sec:GuessworkAnyHash}.
\end{proof}

\begin{rem}
We assume that the number of users $M=2^{H\pa{s}\cdot m-1}$ since based on Definition \ref{def:PasswordAllocationwithNoDetails}, it assures that the type of each bin $b\in B$ satisfies
$s\le q\pa{b}\le 1$, where $q\pa{b}$ is the type of bin $b$. See Remark \ref{rem:SomeMoreDetailsonTheNumberOfUsersandTypes} for more details.
\end{rem}


\subsection{The Case where Passwords are Biased}\label{subsec:TheCaseWhereoasswordsAreBiased}
The next corollary presents the effect of a biased password on the average guesswork. We consider the case where passwords are drawn i.i.d. Bernoulli$\pa{\qp}$. Based on the concentration result of Theorem \ref{th:ConcentrationAnyHash}, we characterize a region in which the guesswork of Theorem \ref{th:LowerBoundExpAverageGuessworkAnyHashFunction} dominates the average guesswork as well as another region in which the average guesswork of a password \eqref{eq:GuessworkExpGrowthRate} is the dominant one.
Since in this case the optimal strategy is guessing passwords based on their probabilities in descending order \cite{Arikan_Ineq_Guessing}, the optimal average guesswork should be averaged over a keyed $P_{H}$-hash function (i.e., not over all possible strategies of guessing passwords).
\begin{cor}[Biased passwords]\label{cor:EffectofBiasedPassword}
Under the following assumptions.
\bitm
{
\item The attack defined in Definition \ref{def:AnOnlineAttackModel}.
\item Bins are allocated as in Definition \ref{def:BackdoorForAnyHash} to $M=2^{H\pa{s}\cdot m-1}$ users, where $1/2\le s\le 1$.
\item Unlike Definition \ref{def:BackdoorForAnyHash}, here the passwords are drawn i.i.d. Bernoulli$\pa{\qp}$, where $0<\qp<1/2$.
\item $P_{H}$ is based on the method of types as presented in Definition \ref{def:DefinitionOfIidCase}, where $p\le 1/2$.
\item $n\ge\pa{1+\eps}\cdot m\cdot \pa{\log\pa{1/p}+H\pa{s}}$ where $\eps > 0$.
}
When passwords are guessed based on their probabilities in descending order, the average guesswork of a user who is mapped to bin $b\in B$, averaged over a keyed $P_{H}$-hash function as in Definition \ref{def:AveragingArgument} is equal to
\bal{}
{
\lim_{m\to\infty}\frac{1}{m}\log &\pa{E\pa{\Gbin}}&=\nn\\
&\lim_{m\to\infty}&\barr{cc}{\frac{n}{m}H_{1/2}\pa{\qp} &2\cdot\frac{n}{m}\cdot H\pa{\frac{\sqrt{\qp}}{\sqrt{\qp}+\sqrt{1-\qp}}}< H\pa{q\pa{b}}+D\pa{q\pa{b}||p}\\ H\pa{q\pa{b}}+D\pa{q\pa{b}||p} &\frac{n}{m}H\pa{\qp}>H\pa{q\pa{b}}+D\pa{q\pa{b}||p} }.
}
Furthermore, the average guesswork across users is
\bal{eq:AverageAcrossUsersPasswordBiased1}
{
\lim_{m\to\infty}\frac{1}{m}\log \pa{E\pa{\Gon}}=
\lim_{m\to\infty}\frac{n}{m}H_{1/2}\pa{\qp}\quad \frac{2n}{m}H\pa{\frac{\sqrt{\qp}}{\sqrt{\qp}+\sqrt{1-\qp}}}< H\pa{s}+D\pa{s||p}
}
and
\bal{eq:AverageAcrossUsersPasswordBiased2}
{
\lim_{m\to\infty}\frac{1}{m}\log\pa{E\pa{\Gon}}=\barr{cc}{H\pa{s}+D\pa{s||p} &\pa{1-p}\le s\le 1\\ 2\cdot H\pa{p}+D\pa{1-p||p}-H\pa{s} &1/2\le s\le \pa{1-p} }
}
when $\lim_{m\to\infty}\frac{n}{m}H\pa{\qp}>\log\pa{1/p}$.
\end{cor}
\begin{proof}
The full proof appears in Section \ref{sec:GuessworkAnyHash}.
\end{proof}

We now give an illustrative example for the result above.
\begin{example}\label{ex:BiasedPasswords190926}
Consider the case where  $n=\pa{1+\eps}\cdot\log\pa{1/p}\cdot m$ and every user out of the $2^{H\pa{s}\cdot m-1}$ users draws his password i.i.d. Bernoulli$\pa{\qp}$. In this case, when $2\pa{1+\eps}\cdot\log\pa{1/p}\cdot H\pa{\frac{\sqrt{\qp}}{\sqrt{\qp}+\sqrt{1-\qp}}}< H\pa{s}+D\pa{s||p}$ the rate at which the average guesswork across users increases equals $\pa{1+\eps}\cdot\log\pa{1/p}\cdot H_{1/2}\pa{\qp}$, whereas when $\pa{1+\eps}\cdot\log\pa{1/p}\cdot H\pa{\qp}>\log\pa{1/p}$ the rate is equal to
$$\barr{cc}{H\pa{s}+D\pa{s||p} &\pa{1-p}\le s\le 1\\ 2\cdot H\pa{p}+D\pa{1-p||p}-H\pa{s} &1/2\le s\le \pa{1-p} }$$.
\end{example}

\begin{rem}
Example \ref{ex:BiasedPasswords190926} highlights the fact that based on the relation between $n$ and $m$ (i.e., the rate at which $n$ grows compared to $m$), the bias in the hash function (parameterized by $p$), and the bias in drawing the passwords (parameterized by $\theta$), we get a shift in the exponent of the average guesswork. When passwords are drawn from sufficiently biased distribution compared to the ratio $n/m$ and to $p$ we get that the most likely event that determines the average guesswork is the case where the attacker simply guesses the password of the user; in this case the average guesswork is the same as classic average guesswork. For example, when the password of the user is $Ps$, and the hashed value of the password is $H_{F}\pa{Ps}$, is it very likely in this case that the attacker manages to guess an input $n^{\ast}=Ps$ so that $H_{F}\pa{n^{\ast}} = H_{F}\pa{Ps}$. On the other hand, in the case where the bias in drawing passwords is not that severe, the most likely event that dominates the average guesswork is that the attacker manages to guess an input that is mapped to the same bin as the password of the user; therefore, in this case we get that the average guesswork deviates from classic results. Essentially, it means that when the bias in drawing passwords is not that severe, the most likely event is that the attacker guesses an input $n^{\ast}\neq Ps$ so that $H_{F}\pa{n^{\ast}}=H_{F}\pa{Ps}$. 
\end{rem}

\section{The Average Guesswork of a Strongly Universal Set of Hash Functions}\label{sec:AverageGuessworkCrackingPasswords}

In this section we derive the average guesswork for strongly universal set of hash functions when the key is drawn i.i.d. Bernoulli$\pa{p}$, $0<p\le 1/2$. In Subsection \ref{subsec:GWEachBin} we derive the average guesswork for any bin as a function of the bias $p$. Then, in Subsection \ref{subsec:GWCentralized} we calculate the average guesswork as a function of the number of users and the bias. Furthermore, we present the backdoor mechanism to allocate bins. Finally, in Subsection \ref{subsec:distributedpasswordallocation} we analyze the guesswork when the hash function is not modified.

We begin by defining strongly universal set of hash functions \cite{Wegman_CarterNewHashFunc1981}, which is nothing but the set of all hash functions with an input of length $n$ bits and an output of length $m$ bits.

\begin{defn}
A strongly universal set of hash functions is the set  $\left\{H_{f\pa{\udl{k}}}\pa{\cdot}\right\}$, where $\udl{k}$ is a key of size $m\cdot 2^{n}$, $f\pa{\udl{k}}$ is a bijection, and $\left\{H_{i}\pa{\cdot}\right\}$, $i\in\left\{1,\dots ,m\cdot 2^{n}\right\}$, is the set of all hash function with input of length $n$ bits and output of length $m$ bits. 
\end{defn}

We now define the mapping of keys to hash functions that we consider in this work for biased keys.

\begin{defn}\label{def:StronglyUniversalSetofHashFunctionsBiasedKeys}
When the key is broken into $2^{n}$ segments of size $m$ such that
\beq{eq:KeySegmentation}
{
\udl{k}=\left\{k_{1},\dots,k_{2^{n}}\right\}.
}
we state that the mapping $g\pa{\cdot}$ leads to the following mapping from keys to hash functions
\beq{eq:USHF}
{
H_{g\pa{\udl{k}}}\pa{i}=k_{i}\quad i\in\left\{1,\dots,2^{n}\right\}.
}
\end{defn}
 
In this work, when considering biased keys we assume the mapping $g\pa{\cdot}$. Figure \ref{fig:Keyed_Hash_Def} illustrates the definition presented above.

\begin{figure}[h!]\input{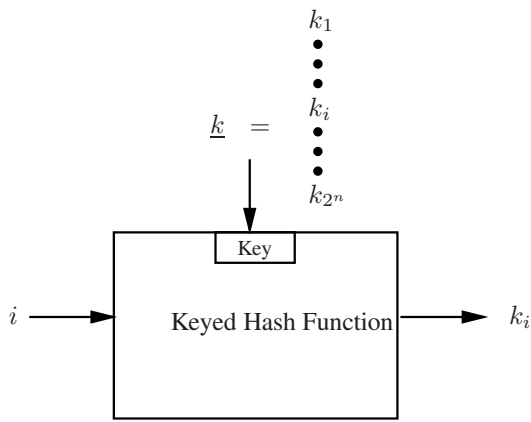}\caption{The definition of the keyed hash function that we analyze in this paper. When the input equals $i$, the output is equal to the $i$th segment of the key, that is, $k_{i}$. Essentially, this is a strongly universal set of hash functions.}\label{fig:Keyed_Hash_Def}
\end{figure}

\begin{rem}
Note that the set of hash functions presented in Definition \ref{def:TheSetofPhHashFunctions} is a subset of a strongly universal set of hash functions.
\end{rem}

\begin{rem}
Note that the mapping defined in equations \eqref{eq:KeySegmentation}, \eqref{eq:USHF} from the set of keys to the set of all possible hash functions, is a strongly universal set of hash functions. Any other mapping that is a bijection from the set of keys to the set of all possible hash function also constitutes a strongly universal set of hash functions. Essentially, in section \ref{sec:Discussion} we show that this mapping along with a biased key achieves better performance than an unbiased key with any mapping that is a bijection (see Corollary  \ref{cor:RelationBetweenRateOfAverageGuessworkBiasandUnbias} for the average guesswork of an unbiased key and any mapping that is a bijection). A biased key with other mappings may also lead to better performance than the one achieved with an unbiased key; however, we have been able to quantify the average guesswork for the mapping defined in equations \eqref{eq:KeySegmentation}, \eqref{eq:USHF}.
\end{rem}

In terms of the knowledge the attacker has under a remote and an insider attacks defined in Definition \ref{def:AnOnlineAttackModel} and Definition \ref{def:AnOfflineAttackModel} respectively, when the system uses the universal set of hash functions defined in equations \eqref{eq:KeySegmentation}, \eqref{eq:USHF}, we make two additional assumptions: The key is drawn i.i.d. Bernoulli$\pa{p}$; and the attacker knows the distribution according to which the key is drawn, but it does not know the actual realization of the key.

\subsection{The Guesswork of Each Bin}\label{subsec:GWEachBin}
In this subsection we derive the guesswork for each bin, when the key is i.i.d. Bernoulli$\pa{p}$. Assume a strongly universal set of hash functions as defined equations \eqref{eq:KeySegmentation}, \eqref{eq:USHF}. Note that although in Definition \ref{def:PasswordsAllocation} we describe procedure according to which bins are allocated, this procedure is different from the backdoor mechanism which is given in Definition \ref{def:backdoor}.

\begin{defn}[Bins allocation]\label{def:PasswordsAllocation}
Assume there are $\lfloor 2^{H\pa{s}\cdot m-1}\rfloor$ users, where $s\in\left[1/2,1\right]$. For each user a procedure allocates a different integer $b\in\left\{1,\dots,2^{m}\right\}$ according to the following method: The first user receives the least likely bin (i.e., the all ones bin), the second user receives the second least likely bin, whereas the last user receives the $\lfloor 2^{H\pa{s}\cdot m-1}\rfloor$th least likely bin. Then, assuming $L_{b}$ passwords are mapped to bin $b$, the procedure uniformly draws a password from $\left\{1,\dots,|L_{b}|\right\}$. In the case when $|L_{b}|=0$, the bin is not mapped to any password and we assume that the average guesswork is equal to zero.
\end{defn}

A remark is in order regarding the case where a user has no password.
\begin{rem}
Note that the backdoor mechanism presented in Definition \ref{def:backdoor} does not allow a scenario where a user has no password. As shown in Theorem \ref{th:backdoor}, this does not affect the average guesswork.
\end{rem}

The next lemma proves that the order according to which the attacker guesses passwords, does not affect the average guesswork.
\begin{lem}\label{lem:AnyStrategySameCrackwork}
For any bin $b$, the average guesswork is not affected by the order according to which the attacker guesses passwords one by one, when averaging over all possible keys.
\end{lem}
\begin{proof}
Consider a certain bin $b\in\pac{1,\dots,2^{m}}$. The attacker knows the probability mass function according to which the key is drawn. It also knows the set of bins allocated to the $\lfloor 2^{H\pa{s}\cdot m-1}\rfloor$ users, as presented in Definition \ref{def:PasswordAllocationwithNoDetails}. However, it does not know the actual realizations, and therefore the mapping from passwords to bins. Hence, from symmetry arguments, any strategy according to which the attacker guesses passwords one by one results in the same average guesswork, when averaging over all possible keys.
\end{proof}

Based on Lemma \ref{lem:AnyStrategySameCrackwork} we assume without loss of generality that the attacker guesses passwords one by one in ascending order. Now, we can derive the average guesswork for any bin $b$.
\begin{theorem}[Average number of guesses for each bin]\label{th:AverageNumberOfGuessesForEachBin}
Assume that the attacker tries to guess a password that is mapped to bin $b$. In this case when $n\ge \pa{1+\eps_{1}}\cdot\log\pa{1/p}\cdot m$, $\eps_{1}>0$, the average guesswork of $\Gbin$ is equal to
\beq{eq:theoremGuessworkofeachBin}
{
E\pa{\Gbin}=2^{m\cdot\pa{H\pa{q\pa{b}}+D\pa{q\pa{b}||p}}}-\pa{1-2^{-m\cdot\pa{H\pa{q\pa{b}}+D\pa{q\pa{b}||p}}}}^{2^{n}}\cdot\pa{2^{m\cdot\pa{H\pa{q\pa{b}}+D\pa{q\pa{b}||p}}}+2^{n}}}
where $q\pa{b}=\frac{N\pa{1|b}}{m}$ is the type of $b$, the key is drawn i.i.d. Bernoulli$\pa{p}$, and the universal set of hash functions is defined in equations \eqref{eq:KeySegmentation}, \eqref{eq:USHF}.
\end{theorem}
\begin{proof}
Based on Lemma \ref{lem:AnyStrategySameCrackwork} and without loss of generality we assume that the attacker guesses passwords one by one in ascending order. Essentially, in order to crack bin $b$ the attacker has to find \emph{a password} that is mapped to this bin. The proof relies on the observation that the chance of success in the $l$th guess (i.e., the first time a password which is mapped to bin $b$ is guessed) is drawn according to the following \emph{geometric distribution}
\beq{eq:geometricdistributiongeneralexpression}
{
P_{K}\pa{b}\cdot\pa{1-P_{K}\pa{b}}^{l-1}\quad 1\le l\le 2^{n}
}
where $P_{K}\pa{\cdot}$ is the probability mass function of a vector of length $m$ that is drawn i.i.d. Bernoulli$\pa{p}$. The probability that bin $b$ is not mapped to any password is
\beq{}
{
\pa{1-P_{K}\pa{b}}^{2^{n+1}}.
}
Since $P_{K}\pa{b}\ge p^{m}$, we can state that as long as $n\ge \pa{1+\eps_{1}}\cdot\log\pa{1/p}\cdot m$ where $\eps_{1}>0$, the probability for this event vanishes exponentially fast and therefore does not affect the average guesswork; furthermore, we assume that when this event occurs, it takes zero guesses for the attacker to crack the bin, such that this event does not even add up to the average guesswork. We wish to stress that the ``backdoor'' procedure presented in Definition \ref{def:backdoor} does not allow for a situation where a user does not have a password.

The mean value of \eqref{eq:geometricdistributiongeneralexpression} is equal to
\beq{eq:MeanValueTruncGeometricDist}
{
1/P_{K}\pa{b}-\pa{1-P_{K}\pa{b}}^{2^{n}}\cdot\pa{1/P_{K}\pa{b}+2^{n}}.
}
Assuming bin $b$ is of type $q\pa{b}=\frac{N\pa{1|b}}{m}$, and that the key is drawn i.i.d. Bernoulli$\pa{p}$, we get from Lemma \ref{th:methodoftypesprobth} that
\beq{}
{
1/P_{K}\pa{b}=2^{m\cdot\pa{H\pa{q\pa{b}}+D\pa{q\pa{b}||p}}}.
}
By assigning the equation above in \eqref{eq:MeanValueTruncGeometricDist} we get \eqref{eq:theoremGuessworkofeachBin}.
\end{proof}

Now, we show that the average guesswork converges uniformly to $2^{m\cdot\pa{H\pa{q\pa{b}}+D\pa{q\pa{b}||p}}}$, $\forall b$.

\begin{cor}\label{cor:MeanValueUniformlyConverges}
The average guesswork uniformly converges to $1/P_{K}\pa{b}$ across bins. The difference is upper bounded by the following term.
\bal{eq:TheConvergesRatetooneoverPForEachBin}
{
|2^{m\cdot\pa{H\pa{q\pa{b}}+D\pa{q\pa{b}||p}}}-E\pa{\Gbin}|\le \pa{2^{m\cdot\pa{H\pa{q\pa{b}}+D\pa{q\pa{b}||p}}}+2^{n}}\cdot e^{-2^{n-m\cdot\pa{H\pa{q\pa{b}}+D\pa{q\pa{b}||p}}}}\nonumber\\
\le\pa{2^{m\log\pa{1/p}}+2^{n}}\cdot e^{-2^{n-m\log\pa{1/p}}} \quad\forall b
}
\end{cor}
\begin{proof}
we wish to upper bound the term $\pa{1-2^{-m\cdot\pa{H\pa{q\pa{b}}+D\pa{q\pa{b}||p}}}}^{2^{n}}\cdot\pa{2^{m\cdot\pa{H\pa{q\pa{b}}+D\pa{q\pa{b}||p}}}+2^{n}}$ for any $b\in\pac{1,\dots,2^{m}}$ in order to show that it converges uniformly when $n\ge \pa{1+\eps_{1}}\cdot\log\pa{1/p}$, $\eps_{1}>0$. We begin by using the inequality
\beq{eq:PolynomialUpperBound}
{
\pa{1-\frac{1}{x}}^x\le e^{-1}\quad\forall x>1
}
in order to get
\beq{eq:uniformboundpolynomial}
{
\pa{1-2^{-m\cdot\pa{H\pa{q\pa{b}}+D\pa{q\pa{b}||p}}}}^{2^{n}}\le e^{-2^{n-m\cdot\pa{H\pa{q\pa{b}}+D\pa{q\pa{b}||p}}}}.
}
Furthermore, since the least likely type $q\pa{b_{0}}=1$ occurs with probability $p^{m}=2^{-m\log\pa{1/p}}$, we get
\beq{}
{
H\pa{q}+D\pa{q||p}\leq\log\pa{1/p}\quad 0\le q\le 1
}
with equality at q=1. Therefore, we get that $2^{-m\cdot\pa{H\pa{q\pa{b}}+D\pa{q\pa{b}||p}}}\ge 2^{-m\log\pa{1/p}}$ and
$2^{m\cdot\pa{H\pa{q\pa{b}}+D\pa{q\pa{b}||p}}}\le 2^{m\log\pa{1/p}}$. By assigning these terms along with the inequality in equation \eqref{eq:uniformboundpolynomial} we get the desired result.
\end{proof}

The next theorem shows that the guesswork concentrates around its mean value.
\begin{theorem}[Concentration]\label{th:Concentration}
The probability of finding a password that is mapped to bin $b$ in less than $2^{m\cdot l}$ guesses, where $l\le n$, is upper bounded by
\beq{eq:TheRateAtWhichTheConcentrationDecreases}
{
P\pa{\Gbin\le 2^{m\cdot l}}\le 1-e^{-2\cdot 2^{-\pa{H\pa{q\pa{b}}+D\pa{q\pa{b}||p}-l}\cdot m}}.
}
Therefore, whenever $l<H\pa{q\pa{b}}+D\pa{q\pa{b}||p}$ the probability of success in less than $2^{m\cdot l}$ attempts decays to zero.
\end{theorem}
\begin{proof}
Based on the fact that the guesswork has a geometric distribution, we get that the chance of success within $2^{m\cdot l}$ guesses is equal to
\beq{}
{
P\pa{\Gbin\le 2^{m\cdot l}}=P_{K}\pa{b}\cdot\sum_{i=1}^{2^{m\cdot l}}\pa{1-P_{K}\pa{b}}^{i-1}=1-\pa{1-P_{K}\pa{b}}^{2^{m\cdot l}}.}
The upper bound for $P\pa{\Gbin\le 2^{m\cdot l}}$ is obtained by assigning $P_{K}\pa{b}=2^{-m\pa{H\pa{q\pa{b}}+D\pa{q\pa{b}||p}}}$ to the following inequality
\beq{}
{
e^{-x}\le{1-\frac{x}{2}}\quad x\in\left[0,1.58\right].
}
Therefore, when $l<H\pa{q\pa{b}}+D\pa{q\pa{b}||p}$ the probability of finding a password that is mapped to bin $b$ goes to zero.
\end{proof}

\begin{rem}
Note that although the number of passwords scales exponentially with $n$ (i.e., there are $2^{n}$ passwords), the average number of guesses scales with the size of the bins $m$. This is due to the fact that cracking a password requires finding \emph{a password} which is mapped to the same bin as the true password.
\end{rem}

\subsection{The Average Guesswork for Cracking a Password of a User under Bins Allocation}\label{subsec:GWCentralized}

We later use the solution to the following optimization problem, in order to derive the optimal average guesswork for cracking a password of a user.
\begin{lem}\label{lem:OptimizationProblemForMaximalArgument}
\beq{eq:OptimizationProblemIncludingTypes}
{
\max_{0\le q\le 1}2\cdot H\pa{q}+D\pa{q||p}=2\cdot H\pa{p}+D\pa{1-p||p}
}
where the optimal solution occurs at $q=1-p$. Furthermore, $2\cdot H\pa{p}+D\pa{1-p||p}$ is a positive and unbounded function that monotonically increases as $p$ decreases.
\end{lem}

\begin{proof}
First let us break down the expression in \eqref{eq:OptimizationProblemIncludingTypes}.
\beq{}
{
2\cdot H\pa{q}+D\pa{q||p}=H\pa{q}+q\cdot\log\pa{1/p}+\pa{1-q}\cdot\log\pa{1/\pa{1-p}}.
}
The above expression is a concave function as a function of $q$ as it is the summation of two concave functions; hence, this function has a maximal value. By calculating the first derivative and making it equal to zero we get
\beq{}
{
\log\pa{\frac{1-q}{q}}=\log\pa{\frac{p}{1-q}}.
}
Equality holds only when $q=1-p$. By assigning it to the expression on the left hand side in \eqref{eq:OptimizationProblemIncludingTypes} we get the desired result.
\end{proof}

We now quantify the number of mappings from passwords to each bin.
\begin{lem}
Following the definition in  equations \eqref{eq:KeySegmentation}, \eqref{eq:USHF}, assume that $k_{i}$ is the mapping from the $i$th password to the range of the hash function, where $k_{i}\in\pac{1,\dots,2^{m}}$ and $i\in\pac{1,\dots, 2^{n}}$. In this case
\beq{eq:ChebyshevsInequalityForTypes}
{
P\pa{|2^{-n}\cdot \sum_{i=1}^{2^{n}}\mathbbm{1}_{b}\pa{k_{i}}-P_{K}\pa{b}|>\epsilon}\le\frac{1}{2^{n}\cdot\epsilon^{2}}
}
where $\mathbbm{1}_{b}\pa{x}=\barr{cc}{1 &x=b\\ 0 &x\neq b}$.
\end{lem}
\begin{proof}
Essentially $1_{b}\pa{k_{i}}$ is a random variable whose mean value is equal to $P_{K}\pa{b}$ and variance value is smaller than $1$. Furthermore, $1_{b}\pa{k_{i}}$, $i\in\pac{1,\dots,2^{n}}$ are all i.i.d. random variables. Therefore, we get \eqref{eq:ChebyshevsInequalityForTypes} from Chebyshev's inequality.
\end{proof}

Now we can derive the maximal average guesswork.
\begin{theorem}[The average number of guesses]\label{th:TheAverageGuesseork}
Let us denote the number of users by $\lfloor 2^{H\pa{s}\cdot m-1}\rfloor$, where $s\in\paq{1/2,1}$ is a parameter. The exponential rate at which the optimal average guesswork increases as a function of $m$ is equal to
\beq{eq:AverageGuessworkTheMostCopactEpression}
{
\lim_{m\to\infty}\frac{1}{m}\log\pa{E\pa{\Gon}}=\barr{cc}{H\pa{s}+D\pa{s||p} &\pa{1-p}\le s\le 1\\ 2\cdot H\pa{p}+D\pa{1-p||p}-H\pa{s} &1/2\le s\le \pa{1-p} }
}
where $n\ge\pa{1+\eps_{1}}\cdot m\cdot \log\pa{1/p}$, $\eps_{1} >0$.
\end{theorem}
\begin{proof}
The general idea behind the proof is to write the average in terms of summation over types, and then use the fact that a vector of length $m$ has up to $m$ types in order to bound the average by the most dominant term among the $m$ types.

First let us define the set of all possible bins $\pac{b_{1},\dots,b_{2^{m}}}$ such that $P_{K}\pa{b_{i}}\ge P_{K}\pa{b_{j}}$ if and only if $j<i$; therefore, according to Definition \ref{def:PasswordAllocationwithNoDetails} $b_{1}$ is assigned to the first user, whereas $b_{2^{H\pa{s}\cdot m}}$ goes to the $2^{H\pa{s}\cdot m-1}$th user. The average guesswork when the attacker chooses a user name to attack, uniformly from $\pac{1,\dots, 2^{H\pa{s}\cdot m-1}}$, is equal to
\beq{eq:AverageGuessStraightForwardExpression}
{
E\pa{\Gon}=2^{-H\pa{s}\cdot m+1}\sum_{i=1}^{2^{H\pa{s}\cdot m-1}}\frac{1}{p_{K}\pa{b_{i}}}-\eps\pa{m,n,q\pa{b_{i}}}
}
where $\eps\pa{m,n,q\pa{b_{i}}}=\pa{2^{m\cdot\pa{H\pa{q\pa{b}}+D\pa{q\pa{b}||p}}}+2^{n}}\cdot e^{-2^{n-m\cdot\pa{H\pa{q\pa{b}}+D\pa{q\pa{b}||p}}}}$ in accordance with Corollary \ref{cor:MeanValueUniformlyConverges}. Note that the term above takes into consideration the event when a mapping to $b$ does not occur within the $2^{n}$ inputs; this probability is multiplied by zero guesses, which leads to the expression in equation \eqref{eq:AverageGuessStraightForwardExpression}.

Based on the method of types \cite{CoverBook} every bin in $\pac{b_{1},\dots,2^{H\pa{s}\cdot m-1}}$, which is the set of bins that are allocated to the $2^{H\pa{s}\cdot m-1}$ users, is of type $s\le q\pa{b_{i}}\le 1$, where $1\le i\le 2^{H\pa{s}\cdot m-1}$, and $q\pa{b_{i}}\ge q\pa{b_{j}}$ if and only if $i<j$. Therefore, from Lemma \ref{th:methodoftypesprobth} and Lemma \ref{th:methodoftypessize} we can rewrite the summation such that is goes across types, such that we can bound the average guesswork by the following terms
\beq{eq:TheSummationWhenTruncatingTheMinimalValue}
{
\frac{1}{\pa{m+1}^{2}}2^{-H\pa{s}\cdot m+1}\max_{{s\le q\le 1}}\pa{2^{m\cdot\pa{2\cdot H\pa{q}+D\pa{q||p}}}-\eps\pa{m,n,q},0}\le E\pa{\Gon}\le 2^{-H\pa{s}\cdot m+1}\sum_{s\le q\le 1}2^{m\cdot\pa{2\cdot H\pa{q}+D\pa{q||p}}}.
}
The number of types is equal to the size of each bin, and therefore there are only $m$ types; furthermore, the elements in the summations above all are non-negative. Hence, when $n\ge\pa{1+\eps_{1}}\cdot m\cdot \log\pa{1/p}$ and $m\gg 1$ the following terms bound the average guesswork.
\beq{eq:TheBoundWithMaximalValueOnly}
{
\frac{2^{-H\pa{s}\cdot m+1}}{\pa{m+1}^{2}}\cdot 2^{m\cdot\max_{s\le q\le 1}\pa{ 2\cdot H\pa{q}+D\pa{q||p}}}\le E\pa{\Gon}\le m\cdot 2^{-H\pa{s}\cdot m+1}\cdot 2^{m\cdot \max_{s\le q\le 1}\pa{2\cdot H\pa{q}+D\pa{q||p}}}.
}
We are interested in the rate at which the guesswork grows, which based on the above bound is equal to
\beq{eq:OptimizationProblemForTheAverageGuesswork}
{
\lim_{n\to\infty}\frac{1}{m}\log\pa{E\pa{\Gon}}=-H\pa{s}+\max_{s\le q\le 1}\pa{2\cdot H\pa{q}+D\pa{q||p}}.
}
The function $2\cdot H\pa{q}+D\pa{q||p}$ is concave as a function of $q$, with a maximal value at $q=1-p$ as shown in Lemma \ref{lem:OptimizationProblemForMaximalArgument}. Therefore, when $\pa{1-p}\le s\le 1$ the maximal value is obtained at $q=s$; whereas when $1/2\le s\le \pa{1-p}$ the maximal value $2\cdot H\pa{1-p}+D\pa{q||p}=2\cdot H\pa{p}+D\pa{q||p}$ is achieved for $q=1-p$. Note that $n\ge\log\pa{1/p}\ge H\pa{q}+D\pa{q||p}$ for $0\le q\le 1$ which leads to the elimination of $\eps\pa{n,m,q}$. These arguments conclude the proof.
\end{proof}

We now present a backdoor mechanism that enables to modify a hash function efficiently without decreasing the average guesswork.

\begin{rem}[A Back door Mechanism for Allocating bins]
Essentially, in order for a procedure to allocate bins to the users, it first allocates a bin to a user, and then maps the password of this user to this bin; furthermore, there is a certain chance that because of the key realization, there is no mapping to a specific bin. This operation is as exhaustive as cracking the bin itself. In order for the system to allocate bins efficiently it has to use a \emph{backdoor} that enables it to allocate bins without the need to crack the hash function as well as without compromising the security of the system (i.e., maintaining the same average guesswork). In Definition \ref{def:backdoor} and Theorem \ref{th:backdoor} we present a back door mechanism that satisfies both of these requirements.
\end{rem}

Now we define a back door mechanism that enables to plant mappings in an efficient way without compromising the security level of the hash function.
\begin{defn} \label{def:backdoor}
Given a key $K=\pac{k_{1}, \dots, k_{2^{n}}}$ where $k_{i}\in\pac{1,\dots, 2^{m}}$ and $1\le i\le 2^{n}$, and a strongly universal set of hash functions as defined in equations \eqref{eq:KeySegmentation}, \eqref{eq:USHF}, the back door mechanism for allocating bins is defined as follows.
\bitm
{
\item For each user choose a different bin according to the procedure described in Definition \ref{def:PasswordAllocationwithNoDetails} (i.e., the first user gets the least likely bin $b_{1}$, the second user receives the second least likely bin $b_{2}$, and the $2^{H\pa{s}\cdot m -1}$th user receives $b_{2^{H\pa{s}\cdot m -1}}$).
\item Each of the users from the first to the $2^{H\pa{s}\cdot m -1}$th draws a password uniformly by drawing $n$ bits i.i.d. Bernoulli$\pa{1/2}$.
\item For each user, if the number that was drawn $g\in\pac{1,\dots,2^{n}}$ has not been drawn by any of the users that have bins that are less likely than the current bin, then replace $k_{g}$ with the bin number allocated to the user.
\item In the case when another user who is coupled with a less likely bin, has drawn $g$: \emph{Do not change} the value of the key again (i.e., $k_{g}$ is changed only once by the user who is coupled with the least likely bin among the bins of the users whose password is $g$). Instead, change the bin allocated to the user to the least likely bin among the bins allocated to users whose password is $g$. Therefore, $k_{g}$ is changed only once, to the least likely bin that is mapped to $g$.
}
\end{defn}
Figure \ref{fig:Back_Door} illustrates the back door mechanism.

\begin{figure}[h!]
\input{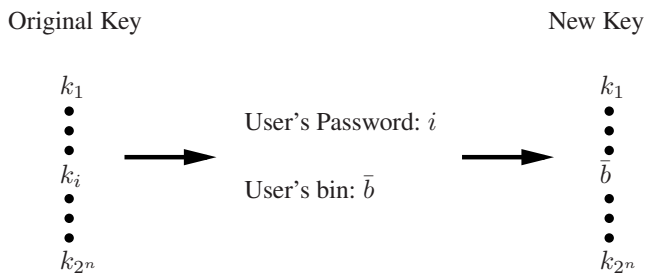}\caption{The backdoor mechanism. When the password of the user is equal to $i$, and no other user who is coupled with a bin that is less likely than the bin of the user, has come up with this password, the $i$th segment of the key is replaced with the bin that is coupled with the user. In the case when another user who is coupled with a less likely bin, has already come up with the exact same password, the more likely bin of the user (and not the key) changes to the bin of the other user.}\label{fig:Back_Door}
\end{figure}

\begin{theorem}\label{th:backdoor}
The optimal average guesswork when allocating bins using the back door mechanism given in Definition \ref{def:backdoor} is the same as the guesswork from Theorem \ref{th:TheAverageGuesseork} and is equal to
\beq{}
{
\lim_{m\to\infty}\frac{1}{m}\log\pa{E\pa{\Gon}}=\barr{cc}{H\pa{s}+D\pa{s||p} &\pa{1-p}\le s\le 1\\ 2\cdot H\pa{p}+D\pa{1-p||p}-H\pa{s} &1/2\le s\le \pa{1-p} }
}
\end{theorem}
when $\pa{1+\eps_{1}}\cdot m\cdot \log\pa{1/p}$, $\eps_{1}>0$.
\begin{proof}
The general idea behind the proof is that the back door procedure described in Definition \ref{def:backdoor} cannot add more than one mapping to each bin. Therefore, given any strategy of guessing passwords one by one, we show that the chance of drawing a password (in the backdoor procedure) that decreases the number of guesses bellow the average guesswork vanishes exponentially fast.

First, a few words are in order regarding the effect of the backdoor procedure presented in Definition \ref{def:backdoor}. The procedure cannot add more than one mapping to each bin. This is due to the fact that when a password that is mapped to a specific bin is drawn, then the mapping is changed to the bin of the user for which the password was drawn. Furthermore, if the same password is drawn more than once by several users, then the mapping does not change again; at worst more than one user is mapped to the same bin (the probability that two users have the same password is the same as the probability of a collision, and is equal to $2^{-n}$). When a collision occurs, the mapping to the least likely bin among the bins that are coupled with these users, is the one assigned to all these users. Therefore, collision does not decrease the average guesswork.

Following what we have discussed above, for any strategy we are interested in the probability of drawing a password which is guessed within a smaller number of attempts than the average guesswork (i.e., given any strategy of guessing passwords one by one, what is the probability that the backdoor mechanism decreases the number of guesses bellow the average guesswork). When this event occurs we assume that the number of guesses required to break the strongly universal set of hash functions is zero. When there are $2^{H\pa{s}\cdot m-1}$ users, the set of type $q\pa{b}=\max\pa{s,1-p}$ bounds the average guesswork as shown in equation \eqref{eq:TheBoundWithMaximalValueOnly}. Therefore, it is sufficient to focus on the bins of this type; as even when the guesswork of bins of other type is equal to zero, the average guesswork remains the same.

Without loss of generality let us focus on the case when passwords are guessed one by one in ascending order; the following argument holds for any other strategy of guessing passwords. For every user, the backdoor mechanism draws a password uniformly. When $n\ge\pa{1+\eps_{1}}\cdot m\cdot \log\pa{1/p}$ as defined in Theorem \ref{th:TheAverageGuesseork}, the probability of drawing a password that is guesses within a smaller number of attempts than $\pa{1+\eps_{2}}\cdot m\cdot \log\pa{1/p}$ where $0<\eps_{2}<\eps_{1}$ is
\beq{}
{
\frac{2^{\pa{1+\eps_{2}}\cdot m\cdot \log\pa{1/p}}}{2^{n}} \le 2^{-\pa{\eps_{1}-\eps_{2}}\cdot\log\pa{1/p}\cdot m}.
}
Thus, for each bin of type $q\pa{b}=\max\pa{s,1-p}$, the average guesswork is multiplied by a factor $\pa{1-2^{-\pa{\eps_{1}-\eps_{2}}\cdot\log\pa{1/p}\cdot m}}$ that approaches $1$ exponentially fast as $m$ increases. Therefore, the average guesswork is not affected by the back door mechanism. The only difference compared to Theorem \ref{th:TheAverageGuesseork} is that now the decay of $\eps\pa{n,m,1-p}$ in is dictated by $\eps_{2}$ instead of $\eps_{1}$.
\end{proof}

\begin{rem}
Note that the procedure for allocating bins which is described in Definition \ref{def:backdoor}, does not require knowledge of the key in order to achieve the average guesswork of Theorem \ref{th:backdoor}. The back door mechanism requires only knowledge of which bins are the least likely to occur based on the probability mass function according to which the key is drawn.
\end{rem}

\subsection{The Average Guesswork when Bins are not Allocated to Users}\label{subsec:distributedpasswordallocation}
In this subsection we show that when the users choose their own passwords, without changing the hash function accordingly, the average guesswork of the strongly universal set of hash functions (averaged over all passwords) is equal to $2^{m}$ for any $p$.

\begin{defn}[No bins allocation]\label{def:DIstributedPasswordsAllocation}
When bins are not allocated to users each user chooses his own password, without modifying the hash function. The server stores the bin to which a password is mapped.
\end{defn}

\begin{cor}[The Average Guesswork of the strongly universal set of hash functions with no bins Allocation]\label{Cor:NoBinsAllocationAveragedOverAllPasswords}
When each user chooses the password uniformly over $\pac{1,\dots, 2^{n}}$, the average guesswork \emph{when averaging over the passwords}, increases at rate that is equal to
\beqn{}
{
\lim_{m\to\infty}\frac{1}{m}\log\pa{E\pa{G_{D}\pa{B}}}=1
}
where $G_{D}\pa{B}$ is the guesswork of any user, $n\ge\pa{1+\eps_{1}}\cdot m\cdot\log\pa{1/p} $, and the hash function is defined in subsection \ref{subsec:SUSHF} equations \eqref{eq:KeySegmentation}, \eqref{eq:USHF}.
\end{cor}
\begin{proof}
From Theorem \ref{th:AverageNumberOfGuessesForEachBin} and equation \eqref{eq:MeanValueTruncGeometricDist} we can state that the average guesswork for each bin is equal to
\beq{}
{
E\pa{G_{D}\pa{b}}=1/P_{K}\pa{b}-\pa{1-P_{K}\pa{b}}^{2^{n}}\cdot\pa{1/P_{K}\pa{b}+2^{n}}\quad b\in\pac{1,\dots, 2^{m}}.
}
From Corollary \ref{cor:MeanValueUniformlyConverges} we know that the right hand side of the equation above goes to zero when $n\ge\pa{1+\eps_{1}}\cdot m\cdot\log\pa{1/p} $.

When a user chooses a password uniformly over $\pac{1,\dots,2^{n}}$, and the key is drawn i.i.d. Bernoulli$\pa{p}$, the probability of average guesswork that is equal to $E\pa{G\pa{b}}$ (i.e., the probability of drawing a password that is mapped to $b$) is equal to $P_{K}\pa{b}$. Hence, we get that the average guesswork is equal to
\beq{}
{
E\pa{G_{D}\pa{B}}=\sum_{b=1}^{2^{m}}P_{K}\pa{b}E\pa{G_{D}\pa{b}}
}
which leads to
\beq{}
{
\lim_{m\to\infty}\log\pa{E\pa{G_{D}\pa{B}}}=1.
}
\end{proof}

\begin{rem}
Note that when the least likely bins are allocated to the users (e.g., bins are allocated as in Definition \ref{def:backdoor}), the fact that there is a very small number of mappings to these bins, enables to achieve a better average guesswork with a shorter biased key (as we show in the next section). However, when the bins  are not allocated (i.e., the hash function is not modified), the chance of drawing a password that is mapped to a certain bin is equal to the probability of mapping a password to the very same bin. Therefore, in this case there is a small probability of drawing a bin that has a very small number of passwords mapped to it. This leads to an average guesswork that grows like $2^{m}$.
\end{rem}


\section{Uniform Keys Versus Biased Keys: Size and Storage Requirements}\label{sec:Discussion}

In this section we show that when the average guesswork is larger than the number of users, the minimal size of a biased key that achieves the average guesswork above with a universal set of hash functions defined in equations \eqref{eq:KeySegmentation}, \eqref{eq:USHF}, is smaller than the size of a uniform key (i.e., a key which is i.i.d. Bernoulli$\pa{1/2}$) that achieves the same guesswork for the same number of users, with any mapping from the set of keys to the set of hash functions. Furthermore, we show that the storage space required to store the bins of all users is also significantly smaller. We wish to point out that this holds when the passwords of the users are mapped to the least likely bins, that is, bins allocation as defined in Definition \ref{def:backdoor} is in place.

First, we show that the average guesswork of biased keys is larger than the average guesswork of any strongly universal set of hash functions whose key is drawn i.i.d. according to Bernoulli$\pa{1/2}$, that is, any mappings from the set of keys to the set of hash functions would lead to the same guesswork when the key is uniform.
\begin{cor}\label{cor:RelationBetweenRateOfAverageGuessworkBiasandUnbias}
When $n\ge m$, the guesswork of any strongly universal set of hash functions whose key is i.i.d. Bernoulli$\pa{1/2}$ satisfies
\beq{}
{
\lim_{m\to\infty}\frac{1}{m}\log\pa{G_{remote}^{\pa{u}}\pa{B,f\pa{\cdot}}} = 1\le\lim_{m\to\infty}\frac{1}{m}\log\pa{ \Gon}\quad 0\le s\le 1
}
with equality only when $p=1/2$; where $f\pa{\cdot}$ is a bijection from the set of keys to the set of all possible hash functions. Therefore, whenever $0<p<1/2$, the average guesswork of the strongly universal set of hash functions defined in equations \eqref{eq:KeySegmentation}, \eqref{eq:USHF}, is larger than the one achieved by any mapping $f\pa{\cdot}$ over a balanced key.
\end{cor}
\begin{proof}
When the key is drawn i.i.d. Bernoulli$\pa{1/2}$ the problem of finding the average guesswork can be regarded as a counting problem. Since $f\pa{\cdot}$ is a bijection, the number of functions for which the $i$th password is mapped to the $j$th bin is the same for any $i$ and $j$. Furthermore, when $l$ mappings are revealed, the set of keys decreases by a factor of $2^{l\cdot m}$ and the remaining mappings are still balanced. Therefore, for any bin the chance of cracking a password in the $l$th guess is a geometric distribution that equals to $2^{-m}\cdot \pa{1-2^{-m}}^{l-1}$, $l\ge 0$; this holds whenever $f\pa{\cdot}$ is a bijection. Hence, the rate at which the average guesswork increases when $n\ge\pa{1+\eps}\cdot m$ and $\eps > 0$, is equal to
\beq{}
{
\lim_{m\to\infty}\log\pa{G_{remote}^{\pa{u}}\pa{B,f\pa{\cdot}}} = 1\quad s\in\paq{1/2,1}
}
whereas from Theorem \ref{th:TheAverageGuesseork} we get that
\beq{}
{
\lim_{m\to\infty}\log\pa{ \Gon}\ge 2\cdot H\pa{p}+D\pa{1-p||p}-1\ge 1
}
with equality only when $p=1/2$, where the second inequality results from the fact that $2\cdot H\pa{p}+D\pa{1-p||p}$ decreases as $p$ approaches $1/2$, and is equal to $2$ at $p=1/2$.
\end{proof}

We now decrease the minimal size of the input, $n$, such that when $1/2\le s\le \pa{1-p}$, the average guesswork is smaller than the one in \eqref{eq:AverageGuessworkTheMostCopactEpression}; by decreasing $n$, we later show that in some cases the size of a biased key required to achieve a certain average guesswork, is significantly smaller than the size of an unbiased key that achieves the same guesswork.
\begin{lem}
When the number of users is $2^{H\pa{s}\cdot m-1}$, $n=\pa{1+\eps_{1}}\cdot m\cdot \pa{H\pa{s}+D\pa{s||p}}$ where $1/2\le s\le 1$, and the key is drawn i.i.d. Bernoulli$\pa{p}$, the average guesswork increases at rate that is equal to
\beq{}
{
\lim_{m\to\infty}\frac{1}{m}\log\pa{E\pa{\Gon}}=H\pa{s}+D\pa{s||p}.
}
\end{lem}
\begin{proof}
The general idea behind the proof is related to the fact that when the number of users is $2^{m\cdot H\pa{s}-1}$, and the users are coupled with the $2^{m\cdot H\pa{s}-1}$ least likely bins, the average guesswork of any user is larger than or equal to $2^{m\cdot \pa{\cdot{H\pa{s}+D\pa{s||p}}}}$ as long as $n\ge\pa{1+\eps_{1}}\cdot m\cdot\log\pa{1/p}$; this follows from Theorem \ref{th:AverageNumberOfGuessesForEachBin}, the fact that $s\le q\pa{b_{i}}\le 1$ where $i\in\pac{1,\dots,2^{m\cdot H\pa{s}-1}}$, and also because the backdoor mechanism does not decrease the guesswork of any bin.
However, here we consider the case when the input is equal to $n=\pa{1+\eps_{1}}\cdot m\cdot \pa{H\pa{s}+D\pa{s||p}}$, for which some of the assumptions above may not hold. Therefore, we examine two cases: The case where $q\pa{b_{i}}=s$ as well as $s<q\pa{b_{i}}\le 1$. Once we show that for both cases the rate at which the average guesswork increases is larger than or equal to $H\pa{s}+D\pa{s||p}$, we can also show that the average guesswork across all users increases at this rate.

For $q\pa{b_{i}}=s$ the size of the input supports the average guesswork, and so following along the same lines as Theorem \ref{th:TheAverageGuesseork} and Theorem \ref{th:backdoor}, we get that the average guesswork of these bins increases at rate $H\pa{s}+D\pa{s||p}$.

For $s<q\pa{b_{i}}\le 1$ the probability of a mapping to these bins is $P_{K}\pa{b_{i}}=^{-m\cdot\pa{H\pa{q\pa{b_{i}}}+D\pa{q\pa{b_{i}}||p}}}$, where $H\pa{q\pa{b_{i}}}+D\pa{q\pa{b_{i}}||p}>H\pa{s}+D\pa{s||p}$. Therefore, the average guesswork of these elements has to be larger than or equal to the case when $q\pa{b_{i}}=s$. In fact, when $\eps_{1}$ is small enough and $n=\pa{1+\eps_{1}}\cdot m\cdot \pa{H\pa{s}+D\pa{s||p}}$, the probability that there is a mapping to bin $b_{i}$ within the $2^{n}$ inputs goes to zero, and so the only occurrences of $b_{i}$ is due to the backdoor mechanism; the backdoor mechanism achieves average guesswork that increases at rate $\pa{1+\eps_{1}}\cdot \pa{H\pa{s}+D\pa{s||p}}$ when the password is drawn uniformly in accordance with Definition \ref{def:backdoor} (again, following the same lines as Theorem \ref{th:backdoor}).

\end{proof}

\begin{rem}
In order to show that a shorter biased key achieves the same guesswork that a longer balanced key achieves we need to find the actual guesswork and not only the rate at which it increases. The next lemma resolves this issue.
\end{rem}

Now, we derive bounds for the actual average guesswork and not only for the rate at which it increases.
\begin{lem}\label{lem:ActualLowerBoundAverageGuesswork}
When the number of users is $2^{H\pa{s}\cdot m-1}$, $1/2\le s\le 1$, and $n=\pa{1+\eps_{1}}\cdot m\cdot \pa{H\pa{s}+D\pa{s||p}}$ the average guesswork of the universal hash function defined in subsection \ref{subsec:SUSHF} equations \eqref{eq:KeySegmentation}, \eqref{eq:USHF}, a key which is i.i.d. Bernoulli$\pa{p}$, and an output of size $m$, is lower bounded by
\beq{eq:ExactLowerBoundAverageGuessworkWhenTheKeyisBiased}
{
E\pa{\Gon}\ge \pa{1-\gamma\pa{m,p,s}}\cdot 2^{m\cdot\pa{H\pa{s}+D\pa{s||p}}}-e^{-2^{\eps_{1}\cdot m}}\cdot\pa{2^{m\cdot\pa{H\pa{s}+D\pa{s||p}}}+2^{n}}
}\footnote{Note that the result from Theorem \ref{th:backdoor} can also be extended to a lower bound when $\pa{1-p}\le s\le 1$ by arguments that follow the same lines as Lemma \ref{lem:ActualLowerBoundAverageGuesswork}.}
where $\gamma\pa{s,p,m}$ decays exponentially to zero (i.e., $\lim_{m\to\infty}\gamma\pa{s,p,m}= 0$). On the other hand, for any universal set of hash functions, with a key that is i.i.d. Bernoulli$\pa{1/2}$, and an output of size $\pa{H\pa{s}+D\pa{s||p}}\cdot m$, the average guesswork is equal to
\beq{eq:ExactLowerBoundAverageGuessworkWhenTheKeyisUnbiased}
{
E\pa{G_{remote}^{\pa{u}}\pa{B},f\pa{\cdot}}=2^{m\cdot\pa{H\pa{s}+D\pa{s||p}}}-\pa{1-2^{-m\cdot\pa{H\pa{s}+D\pa{s||p}}}}^{2^{n}}\cdot\pa{2^{m\cdot\pa{H\pa{s}+D\pa{s||p}}}+2^{n}}.
}
Therefore, we get that
\beqn{}
{
\lim_{m\to\infty}\frac{E\pa{\Gon}}{E\pa{G_{remote}^{\pa{u}}\pa{B},f\pa{\cdot}}}\ge 1.
}
\end{lem}
\begin{proof}
Let us start by proving the inequality in \eqref{eq:ExactLowerBoundAverageGuessworkWhenTheKeyisBiased}. When $n=\pa{1+\eps_{1}}\cdot m\cdot \pa{H\pa{s}+D\pa{s||p}}$, the probability of a bin of type $q\pa{b}>s$ to be mapped by the key to any password decreases according to the expression
\beq{}
{
1-\pa{1-2^{-m\cdot \pa{H\pa{q\pa{b}}+D\pa{q\pa{b}||p}}}}^{2^{\pa{1+\eps_{1}}\cdot m\cdot \pa{H\pa{s}+D\pa{s||p}}}};
}
this is because of the fact that when $q\pa{b}>s$ the term $\pa{H\pa{q\pa{b}}+D\pa{q\pa{b}||p}}>\pa{H\pa{s}+D\pa{s||p}}$. The back door mechanism in Definition \ref{def:backdoor} plants a mapping by drawing a password uniformly. Therefore, we get with probability $\pa{1-\gamma_{1}\pa{s,p,m}}$ where $\gamma_{1}\pa{s,p,m}$ decays to zero exponentially fast (i.e., with probability that approaches $1$ very quickly), that when $q\pa{b}>s$ the problem of cracking a bin is as hard as guessing the mapping created by the back door mechanism. Since in this case the password is uniformly distributed over all possible passwords (i.e., $2^{n}$ possibilities), it can be easily shown that the average guesswork is equal to $\frac{2^{n}}{2}\ge 2^{\pa{H\pa{s}+D\pa{s||p}}\cdot m}$ when $m\ge m_{0}$. Finally, by assigning the right parameters to equations \eqref{eq:theoremGuessworkofeachBin}, \eqref{eq:TheConvergesRatetooneoverPForEachBin} when $q\pa{b}=s$, the average guesswork is lower bounded by
\beq{}
{
G\pa{b|q\pa{b}=s}\ge \pa{1-\gamma_{2}\pa{s,p,m}}\cdot 2^{\pa{H\pa{s}+D\pa{s||p}}\cdot m}-e^{-2^{\eps_{1}\cdot m}}\cdot\pa{2^{m\cdot\pa{H\pa{s}+D\pa{s||p}}}+2^{n}}.
}
Hence, the average guesswork over all users (or alternatively all bins for which $s\le q\pa{b}\le 1$) is lower bounded by the term in equation \eqref{eq:ExactLowerBoundAverageGuessworkWhenTheKeyisBiased}.

In the case when the key is uniform, \eqref{eq:ExactLowerBoundAverageGuessworkWhenTheKeyisUnbiased} is proven by assigning the right parameters in equation \eqref{eq:theoremGuessworkofeachBin} and keeping in mind that when the key is uniform and $f\pa{\cdot}$ is a bijection, the average guesswork is the same for any bin. These arguments prove the equality in equation \eqref{eq:ExactLowerBoundAverageGuessworkWhenTheKeyisUnbiased}.
\end{proof}

\begin{rem}\label{rem:SomeMoreDetailsonTheNumberOfUsersandTypes}
Note that when $s=1/2$, the number of users is $\sum_{i=1}^{m/2}\binom{m}{i}$. Thus, in this case as $m$ increases we get that $\frac{\sum_{i=1}^{m/2}\binom{m}{i}}{2^{m}}=1/2$. By doing so we do not allocate bins that have an average guesswork which is smaller than the total average guesswork. Furthermore, in our analysis we restrict the size of the input, $n$, to be proportional to $m$. However, in practice $n$ can be much larger. This can be achieved by simply duplicating the mappings for the first $\pa{H\pa{s}+D\pa{s||p}}\cdot m$ inputs over and over again. By doing so the number of inputs can increase as much as needed, whereas the size of the key and the total average guesswork remain the same. Finally, by allocating a bin to each user and performing the back door procedure in Definition \ref{def:backdoor}, the chance of collision between any two users is the same as the chance of collision for a strongly universal hash function with a uniform key, that has the same average guesswork, that is, $2^{n}$.
\end{rem}

We now show that when the average guesswork is larger than the number of users, a biased key which ``beamforms'' to the subset of users, requires a key of smaller size than an unbiased key that achieves the same level of security.
\begin{theorem}
When the number of users is $2^{H\pa{s}\cdot m-1}$, $1/2\le s\le 1$, any universal set of hash functions for which $f\pa{\cdot}$ is a bijection, requires an unbiased key $k_{u}$ (i.e., which is drawn i.i.d. Bernoulli$\pa{1/2}$) of size
\beq{}
{
|k_{u}|=\pa{H\pa{s}+D\pa{s||p}}\cdot m\cdot 2^{\pa{H\pa{s}+D\pa{s||p}}\cdot m}
}
and output of size $|m_{u}|=\pa{H\pa{s}+D\pa{s||p}}\cdot m$, in order to achieve an average guesswork for which
\beqn{}
{
\lim_{m\to\infty}\frac{E\pa{G_{remote}^{\pa{u}}\pa{B,f\pa{\cdot}}}}{2^{\pa{H\pa{s}+D\pa{s||p}}\cdot m}} = 1.
}
On the other hand when the key is drawn Bernoulli$\pa{p}$ the universal set of hash functions defined in equations \eqref{eq:KeySegmentation}, \eqref{eq:USHF}, achieves the same average guesswork with a key $k_{b}$ of size
\beq{}
{
|k_{b}| = m\cdot 2^{\pa{H\pa{s}+D\pa{s||p}}\cdot m}
}
and an output of size $m$. In both cases the minimal input is of size $n=\pa{1+\eps_{1}}\cdot\pa{H\pa{s}+D\pa{s||p}}\cdot m$. The average size of the biased key can be decreased further to $H\pa{p}\cdot |k_{b}|$ through entropy encoding.

Finally, distributing passwords according to the procedure in Definition \ref{def:backdoor} requires an unbiased key of size
\beq{}
{
\pa{1+\eps_{1}}\cdot\pa{H\pa{s}+D\pa{s||p}}\cdot m\cdot 2^{H\pa{s}\cdot m}.
}

\end{theorem}
\begin{proof}
From Corollary \ref{cor:RelationBetweenRateOfAverageGuessworkBiasandUnbias} and Lemma \ref{lem:ActualLowerBoundAverageGuesswork} we get that when the key is drawn i.i.d. Bernoulli$\pa{1/2}$, the average guesswork is determined by the size of the output. Therefore, in order to achieve an average guesswork that equals $2^{m\cdot\pa{H\pa{s}+D\pa{s||p}}}$ the size of the output must be $\pa{H\pa{s}+D\pa{s||p}}\cdot m$. In addition $n=\pa{1+\eps_{1}\cdot}\pa{H\pa{s}+D\pa{s||p}}$ and so the size of the key is equal to $|k_{u}|=\pa{\hdps}\cdot m\cdot 2^{\pa{1+\eps_{1}}\cdot m\cdot\pa{\hdps}}$.

For the case when the key is drawn i.i.d. Bernoulli$\pa{p}$, we get from Lemma \ref{lem:ActualLowerBoundAverageGuesswork} that an output of size $m$ achieves the desired guesswork. Hence, in this case the size of the key is $|k_{b}|=m\cdot 2^{\pa{1+\eps_{1}}\cdot m\cdot\pa{\hdps}}$

Finally, in order to allocate bins according to the method in Definition \ref{def:backdoor} each user has a key of size $n$, and since there are $2^{H\pa{s}\cdot m}$ users the size of the key required to draw passwords is $\pa{1+\eps_{1}}\cdot\pa{H\pa{s}+D\pa{s||p}}\cdot m\cdot 2^{H\pa{s}\cdot m}$.
\end{proof}

\begin{cor}\label{cor:Thekeysratio}
When the number of users is equal to $2^{H\pa{s}\cdot m-1}$ and the average guesswork increases at rate $\pa{H\pa{s}+D\pa{s||p}}$, the ratio between the size of an unbiased key and a biased key which is drawn Bernoulli$\pa{p}$, when both keys achieve the average guesswork above, is
\beqn{}
{
\lim_{m\to\infty}|k_{u}|/|k_{b}|=H\pa{s}+D\pa{s||p}.
}
\end{cor}
\begin{proof}
The proof results from the fact that
\beqn{}
{
\lim_{m\to\infty}\frac{\pa{\hdps}\cdot m\cdot 2^{\pa{1+\eps_{1}}\cdot m\cdot\pa{\hdps}}}{\pa{1+\eps_{1}}\cdot\pa{H\pa{s}+D\pa{s||p}}\cdot m\cdot 2^{H\pa{s}\cdot m}+ m\cdot 2^{\pa{1+\eps_{1}}\cdot m\cdot\pa{\hdps}}}=\hdps .
}
\end{proof}

In the next corollary we find the minimal size of a key that enables to achieve through bias, a desired average guesswork that is larger than the number of users.
\begin{cor}\label{cor:EngOptimizationFortheKeysRatio}
When there are $2^{m-1}$ users and the desired average guesswork increases at rate $\Al>1$, the key of minimal size is drawn Bernoulli$\pa{p_{0}}$ such that
\beq{}
{
1+D\pa{1/2||p_{0}}=\Al.
}
A universal set of hash functions defined in equations \eqref{eq:KeySegmentation}, \eqref{eq:USHF}, whose output is of size $m$ achieves the average guesswork above.
\end{cor}
\begin{proof}
From Corollary \ref{cor:RelationBetweenRateOfAverageGuessworkBiasandUnbias} and Lemma \ref{lem:ActualLowerBoundAverageGuesswork} we know that when $\Al > 1$ and the key is i.i.d. Bernoulli$\pa{1/2}$, the size of the output has to be $\Al\cdot m$ in order to achieve an average guesswork $2^{\Al\cdot m}$. Since in this case the size of the input has to be larger than $\Al\cdot m$, the size of an unbiased key that achieves the desired average guesswork is $\Al\cdot m\cdot 2^{\Al\cdot m}$.

From Lemma \ref{lem:ActualLowerBoundAverageGuesswork} we also know that for any biased key that achieves average guesswork $2^{\Al\cdot m}$ the size of the input has to be at least $\Al\cdot m$. However, the size of the output can in fact be smaller than $\Al\cdot m$. Essentially, we can define the output as $m^{\prime}$ such that number of users $2^{m-1}=2^{H\pa{s_{0}}\cdot m^{\prime}-1}$ and $m^{\prime}<\Al\cdot m$, where $1/2\le s_{0}\le 1$. The size of the key in this case is equal to $m^{\prime}\cdot 2^{\Al\cdot m}$. The minimal value of $m^{\prime}$ that still enables to allocate a different bin to each user is $m^{\prime}=m$ (i.e., $s=1/2$). In this case the size of the key is $m\cdot 2^{\Al\cdot m}$; this key is drawn i.i.d. Bernoulli$\pa{p_{0}}$ such that
\beqn{}
{
1+D\pa{1/2||p_{0}}=\Al.
}
\end{proof}

\begin{theorem}\label{th:TheSecondMainTheoremKeyedHashFunctions}
When the number of users $M=2^{m-1}$, a uniform key that achieves average guesswork $2^{m\cdot\Al}$, where $\Al>1$, with any mapping from the set of keys to the set of all hash functions, is $\Al$ times larger than a biased key which is i.i.d. Bernoulli$\pa{p_{0}}$ that achieves the same average guesswork with the universal set of hash functions defined in subsection \ref{subsec:SUSHF} equations \eqref{eq:KeySegmentation}, \eqref{eq:USHF}, where $p_{0}$ satisfies
\beq{}
{
1+D\pa{1/2||p_{0}}=\Al.
}
\end{theorem}
\begin{proof}
The proof follows directly from Corollary \ref{cor:Thekeysratio} and Corollary \ref{cor:EngOptimizationFortheKeysRatio} of section \ref{sec:Discussion}.
\end{proof}

Figure \ref{fig:Key_Size} illustrates Theorem \ref{th:TheSecondMainTheoremKeyedHashFunctions}.

\begin{figure}
\input{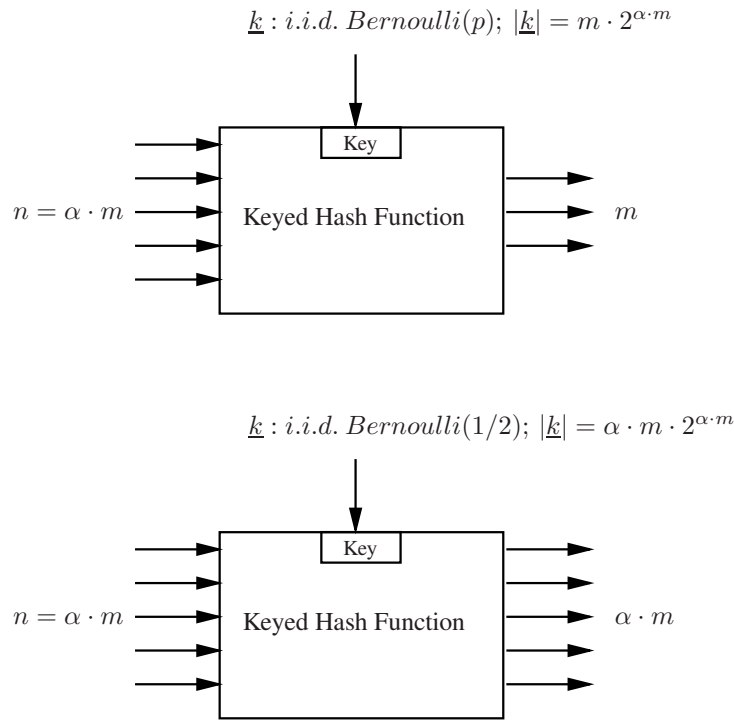}\caption{When the number of users is $2^{m-1}$, and the average guesswork is equal to $2^{\Al\cdot m}$, the size of a biased key is $\Al>1$ times smaller than a uniform key, where  $\Al=1+D\pa{1/2||p}$. Note that in order to achieve an average guesswork $2^{\Al\cdot m}$, the output of a strongly universal set of hash functions has to be of size $\Al\cdot m$ when the key is unbiased, whereas when the key is biased the output can reduce to $m$ bits, as long as the number of users is no larger than $2^{m-1}$.}\label{fig:Key_Size}
\end{figure}

Now, we show that biased keys also require smaller space to store the bins on the server.
\begin{cor}
When the number of users is $2^{H\pa{s}\cdot m-1}$ and the average guesswork is equal to $2^{m\cdot\pa{H\pa{s}+D\pa{s||p}}}$, a biased key which is drawn Bernoulli$\pa{p}$ requires to store $2^{H\pa{s}\cdot m-1}$ bins of size $m$ each, whereas an unbiased key requires to store $2^{H\pa{s}\cdot m-1}$ bins of size $\pa{H\pa{s}+D\pa{s||p}}\cdot m$. Therefore, under the constraints above a biased key decreases the storage space by a factor of $\pa{H\pa{s}+D\pa{s||p}}$. The average stored space can be reduced further to $H\pa{p}\cdot m\cdot 2^{H\pa{s}\cdot m-1}$ through entropy encoding.
\end{cor}
\begin{proof}
The proof is straight forward and results from the fact that when considering a biased key the output is $\hdps$ times smaller than the output when the key is unbiased.
\end{proof}

A remark is in order regarding the results above.
\begin{rem}
Note that the results above apply to strongly universal sets of hash functions, which are balanced sets. On the other hand in sections \ref{sec:GuessworkAnyHash} and \ref{sec:BoundsAverageGuessworkNoBinsAllocation} we consider sets of hash functions which are not balanced, that is,  the size of a conditional set depends on the actual values of the bins on which it is conditioned; for example, if a fraction of $\pa{1-P_{K}\pa{b_{0}}}$ mappings is revealed such that none of these mappings map to $b_{0}$, then the remaining fraction of $P_{k}\pa{b_{0}}$ mappings are all mapped to $b_{0}$. Hence, in this case the conditional set is of size $1$.
\end{rem}

\section{Proofs of the Results of Section \ref{sec:ModifiedHashFunctionAverageGueddwork}}\label{sec:GuessworkAnyHash}

We first calculate the guesswork when bins are not allocated to the users, and for any bin the attacker knows of a password that is mapped to it.
\begin{theorem}[Guesswork of a broken hash function]\label{th:AverageGuessworkAnyHashFunction}
When each user draws his password uniformly over $\pa{1,\dots,2^{n}}$ as in Definition \ref{def:DIstributedPasswordsAllocation}, and for every bin the attacker knows of a password that is mapped to it, the $\rho$th moment of the guesswork of every user is
\beq{}
{
E\pa{\pa{G_{H}\pa{B}}^{\rho}}\ge \frac{1}{\pa{1+m\cdot\ln\pa{2}}^\rho}\pa{\sum_{b=1}^{2^{m}}\pa{\ph}^{1/\pa{1+\rho}}}^{1+\rho}.
}
When $P_{H}$ is defined based on the method of types as presented in Definition \ref{def:DefinitionOfIidCase}, the $\rho$th moment of the guesswork increases at rate that is equal to
\beq{}
{
\lim_{m\to\infty}\frac{1}{m}\log\pa{E\pa{\pa{G_{H}\pa{B}}^{\rho}}}=H_{1/\pa{1+\rho}}\pa{p}
}
where $H_{1/\pa{1+\rho}}\pa{p}$ is the R\'{e}nyi entropy which is defined in equation \eqref{eq:RenyiEntropyFunctionofRho}.
\end{theorem}
\begin{proof}
We assume that for every bin the attacker knows of a password that is mapped to it. Therefore, the problem of guessing a password reduces to the problem of guessing the correct bin for each user. Since each user chooses the password uniformly over $\pac{1,\dots, 2^{n}}$, the chance of it being mapped to bin $b\in\pac{1,\dots,2^{m}}$ is $\ph$. Hence, the problem is reduced to the guesswork problem that is analyzed in \cite{Arikan_Ineq_Guessing}. We get the average guesswork that is stated above by following the same lines as \cite{Arikan_Ineq_Guessing} .
\end{proof}

\begin{rem}
Essentially, the result of Theorem \ref{th:AverageGuessworkAnyHashFunction} captures the average time required to crack a hash function when a \emph{key stretching mechanism} \cite{WagnerKeyStreatching} is used in order to protect a system against passwords cracking. For example, when $\rho=2$, it means that the time that elapses between attempts increases quadratically.
\end{rem}

\begin{rem}
When passwords are not drawn uniformly, the optimal strategy of guessing passwords one by one is to first calculate the probability of each bin by summing the probabilities over all passwords that are mapped to this bin, and then guessing bins according to their probabilities in descending order.
\end{rem}

We begin by defining a back door procedure for any hash function.
\begin{defn}\label{def:BackdoorForAnyHash}
The back door mechanism for modifying a hash function is defined as follows.
\bitm
{
\item For each user choose a different bin according to the procedure described in Definition \ref{def:PasswordAllocationwithNoDetails} (i.e., the first user gets the least likely bin $b_{1}$, the second user receives the second least likely bin $b_{2}$, and the $2^{H\pa{s}\cdot m -1}$th user receives $b_{2^{H\pa{s}\cdot m -1}}$).
\item Each of the users from the first to the $2^{H\pa{s}\cdot m -1}$th draws a password uniformly by drawing $n$ bits i.i.d. Bernoulli$\pa{1/2}$.
\item For each user, if the number that was drawn $g\in\pac{1,\dots,2^{n}}$ has not been drawn by any of the users that have bins that are less likely than the current bin, then replace the mapping of $g$ with a mapping from $g$ to the bin assigned to the user.
\item In the case when another user who is assigned to a less likely bin, has drawn $g$: \emph{Do not change} the value of the mapping of $g$ again (i.e., the mapping from $g$ is changed only once by the user who is coupled with the least likely bin among the bins of the users whose password is $g$). Instead, change the bin allocated to the user to the least likely bin among the bins allocated to users whose password is $g$. Therefore, the mapping from $g$ is changed only once, to the least likely bin that is mapped to $g$.
}
\end{defn}

We are now ready to prove Theorem \ref{th:LowerBoundExpAverageGuessworkAnyHashFunction}.
\begin{proof}[The proof of Theorem \ref{th:LowerBoundExpAverageGuessworkAnyHashFunction}]
The underlying idea behind the proof is to lower bound the average guesswork by upper bounding the probability of finding a password that cracks the bin at each round; we upper bound the probability by incorporating the maximal number of mappings that the backdoor mechanism can add to each bin as well as the number of unsuccessful guesses that have been made so far. We start by proving this theorem for any $P_{H}$-hash function when averaged over all possible strategies of guessing passwords one by one, and we then show this also for the $P_{H}$-set of hash functions and for any strategy. Then, in order to prove the equality rather than just an upper bound, we show that the probability that the number of mappings eliminated by the backdoor mechanism affects the fraction of mappings, is small to the extent that does not allow the average to grow beyond the lower bound discussed above.

Since the attacker does not know the mappings of the hash function, there is no reason for it to favor any guessing strategy over the other. Hence, we average over all strategies of guessing passwords one by one; essentially, this is equivalent to averaging over all possible permutations of the inputs, assuming that the permutations are uniformly distributed. This can be achieved by first drawing a random variable uniformly over $\pac{1,\dots, 2^{n}}$, then drawing a random variable over the remaining $2^{n}-1$ possibilities, etc.

We now wish to lower bound the average guesswork by upper bounding $\ph$. In order to lower bound the average guesswork we consider only the first $\pa{1+\eps_{1}}\cdot m\cdot \log\pa{1/p}$ guesses when averaging, where $0<\eps_{1}<\eps$. Clearly the number of guesses can be greater than the term above and may go up to $2^{n}$. Furthermore, we observe that after $k$ unsuccessful guesses, the chance of guessing bin $b$ is $\ph\cdot\frac{2^{n}}{2^{n}-k}$ where $0\le k\le \pa{1-\ph}2^{n}$. In addition, the back door procedure of Definition \ref{def:BackdoorForAnyHash} can add one mapping at most to each bin. Thus, we can upper bound the probability by
\beq{eq:UpperBoundOnTheProbabilityOfAnyHashFunction}
{
\ph\le \frac{\ph\cdot 2^{n}+1}{2^{n}-k}=P_{H}^{\pa{k}}\pa{b}\quad k\in\pac{0,\dots, 2^{\pa{1+\eps_{1}}\cdot m\cdot \log\pa{1/p}}}.
}
Note that $\lim_{m\to\infty}\frac{2^{\pa{1+\eps_{1}}\cdot m\cdot \log\pa{1/p}}}{\pa{1-\ph}2^{n}}=0$ since $\eps_{1}<\eps$.

The average guesswork of each bin can be lower bounded by
\beq{eq:LowerBoundByGeometricalDistAvrtageGusswork}
{
E\pa{\Gbin}\ge \ph\cdot\pa{1+ \sum_{i=1}^{2^{m_{0}}-1}\pa{i+1}\cdot\prod_{j=1}^{i}\pa{1-P_{H}^{\pa{j}}\pa{b}}}\ge \ph\cdot\sum_{i=0}^{2^{m_{0}}-1}\pa{i+1}\cdot\pa{1-P_{H}^{\pa{2^{m_{0}}}}\pa{b}}^{i}
}
where $m_{0}=\pa{1+\eps_{1}}\cdot m\cdot \log\pa{1/p}$; the first inequality is due to equation \eqref{eq:UpperBoundOnTheProbabilityOfAnyHashFunction}, whereas the second one is because of the fact that $P_{H}^{\pa{2^{m_{0}}}}\pa{b}\ge P_{H}^{\pa{i}}\pa{b}$ for any $i\in\pac{1,\dots, 2^{m_{0}}-1}$. Hence, we get
\beq{eq:LowerBoundAverageGuessworkOfAnyHashFunctionAnyBin}
{
E\pa{\Gbin}\ge\frac{\ph}{P_{H}^{\pa{2^{m_{0}}}}\pa{b}}\cdot \pa{1/P_{H}^{\pa{2^{m_{0}}}}\pa{b}-\pa{1-P_{H}^{\pa{2^{m_{0}}}}\pa{b}}^{2^{m_{0}}}\cdot\pa{1/P_{H}^{\pa{2^{m_{0}}}}\pa{b}+2^{m_{0}}}}.
}

Now, let us break equation \eqref{eq:LowerBoundAverageGuessworkOfAnyHashFunctionAnyBin} down. First, we know that
\beq{}
{
P_{H}^{\pa{2^{m_{0}}}}\pa{b} = \frac{\ph\cdot 2^{\pa{1+\eps}\cdot m\cdot\log\pa{1/p}}+1}{2^{\pa{1+\eps}\cdot m\cdot\log\pa{1/p}}-2^{\pa{1+\eps_{1}}\cdot m\cdot\log\pa{1/p}}}=\frac{\ph+2^{-\pa{1+\eps}\cdot m\cdot\log\pa{1/p}}}{1-2^{-\pa{\eps-\eps_{1}}\cdot\log\pa{1/p}\cdot m}}
}
In addition, since $\ph\ge 2^{-m\cdot\log\pa{1/p}}$ for any $b\in\pac{1,\dots, 2^{m}}$ we get
\beq{}
{
\ph\le P_{H}^{\pa{2^{m_{0}}}}\pa{b}\le \ph\cdot \frac{1+2^{-\eps\cdot \log\pa{1/p}\cdot m}}{1-2^{-\pa{\eps-\eps_{1}}\cdot\log\pa{1/p}\cdot m}}\quad\forall b\in\pac{1,\dots,2^{m}}.
}

From Corollary \ref{cor:MeanValueUniformlyConverges} we get that
\beq{}
{
\lim_{m\to\infty}\pa{1-P_{H}^{\pa{2^{m_{0}}}}\pa{b}}^{2^{m_{0}}}\cdot\pa{1/P_{H}^{\pa{2^{m_{0}}}}\pa{b}+2^{m_{0}}}=0.
}
Further,
\beq{eq:LowerBoundAverageGuessworkPerBinFOrPHSET}
{
1/P_{H}^{\pa{2^{m_{0}}}}\pa{b}\ge 1/\ph\cdot \frac{1-2^{-\pa{\eps-\eps_{1}}\cdot\log\pa{1/p}\cdot m}}{1+2^{-\eps\cdot \log\pa{1/p}\cdot m}}\quad\forall b\in\pac{1,\dots,2^{m}}
}
Hence, by averaging over all the types of the assigned bins similarly to what is done in Theorem \ref{th:TheAverageGuesseork} we get the inequality
\beq{eq:TheFinalInequalityForAvergeGuessworkAnyHash}
{
\lim_{m\to\infty}\frac{1}{m}\log\pa{E\pa{\Gon}}\ge\barr{cc}{H\pa{s}+D\pa{s||p} &\pa{1-p}\le s\le 1\\ 2\cdot H\pa{p}+D\pa{1-p||p}-H\pa{s} &1/2\le s\le \pa{1-p} }.
}

The equality is achieved based on the following arguments. Assume that the backdoor mechanism eliminates $l$ mappings to bin $b$. In order for the inequality in \eqref{eq:TheFinalInequalityForAvergeGuessworkAnyHash} to turn to equality, $l$ has to increase at a rate that is smaller than $2^{n}\cdot\ph$. Therefore, since there are $2^{H\pa{s}\cdot m-1}$ users, whenever $n\ge\pa{1+\eps}\cdot\pa{\log\pa{1/p}+H\pa{s}}$, $l$ cannot affect the rate at which $2^{n}\cdot\ph$ increases.


Furthermore, this theorem also applies to the case when averaging over the $P_{H}$-set of hash functions for any strategy of guessing passwords one by one (i.e., for any fixed strategy, the average is performed over the set of hash functions). This is due to the fact that any $P_{H}$-hash function can be represented by a bipartite graph (e.g., Figure \ref{fig:Intuitive_Explanation}) for which averaging over the order according to which nodes that represent the domain, are chosen while keeping the edges fixed, leads to the same result as averaging over the order of the edges  connected to these nodes while choosing the nodes in any order. The effect of the backdoor procedure given in Definition \ref{def:BackdoorForAnyHash} is similar. When considering a user who is mapped to bin $b\in B$, the other users can only decrease the number of mappings to this bin. Based on the arguments presented in this proof, the other users cannot increase the average guesswork. Furthermore, the user can add up to one mapping, which does not affect the average guesswork.
\end{proof}

\begin{proof}[The proof of Corollary \ref{cor:AverageGuessworkAnyHashFunctionAnyBin}]
The proof is straightforward based on the proof of Theorem \ref{th:LowerBoundExpAverageGuessworkAnyHashFunction}.
\end{proof}

Next, we prove Theorem \ref{th:ConcentrationAnyHash}.
\begin{proof}[The proof of Theorem \ref{th:ConcentrationAnyHash}]
We are interested in upper bounding the probability
\beq{eq:ProbAyFuncGuessinglessThanMinimalAverage}
{
Pr\pa{\Gbin\le 2^{\pa{1-\eps_{1}}\cdot\pa{\hdps}\cdot m}}=P_{H}^{\pa{0}}\pa{b}+\sum_{i=1}^{2^{\pa{1-\eps_{1}}\cdot\pa{\hdps}\cdot m}}P_{H}^{\pa{i}}\pa{b}\cdot\prod_{j=0}^{i-1}\pa{1-P_{H}^{\pa{j}}\pa{b}}
}
where
\beq{}
{
P_{H}^{\pa{i}}\pa{b}=\frac{P_{H}\pa{b}\cdot 2^{n}+1}{2^{n}-i}.
}
As $P_{H}\pa{b}<P_{H}^{\pa{0}}\pa{b}<\dots <P_{H}^{\pa{2^{m_{1}}}}\pa{b}$ where $m_{1}=\pa{1-\eps_{1}}\cdot\pa{\hdps}\cdot m$, we can upper bound the expression in \eqref{eq:ProbAyFuncGuessinglessThanMinimalAverage} by
\beq{}
{
Pr\pa{\Gbin\le 2^{\pa{1-\eps_{1}}\cdot\pa{\hdps}\cdot m}}\le \sum_{i=0}^{2^{m_{1}}-1}P_{H}^{\pa{2^{m_{1}}}}\pa{b}\cdot\pa{1-P_{H}\pa{b}}^{i}=\frac{P_{H}^{\pa{2^{m_{1}}}}}{P_{H}\pa{b}}\cdot\pa{1-\pa{1-P_{H}\pa{b}}^{2^{m_{1}}}}.
}
Following along the same lines as the proof of Theorem \ref{th:LowerBoundExpAverageGuessworkAnyHashFunction} we know that the ratio $\lim_{m\to\infty}\frac{P_{H}^{\pa{2^{m_{1}}}}}{P_{H}\pa{b}}=1$, and because of the fact that $P_{H}\pa{b}\le 2^{-\pa{\hdps}\cdot m}$ for any $b\in B$ we get
\beq{}
{
\lim_{m\to\infty}1-\pa{1-P_{H}\pa{b}}^{2^{\pa{1-\eps_{1}}\cdot\pa{\hdps}\cdot m}}\le \lim_{m\to\infty}1-\pa{1-2^{-\pa{\hdps}\cdot m}}^{2^{\pa{1-\eps}\cdot\pa{\hdps}\cdot m}}=0
}
for any $b\in B$.

Note that the result above also applies to the case when
\beq{eq:ConcentrationwithBinAllocationfotEveryBin}
{
Pr\pa{\Gbin\le 2^{\pa{1-\eps_{1}}\cdot\pa{H\pa{q\pa{b}}+D\pa{q\pa{b}||p}}\cdot m}}\le \lim_{m\to\infty}1-\pa{1-P_{H}\pa{b}}^{2^{\pa{1-\eps_{1}}\cdot\pa{H\pa{q\pa{b}}+D\pa{q\pa{b}||p}}\cdot m}}=0
}

This result also applies to the case when averaging over the $P_{H}$-set of hash functions and considering any strategy of guessing passwords one by one. It results from the same arguments that are given in Theorem \ref{th:LowerBoundExpAverageGuessworkAnyHashFunction}; the backdoor mechanism can add no more than one mapping which is drawn uniformly over the set of passwords. Therefore, the same bounds and arguments hold for this case as well. This concludes the proof.
\end{proof}

\begin{proof}[The proof of Corollary \ref{cor:EffectofBiasedPassword}]
The proof of this corollary relies on the concentration property of the average guesswork of every bin presented in the proof of Theorem \ref{th:ConcentrationAnyHash}, along with the fact that the rate at which the average guesswork of a password increases is dominated by elements of type $\frac{\sqrt{\qp}}{\sqrt{\qp}+\sqrt{1-\qp}}$ \cite{MaloneGuessworkandEntropy}.

We consider the case where $n\ge\pa{1+\eps}\cdot m\cdot \pa{\log\pa{1/p}+H\pa{s}}$ as for this case the average guesswork of bin $b$ cannot increase due to mappings that are removed by other users (for more details see the proof of Theorem \ref{th:LowerBoundExpAverageGuessworkAnyHashFunction}).

We denote the average guesswork of bin $b\in B$ when averaged over the $P_{H}$-set of hash functions by $G^{\pa{H_{F}}}\pa{b}$. From equation \eqref{eq:LowerBoundAverageGuessworkPerBinFOrPHSET} along with the fact that  $n\ge\pa{1+\eps}\cdot m\cdot \pa{\log\pa{1/p}+H\pa{s}}$, we know that
\beq{}
{
\lim_{m\to\infty}\frac{1}{m}\log \pa{E\pa{G^{\pa{H_{F}}}\pa{b}}}=H\pa{q\pa{b}}+D\pa{q\pa{b}||p}.
}
Furthermore, let us denote the average guesswork of a password that is drawn i.i.d. Bernoulli$\pa{\qp}$ by $G^{\pa{ps}}\pa{\qp}$ whose rate is \cite{Arikan_Ineq_Guessing}
\beq{}
{
\lim_{n\to\infty}\frac{1}{n}E\pa{G^{\pa{ps}}\pa{\qp}}=H_{1/2}\pa{\qp}.
}

The average guesswork of bin $b\in B$ is upper bounded by
\beq{eq:BiassedPasswordTheMinimumUpperBoundonAverageGuesswork}
{
\pa{E\pa{\Gbin}}\le\min\pa{E\pa{G^{\pa{H_{F}}}\pa{b}},E\pa{G^{\pa{ps}}\pa{\qp}}}
}
because of the fact that the guesswork $\Gbin=\min\pa{G^{\pa{H_{F}}}\pa{b},G^{\pa{ps}}\pa{\qp}}$.

We begin be considering the case when $\lim_{m\to\infty}2\cdot \frac{n}{m}H\pa{\frac{\sqrt{\qp}}{\sqrt{\qp}+\sqrt{1-\qp}}}\le\pa{1-\eps_{2}}\cdot \pa{H\pa{q\pa{b}}+D\pa{q\pa{b}||p}}$ for any $0<\eps_{2}<1$. We wish to show that in this case the rate at which the average guesswork increases is equal to $\lim_{m\to\infty}\frac{n}{m}H_{1/2}\pa{\qp}$. First, note that
\beq{eq:BiasedPasswordsCondProbEqualsTheBiasedPassword}
{
Pr\pa{\Gbin=i|G^{\pa{H_{F}}}\pa{b}\ge 2^{\pa{1-\eps_{2}}\cdot \pa{H\pa{q\pa{b}}+D\pa{q\pa{b}||p}}\cdot m}}=Pr\pa{G^{\pa{ps}}\pa{\qp}=i}
}
when $i\in\pac{0,\dots, 2^{\pa{1-\eps_{2}}\cdot \pa{H\pa{q\pa{b}}+D\pa{q\pa{b}||p}}\cdot m}}$ because $G^{\pa{H_{F}}}\pa{b}$ and $G^{\pa{ps}}\pa{\qp}$ are independent as well as $\Gbin=\min\pa{G^{\pa{H_{F}}}\pa{b},G^{\pa{ps}}\pa{\qp}}$. In \cite{MaloneGuessworkandEntropy} it was shown based on large deviation arguments that the rate at which the average guesswork of a password that is drawn i.i.d. Bernoulli$\pa{\qp}$ increases is dominated by elements of type $\frac{\sqrt{\qp}}{\sqrt{\qp}+\sqrt{1-\qp}}$;  and when guessing passwords according to their probabilities in descending order the number of guesses that include all elements of this type increases like $2^{2\cdot n\cdot H\pa{\frac{\sqrt{\qp}}{\sqrt{\qp}+\sqrt{1-\qp}}}}$. Hence, based on this fact and \eqref{eq:BiasedPasswordsCondProbEqualsTheBiasedPassword} we get that the rate at which the conditional average guesswork increases is lower bounded by
\beq{}
{
\lim_{m\to\infty}\frac{1}{m}\log\pa{E\pa{G_{bin}\pa{b|G^{\pa{H_{F}}}\pa{b}\ge 2^{\pa{1-\eps_{2}}\cdot \pa{H\pa{q\pa{b}}+D\pa{q\pa{b}||p}}\cdot m}}}
} \ge \frac{n}{m}H_{1/2}\pa{\qp}}
when $\lim_{m\to\infty}2\cdot \frac{n}{m}H\pa{\frac{\sqrt{\qp}}{\sqrt{\qp}+\sqrt{1-\qp}}}\le\pa{1-\eps_{2}}\cdot \pa{H\pa{q\pa{b}}+D\pa{q\pa{b}||p}}$. In addition, we know from equation \eqref{eq:ConcentrationwithBinAllocationfotEveryBin} that
\beq{}
{
\lim_{m\to\infty}\frac{1}{m}\log\pa{Pr\pa{G^{\pa{H_{F}}}\pa{b}\ge 2^{\pa{1-\eps_{2}}\cdot \pa{H\pa{q\pa{b}}+D\pa{q\pa{b}||p}}\cdot m}}
} = 0
}
Therefore, we can lower bound the rate at which the average guesswork of bin $b$ increases by
\bal{eq:BiasPasswordsLowerBoundBiassedPassword}
{
\lim_{m\to\infty}\frac{1}{m}\log\pa{E\pa{\Gbin}}\ge\frac{1}{m}\log\pa{E\pa{G_{bin}\pa{b|G^{\pa{H_{F}}}\pa{b}\ge 2^{\pa{1-\eps_{2}}\cdot \pa{H\pa{q\pa{b}}+D\pa{q\pa{b}||p}}}}}}\times\nn\\
 Pr\pa{G^{\pa{H_{F}}}\pa{b}\ge 2^{\pa{1-\eps_{2}}\cdot \pa{H\pa{q\pa{b}}+D\pa{q\pa{b}||p}}}} \ge \frac{n}{m}H_{1/2}\pa{\qp}.
}
Hence, based on \eqref{eq:BiassedPasswordTheMinimumUpperBoundonAverageGuesswork} and \eqref{eq:BiasPasswordsLowerBoundBiassedPassword} we get
\bal{}
{
\lim_{m\to\infty}\frac{1}{m}\log\pa{E\pa{\Gbin}}= \frac{n}{m}H_{1/2}\pa{\qp}.
}

Next we wish to find the average guesswork when $H\pa{q\pa{b}}+D\pa{q\pa{b}||p}\le\pa{1-\eps_{2}}\cdot \frac{n}{m}H\pa{\qp}$
for any $0<\eps_{2}<1$. We wish to show that in this case the rate at which the average guesswork increases is equal to $H\pa{q\pa{b}}+D\pa{q\pa{b}||p}$. First, note that
\beq{eq:BiasedPasswordsCondProbEqualsTheHash}
{
Pr\pa{\Gbin=i|G^{\pa{ps}}\pa{\qp}> 2^{\pa{1-\eps_{2}}\cdot \pa{H\pa{\qp}}\cdot n}}=Pr\pa{G^{\pa{H_{F}}}\pa{b}=i}
}
when $i\in\pac{0,\dots, 2^{\pa{1-\eps_{2}}\cdot \pa{H\pa{\qp}}\cdot n}}$ because $G^{\pa{H_{F}}}\pa{b}$ and $G^{\pa{ps}}\pa{\qp}$ are independent as well as $\Gbin=\min\pa{G^{\pa{H_{F}}}\pa{b},G^{\pa{ps}}\pa{\qp}}$.  Based on equation \eqref{eq:ConcentrationwithBinAllocationfotEveryBin} it can be easily shown that the average guesswork of the $P_{H}$-set of hash functions is concentrated around its mean value.  Hence, based on this fact and \eqref{eq:BiasedPasswordsCondProbEqualsTheHash} we get that the rate at which the conditional average guesswork increases is lower bounded by
\beq{}
{
\lim_{m\to\infty}\frac{1}{m}\log\pa{E\pa{G_{bin}\pa{b|G^{\pa{ps}}\pa{\qp}\ge 2^{\pa{1-\eps_{2}}\cdot H\pa{\qp}\cdot n}}}
} \ge H\pa{q\pa{b}}+D\pa{q\pa{b}||p}}
when $H\pa{q\pa{b}}+D\pa{q\pa{b}||p}\le\pa{1-\eps_{2}}\cdot \frac{n}{m}H\pa{\qp}$. In addition
\beq{}
{
\lim_{n\to\infty}\frac{1}{n}\log\pa{Pr\pa{G^{\pa{ps}}\pa{\qp}\ge 2^{\pa{1-\eps_{2}}\cdot H\pa{\qp}\cdot n}}
} = 0.
}
Therefore, we can lower bound the rate at which the average guesswork of bin $b$ increases by
\bal{eq:BiasPasswordsLowerBoundHash}
{
\lim_{m\to\infty}\frac{1}{m}\log\pa{E\pa{\Gbin}}\ge\frac{1}{m}\log\pa{E\pa{G_{bin}\pa{b|G^{\pa{ps}}\pa{\qp}\ge 2^{\pa{1-\eps_{2}}\cdot H\pa{\qp}\cdot n}}}}\times\nn\\
 Pr\pa{G^{\pa{ps}}\pa{\qp}\ge 2^{\pa{1-\eps_{2}}\cdot H\pa{\qp}\cdot n}} \ge H\pa{q\pa{b}}+D\pa{q\pa{b}||p}.
}
Hence, based on \eqref{eq:BiassedPasswordTheMinimumUpperBoundonAverageGuesswork} and \eqref{eq:BiasPasswordsLowerBoundHash} we get
\bal{}
{
\lim_{m\to\infty}\frac{1}{m}\log\pa{E\pa{\Gbin}}= H\pa{q\pa{b}}+D\pa{q\pa{b}||p}.
}

Equations \eqref{eq:AverageAcrossUsersPasswordBiased1} and \eqref{eq:AverageAcrossUsersPasswordBiased2} follow in a straightforward fashion when either $2^{n\cdot H\pa{\frac{\sqrt{\qp}}{\sqrt{\qp}+\sqrt{1-\qp}}}}$  is smaller than the minimal average guesswork, that is, smaller than $2^{m\cdot\pa{\hdps}}$; or when $\frac{n}{m}\cdot H\pa{\qp}$ is larger than the maximal rate at which the average guesswork increases (i.e., larger than $\log\pa{1/p}$).
\end{proof}

We now prove Corollary \ref{cor:TheAverageGuessworkForGuessingAnyPassword}.
\begin{proof}[The proof of Corollary \ref{cor:TheAverageGuessworkForGuessingAnyPassword}]
In order to calculate the average guesswork for guessing a password that is mapped to any bin that is allocated to any user, we first need to bound the probability $P_{B}$ which is the probability of guessing such a password in the first guess (i.e., the fraction of mappings to any of those bins). Since $P_{B}$ is the probability of guessing a password that is mapped to any of the assigned bins, we get
\beq{}
{
P_{B}=\sum_{b\in B}\ph
}
where $B$ is the set of assigned bins as in Definition \ref{def:PasswordAllocationwithNoDetails}.
Therefore, based on Lemma \ref{th:methodoftypesprobth}, Lemma \ref{th:methodoftypessize}, the fact that there are $m$ types, and the fact that the number of users $M=2^{H\pa{s}\cdot m-1}<2^{H\pa{s}\cdot m}$ as well as the fact that $\max_{1/2\le s\le q\le 1}D\pa{q||p} = D\pa{s||p}$ when $p\le 1/2$; $P_{B}$ can be bounded by
\beq{}
{
\frac{1}{\pa{m+1}^{2}}\cdot 2^{-m\cdot D\pa{s||p}}\le P_{B}\le m\cdot 2^{-m\cdot D\pa{s||p}}.
}
Next, similarly to the proof of Theorem \ref{th:LowerBoundExpAverageGuessworkAnyHashFunction} equation \eqref{eq:UpperBoundOnTheProbabilityOfAnyHashFunction} we can further bound the probability of success after $k$ unsuccessful guesses by
\beq{}
{
P_{B}\le P_{B}^{\pa{k}}\le \frac{P_{B}\cdot 2^{n}+m\cdot 2^{m\cdot H\pa{s}}}{2^{n}-k}
}
where $P_{B}^{\pa{k}}$ is the probability of guessing a password that is mapped to any $b\in B$, after $k$ unsuccessful guesses; and $m\cdot 2^{m\cdot H\pa{s}}$ takes into consideration the fact that the backdoor mechanism from Definition \ref{def:BackdoorForAnyHash} may add a mapping to each bin. Similarly to \eqref{eq:LowerBoundByGeometricalDistAvrtageGusswork} we can lower bound the guesswork by
\beq{}
{
E\pa{\Gof}\ge P_{B}\cdot\sum_{i=0}^{2^{m_{0}}-1}\pa{i+1}\cdot\pa{1-P_{B}^{\pa{2^{m_{0}}}}}^{i}.
}
where $m_{0}=\pa{1+\eps_{1}}\cdot m\cdot \log\pa{1/p}$, $0<\eps_{1}<\eps$. Since
\beq{}
{
P_{B}^{\pa{2^{m_{0}}}}\le P_{B}\frac{1+\pa{m+1}^{2}\cdot m\cdot 2^{-\eps\cdot\log\pa{1/p}\cdot m}}{1-2^{-m\cdot\log\pa{1/p}\cdot\pa{\eps-\eps_{1}}}}.
}
Hence, for the exact same arguments as the ones in Theorem \ref{th:LowerBoundExpAverageGuessworkAnyHashFunction} we can state that
\beq{eq:LowerBoundAverageGuessworkFinalGUessingAPasswordofAnyUser}
{
\lim_{m\to\infty}\frac{1}{m}\log\pa{E\pa{G_{Any}\pa{b\in B}}}\ge \lim_{m\to\infty}\frac{1}{m}\log\pa{1/P_{B}}=D\pa{s||p}.
}

The equality is achieved based on similar arguments to the ones provided in the proof of Theorem \ref{th:LowerBoundExpAverageGuessworkAnyHashFunction}; the number of mappings that other users remove in the backdoor mechanism is negligible. Furthermore, these results also apply to the case when averaging over the $P_{H}$-set of hash functions and considering any strategy of guessing passwords one by one. It results from the same arguments that are given in Theorem \ref{th:LowerBoundExpAverageGuessworkAnyHashFunction}; each user in $B$ can add no more than one mapping, and the total number of mappings  does not affect the average guesswork, as shown in equation \eqref{eq:LowerBoundAverageGuessworkFinalGUessingAPasswordofAnyUser}.

\end{proof}

\section{Proofs of the Results of Sections \ref{sec:GeneralExpressions} and \ref{sec:AverageGuessworkPhHashFunctionsNoMod} }\label{sec:BoundsAverageGuessworkNoBinsAllocation}
Before deriving upper and lower bounds for the conditional average guesswork let us formally define these measures.

\begin{defn}[The maximum and minimum conditional average guesswork]\label{def:MaxMinAverageGiesswork}
Given that the set of occupied bins is of size $|B|$ and consists of \textbf{unique} elements (i.e., the bins in the set are different), the minimum and maximum conditional average guesswork, conditioned on the set of occupied bins, are defined as follows:
\bitm{
\item The maximum conditional average guesswork is the maximum average number of guesses required to break into the system under a remote or an insider attack, where the maximum is taken over all sets of occupied bins of size $|B|$.
\item The minimum conditional average guesswork is the minimum number of guesses required to break into the system under a remote or an insider attack, where the minimum is taken over all sets of occupied bins of size $|B|$.  
}
\end{defn}

\begin{lem}[The maximum and minimum conditional average guesswork]
The maximum and minimum conditional average guesswork under a remote (insider) attack when the size of the set of occupied bins is $A$ and it consists of unique elements can be written as follows:
\beq{eq:MaxMinAveageGuessworkOnlineOffline}{
\max_{\pac{B:|B|=A}}E\pa{G_{remote (insider)}\pa{B}},
}
and 
\beq{eq:MaxMinAveageGuessworkOnlineOffline1}{
\min_{\pac{B:|B|=A}}E\pa{G_{remote (insider)}\pa{B}},
}
where averaging is done according to Definition \ref{def:AveragingArgument}.
\end{lem}
\begin{proof}
The proof is straight forward based on Definition \ref{def:MaxMinAverageGiesswork}.
\end{proof}

\begin{rem}
In order to achieve the maximum/minimum average guesswork the users have to draw passwords that are mapped to the set of occupied bins that maximizes/minimizes equation \eqref{eq:MaxMinAveageGuessworkOnlineOffline}/\eqref{eq:MaxMinAveageGuessworkOnlineOffline1}. The probability for this event can be very low. 
\end{rem}

\begin{rem}[The set of occupied bins the maximizes/minimizes that conditional average guesswork]
When averaging over a keyed $P_{H}$-hash function, for every element in this set of hash functions, users draw passwords that are mapped to the set of occupied bins that maximizes/minimizes equation \eqref{eq:MaxMinAveageGuessworkOnlineOffline}/\eqref{eq:MaxMinAveageGuessworkOnlineOffline1}. 
\end{rem}

In order to upper bound the rate at which the average guesswork increases we consider the case where all users are mapped to the $2^{H\pa{s}\cdot m-1}$  most likely bins. This in turn enables us to give an upper bound to the average guesswork by considering the results for the case of bins allocation. The lower bound is derived by assuming that the users are mapped to the $2^{H\pa{s}\cdot m-1}$ most likely bins.

We begin by proving Theorem \ref{th:LowerUpperConcentrationAnyHash}.
\begin{proof}[proof of Theorem \ref{th:LowerUpperConcentrationAnyHash}]
This result follows directly from equation \eqref{eq:LowerBoundAverageGuessworkPerBinFOrPHSET} which  lower bounds the probability when bins are allocated, by using the probability when there is no bins allocation.
\end{proof}

Next, we prove Theorem \ref{cor:LowerUpperTheAverageGuessworkForGuessingAnyPassword}
\begin{proof}[proof of Theorem  \ref{cor:LowerUpperTheAverageGuessworkForGuessingAnyPassword}]
The upper bound follows directly from Corollary \ref{cor:TheAverageGuessworkForGuessingAnyPassword} as the average guesswork with bins allocation is equal to the result when considering the most likely bins.

The lower bound is based on the assumption that the users are mapped to the $2^{H\pa{s}\cdot m-1}$ most likely bins. This leads to the following optimization problem
\beq{}
{
\min_{0\le q\le 1-s}D\pa{q||p}
}
where $1/2\le s\le 1$. The minimal value is achieved at $q=p\le 1/2$, and is equal to $0$. Therefore, the lower bound is equal to
\beq{}
{
\barr{cc}{D\pa{1-s||p} &0\le 1-s\le p\\ 0 &p\le 1-s\le 1/2 }.
}
Finally, because \emph{there is no} backdoor procedure that modify the hash function, the results hold for $n\ge \pa{1+\eps}\cdot m\cdot\log{1/p}$. This concludes the proof.
\end{proof}

We now prove Corollary \ref{cor:NoBinAllocationMostLikelyAverageGuesswork}.
\begin{proof}[proof of Corollary  \ref{cor:NoBinAllocationMostLikelyAverageGuesswork}]
We begin by proving that when $0\le 1-s< q\le p$ the probability that all passwords are mapped to bins of type $q$ is $e^{-2^{m\cdot\pa{2\cdot H\pa{1-s}-H\pa{q}}}}\times 2^{-m\cdot D\pa{q||p}\cdot 2^{H\pa{1-s}\cdot m}}$. Since each password is drawn Bernoulli$\pa{1/2}$ and $P_{H}$ defined based on the method of types as presented in Definition \ref{def:DefinitionOfIidCase}, the probability that a password is mapped to a certain bin of type $q$ is $2^{-m\cdot\pa{H\pa{q}+D\pa{q||p}}}$. Since there are $2^{m\cdot H\pa{s}}$ users, the probability that all users are mapped to a certain set of bins in which all bins are of type $q$ is $2^{-m\cdot 2^{m\cdot H\pa{1-s}}\cdot\pa{H\pa{q}+D\pa{q||p}}}$. The number of combinations in which the users are mapped to different bins of type $q$ is $\frac{2^{m\cdot H\pa{q}}!}{\pa{2^{m\cdot H\pa{q}}-2^{m\cdot H\pa{1-s}}}!}$. According to Sterling's approximation, when $m\gg 1$
\beq{}
{
2^{m\cdot H\pa{q}}!\sim \frac{2^{m\cdot H\pa{q}\cdot 2^{m\cdot H\pa{q}}}}{e^{2^{m\cdot H\pa{q}}}}
}
whereas for $m\gg 1$
\beq{}
{
\pa{2^{m\cdot H\pa{q}}-2^{m\cdot H\pa{1-s}}}!\sim 2^{m\cdot H\pa{q}\cdot \pa{2^{m\cdot H\pa{q}}-2^{m\cdot H\pa{1-s}}}}\times e^{2^{m\cdot\pa{2\cdot H\pa{1-s}-H\pa{q}}}-2^{m\cdot H\pa{q}}}.
}
The ratio of the two equations above increases like
\beq{}
{
2^{m\cdot H\pa{q}\cdot 2^{m\cdot H\pa{1-s}}}\times e^{-2^{m\cdot\pa{2\cdot H\pa{1-s}-H\pa{q}}}}.
}
Therefore, the probability that all passwords are mapped to different bins of type $q$ is equal to
\beq{}
{
2^{-m\cdot 2^{m\cdot H\pa{1-s}}\cdot\pa{H\pa{q}+D\pa{q||p}}}\times 2^{m\cdot H\pa{q}\cdot 2^{m\cdot H\pa{1-s}}}\times e^{-2^{m\cdot\pa{2\cdot H\pa{1-s}-H\pa{q}}}}=e^{-2^{m\cdot\pa{2\cdot H\pa{1-s}-H\pa{q}}}}\times 2^{-m\cdot D\pa{q||p}\cdot 2^{H\pa{1-s}\cdot m}}.
}
Furthermore the smallest rate at which the probability decreases is achieved when $q=p$ and is equal to
\beq{}
{
e^{-2^{m\cdot\pa{2\cdot H\pa{1-s}-H\pa{p}}}}.
}

Based on arguments similar to the ones in Theorem \ref{cor:LowerUpperTheAverageGuessworkForGuessingAnyPassword} we get that when all users are mapped to bins of type $q$ the probability of finding a password that is mapped to any of the bins is $2^{-m\pa{H\pa{q}+D\pa{q||p}-H\pa{1-s}}}$ and therefore the average guesswork increases like $2^{m\pa{H\pa{q}+D\pa{q||p}-H\pa{1-s}}}$. In the case when $q=p$ we get that the average guesswork increase like $2^{m\pa{H\pa{p}-H\pa{1-s}}}$

Finally, when $q<1-s\le 1/2$  the most probable event is when all passwords are mapped to bins of type $p$ in which case the average guesswork does not increase exponentially, that is, the rate at which it increases is equal to zero.
\end{proof}

Next, we now prove Corollary \ref{cor:NoBinAllocationMostLikelyAverageGuessworkOnLine}.
\begin{proof}[proof of Corollary  \ref{cor:NoBinAllocationMostLikelyAverageGuessworkOnLine}]
The probability that the users are mapped to a particular type $q$ follows along the same lines as the proof of Corollary \ref{cor:NoBinAllocationMostLikelyAverageGuesswork}. When all users are mapped to bins of type $q$, the probability of guessing a password that is mapped to a bin of a user is the same across users and is equal to $2^{-m\cdot \pa{H\pa{q}+D\pa{q||p}}}$, and therefore the average guesswork is equal to $2^{m\cdot \pa{H\pa{q}+D\pa{q||p}}}$.
\end{proof}

Finally, we prove Theorem \ref{th:LowerUpperBoundExpAverageGuessworkAnyHashFunction}.
\begin{proof}[proof of Theorem \ref{th:LowerUpperBoundExpAverageGuessworkAnyHashFunction}]
The upper bound follows directly from Theorem \ref{th:LowerBoundExpAverageGuessworkAnyHashFunction} as the average guesswork with bins allocation is equal to the result of averaging across the most likely bins. Because \emph{there is no} backdoor procedure that modify the hash function, the results hold for $n\ge \pa{1+\eps}\cdot m\cdot\log{1/p}$.

The lower bound is based on the assumption that the users are mapped to the $2^{H\pa{s}\cdot m-1}$ most likely bins. By arguments similar to Theorem \ref{th:TheAverageGuesseork} we arrive at the following optimization problem
\beq{}
{
\max_{0\le q\le 1-s}2\cdot H\pa{q}+D\pa{q||p}
}
where $1/2\le s\le 1$. In lemma \ref{lem:OptimizationProblemForMaximalArgument} it is shown that $2\cdot H\pa{q}+D\pa{q||p}$ is a concave function whose maximal value is at $q=1-p\ge 1/2$. Hence, when $0\le q\le 1-s\le 1/2$ the maximal value is achieved for $q=1-s$ and is equal to $2\cdot H\pa{1-s}+D\pa{1-s||p}=2\cdot H\pa{s}+D\pa{1-s||p}$. This concludes the proof.
\end{proof}

\begin{proof}[Proof of Theorem \ref{th:TheAverageGuessworkPerBinGeneralExpressions}]
The proof is along the same lines as Theorem \ref{th:LowerBoundExpAverageGuessworkAnyHashFunction} equations \eqref{eq:UpperBoundOnTheProbabilityOfAnyHashFunction}-\eqref{eq:LowerBoundAverageGuessworkPerBinFOrPHSET}.
\end{proof}

\begin{proof}[Proof of Theorem \ref{th:AverageGuessworkPerBinIidAssumption}]
The proof is along the same lines as Theorem \ref{th:LowerBoundExpAverageGuessworkAnyHashFunction} equations \eqref{eq:UpperBoundOnTheProbabilityOfAnyHashFunction}-\eqref{eq:LowerBoundAverageGuessworkPerBinFOrPHSET}.
\end{proof}

\begin{proof}[Proof of Corollary \ref{Cor:AverageGuessworkOverAllPasswordsNoBInsAllocation}]
When $n \ge\pa{1+\eps}\cdot m\cdot\log\pa{1/p}$ the proof follows along the same line as Corollary \ref{Cor:NoBinsAllocationAveragedOverAllPasswords}
\end{proof}

\section{Summary}
In this work we define and derive the average guesswork for hash functions. We define the conditional average guesswork as a function of the set of occupied bins. We also find the most likely conditional average guesswork as well as the maximum and minimum conditional average guesswork. We reveal an interesting connection between the effect of bias and the number of users on the level of security, and quantify the average guesswork of hash functions by using basic information theoretic  measures.  

Furthermore, we show that when the hash function can be modified by a backdoor mechanism, bias can increase the average guesswork unboundedly.



\IEEEpeerreviewmaketitle

\bibliographystyle{IEEEtran}
\bibliography{IEEEabrv,Generic_YairRef}

\end{document}